\documentclass[a4paper,12pt,twoside,reqno]{amsart}

\usepackage[T1]{fontenc}
\usepackage[utf8]{inputenc}
\usepackage[english]{babel}

\usepackage[osf]{mathpazo}
\usepackage{eucal}
\usepackage{mathrsfs}
\usepackage{microtype}

\usepackage{amsmath,amssymb,amsthm}
\usepackage{mathtools}
\usepackage{esint}
\usepackage{upgreek}

\usepackage{graphicx}
\usepackage{array}
\usepackage{enumitem}
\usepackage{xcolor}
\usepackage{setspace}
\usepackage{tikz}
\usepackage{tikz-cd}
\usepackage{hyperref}
\usepackage{xfrac}
\allowdisplaybreaks

\makeatletter
\@ifundefined{CMcal}{\DeclareMathAlphabet{\CMcal}{OMS}{cmsy}{m}{n}}{}
\makeatother

\theoremstyle{plain}
\newtheorem{theorem}{Theorem}[section]

\newtheorem{proposition}[theorem]{Proposition}
\newtheorem{lemma}[theorem]{Lemma}

\newtheorem{openproblem}[theorem]{Open Problem}

\theoremstyle{definition}
\newtheorem{definition}[theorem]{Definition}

\theoremstyle{remark}
\newtheorem{remark}[theorem]{Remark}

\numberwithin{equation}{section}
\numberwithin{figure}{section}
\numberwithin{table}{section}

\newcommand{\R}{\mathbb{R}}
\newcommand{\N}{\mathbb{N}}
\newcommand{\C}{\mathbb{C}}

\newcommand{\Z}{\mathbb{Z}}

\newcommand{\unit}{\boldsymbol{1}}

\newcommand{\s}[1]{\CMcal{#1}}

\newcommand{\bb}[1]{\mathscr{#1}}
\newcommand{\rr}[1]{\mathfrak{#1}}
\newcommand{\n}[1]{\mathbb{#1}}

\newcommand{\ketbra}[2]{|#1\rangle\langle#2|}
\newcommand{\expo}[1]{\,\mathrm{e}^{#1}\,}
\newcommand{\dd}{\,\mathrm{d}}
\newcommand{\ii}{\mathrm{i}}

\DeclarePairedDelimiter{\norm}{\lVert}{\rVert}

\newcommand{\virg}[1]{\lq\lq#1\rq\rq}
\newcommand{\ie}{\textsl{i.\,e.\,}}
\newcommand{\eg}{\textsl{e.\,g.\,}}
\newcommand{\cf}{\textsl{cf}.\,}

\newcommand{\revtxtmk}[1]{{\color[rgb]{0.85,0,0}#1}}

\DeclareMathOperator{\Tr}{Tr}

\DeclareMathOperator{\Id}{Id}

\DeclareMathOperator{\spec}{spec}
\DeclareMathOperator{\Aut}{Aut}

\newlength{\dhatheight}

\newcommand{\HH}{\rr{h}}

\newcommand{\bk}{\mathbf{k}}
\newcommand{\br}{\mathbf{x}}
\newcommand{\bz}{\mathbf{z}}
\newcommand{\bl}{\mathbf{l}}
\newcommand{\bR}{\mathbf{R}}

\newcommand{\bx}{\mathbf{x}}

\newcommand{\ip}[2]{\langle #1, #2\rangle}
\newcommand{\ap}{\mathfrak{a}^{+}}
\newcommand{\am}{\mathfrak{a}^{-}}
\newcommand{\bp}{\mathfrak{b}^{+}}
\newcommand{\bmm}{\mathfrak{b}^{-}}

\newcommand{\lB}{\ell_{B}}
\newcommand{\eB}{\epsilon_{B}}

\hypersetup{
  colorlinks=true,
  linkcolor=purple,
  citecolor=teal,
  urlcolor=blue,
  filecolor=magenta,
  breaklinks=true,
  bookmarksopen=true,
  pdftitle={Asymptotic quasi-local structure induced by magnetic field},
  pdfauthor={Giuseppe De Nittis},
}

\begin{document}

\title[Asymptotic quasi-local   structure]{Asymptotic quasi--local   structure\\ induced by  {m}agnetic field}

\author[G. De~Nittis]{Giuseppe De Nittis}

\address[G. De~Nittis]{Facultad de Matem\'aticas \& Instituto de F\'isica,
  Pontificia Universidad Cat\'olica de Chile, Santiago, Chile.}
\email{gidenittis@uc.cl}
\date{\today}


\begin{abstract}
We introduce a magnetic coherent--state frame over the Landau levels of the
two--dimensional Landau Hamiltonian, indexed by a discrete metric space
combining spatial and energy degrees of freedom. Although overcomplete, the
frame is almost--orthogonal, and we use it to endow the CAR $C^*$--algebra over the
one--particle Hilbert space with an \emph{asymptotic quasi--local} (AQL)
structure, in which graded commutators of observables localized on disjoint
index sets decay exponentially rather than vanish exactly. By means of a
family of conditional expectations onto the resulting local subalgebras, we
construct a Fr\'echet $*$--subalgebra of almost--local observables, invariant
under the discrete group of lattice translations, which act on it by
continuous automorphisms. For the full continuous group of magnetic
translations we establish a quantitative control on how local observables
spread through this almost--local algebra, and we show that this control {exhibits the compatibility of the AQL structure with graded asymptotic abelianness under the entire continuous translation group. The dependence on the volume of the localization regions of the constants in the AQL bound is left as an open problem.} The results provide a
discrete--continuum dictionary intended as the starting point for extending
Lieb--Robinson--type techniques from lattice systems to interacting
continuum electrons in a magnetic field.
\end{abstract}

\maketitle

\tableofcontents

\section{Introduction}\label{intro}

The analysis of interacting quantum many-body systems is a fundamental problem in mathematical physics, with particular challenges arising in the description of their non-equilibrium dynamics and infinite-volume limits. For quantum lattice systems over $\mathbb{Z}^\nu$, a central ingredient in the rigorous analysis of these questions is provided by the \emph{Lieb–Robinson bounds} \cite{LiebRobinson1972,Robinson1976,NachtergaeleOgataSims2006}. These estimates quantify the quasi-locality of the dynamics by controlling the propagation of initially localized observables and, in particular, by establishing an effective finite velocity for the spreading of information. This control is crucial for passing from finite-volume dynamics to the thermodynamic limit, where it yields a strongly continuous one-parameter group of $*$-automorphisms describing the Heisenberg dynamics on the quasi-local algebra of observables.

\smallskip

By contrast, considerably less is known for interacting quantum fields and multi-fermion systems in the continuum $\mathbb{R}^\nu$, a setting that is of comparable physical relevance. Beyond free-field dynamics, described by Bogoliubov automorphisms \cite{Lundberg1976,Araki1971}, and classical non-relativistic field theories \cite{Streater1968,StreaterWilde1970}, the rigorous analysis of interacting continuum systems is complicated by severe ultraviolet difficulties. Significant recent progress in this direction \cite{GebertNachtergaeleReschkeSims2020,HinrichsLemmSiebert2023} circumvents these difficulties by introducing suitable smearings of the creation and annihilation operators, thereby allowing for the construction of well--defined local interaction Hamiltonians. The resulting framework provides finite--volume Lieb–-Robinson bounds with constants uniform in the volume, which in turn enables the construction of the corresponding thermodynamic--limit dynamics.

\smallskip

In the presence of a uniform perpendicular magnetic field $B$, which is the physical setting of the quantum Hall effect (QHE), a natural analogue of the lattice-based approach emerges. The magnetic field induces cyclotron motion and localizes the electronic states at the characteristic scale set by the magnetic length $\ell_B$. In \cite{BachmannDeNittis2024}, this observation is made precise by introducing an overcomplete family of magnetic coherent states $\{\chi_\gamma\}_{\gamma\in\Gamma}$, centered at the points of a discrete set $\Gamma$. The resulting family forms a lattice-localized frame \cite{Bargmann1971,Perelomov1971,BoonZak1978,Daubechies1988}, providing a discrete framework for encoding the underlying continuum degrees of freedom.
Since the frame vectors are not mutually orthogonal, the associated creation and annihilation operators $a(\chi_\gamma)$ and $a^*(\chi_\gamma)$ do not satisfy the canonical anticommutation relations (CAR) of standard lattice fermions. Nevertheless, the corresponding frame-based second quantization retains sufficient locality structure to allow the transfer of suitable lattice techniques to the continuum setting. 

\smallskip

Although frame-based discretizations of continuum systems have previously been employed in single-particle quantum mechanics (see \eg \cite{Cornean2019,Cornean2024}), their systematic use in the context of second quantization was proposed and initiated in \cite{BachmannDeNittis2024}.
Building on this framework, the primary objective of the present work is to develop a systematic and rigorous analysis of the local algebraic structures induced by the magnetic field on the CAR algebra. In particular, we provide a precise operator-theoretic characterization of these local structures, an aspect that was largely left implicit or only partially addressed in the original framework of \cite{BachmannDeNittis2024}. This analysis provides the rigorous algebraic foundation needed for the treatment of local observables and derivations, as well as for the derivation of dynamical estimates in interacting magnetic continuum systems.

\smallskip

As a starting point let us introduce some preliminary concepts and definitions. Let $\rr{A}\equiv\rr{A}(\HH)$ be the  \emph{canonical anti-commutation relations (CAR) $C^*$-algebra} over the separable Hilbert space $\HH$, as described by  \cite[Theorem 5.2.5]{Bratteli-Robinson-2}.
Let $\sigma\in\Aut(\rr{A})$ be the \emph{parity automorphism}, namely the {quasi--free} automorphism   that changes the sign of both the creation and annihilation generators of the CAR algebra (see Section \ref{sub:car}).
It is an involution in the sense that $\sigma^2={\rm Id}$. 
An element $A\in\rr{A}$ is called \emph{even} if $\sigma(A)=A$ and \emph{odd} if $\sigma(A)=-A$. Each element $A\in\rr{A}$ has a unique decomposition into odd and even parts,
defined by
\begin{equation}\label{eq:proj:pm}
A\;=\;A^{+}+A^{-}\;,\qquad
A^{\pm}\;:=\;\frac{A\pm\sigma(A)}{2}\;.
\end{equation}
The even elements  form a
$C^{*}$--subalgebra $\rr{A}^{+}\subset\rr{A}$, and the odd elements
 form a Banach subspace $\rr{A}^{-}\subset\rr{A}$. 
An element $A\in\rr{A}^{\pm}$ of given parity is called \emph{homogeneous}. If $A$ is even its \emph{degree} is $d_A=0$. For $A$ odd the degree is $d_A=1$. Let $A,B\in\rr{A}^{\pm}$  {homogeneous} elements. The \emph{graded commutator} is defined by
\[
\qquad[A,B]_{\rm gr}\;:=\;AB-(-1)^{d_Ad_B}BA\;.
\]
When at least
one of $A,B$ is even then $[A,B]_{\rm gr}=AB-BA=:[A,B]$ reduces to an ordinary commutator.
When both $A,B$ are odd then $[A,B]_{\rm gr}=AB+BA=:\{A,B\}$ is the anticommutator.

\smallskip

Let  $(\Gamma,d)$ be  a discrete and countable metric space. 
The set  of \emph{finite subsets} of $\Gamma$ will be denoted  $\s{P}_f(\Gamma)$. If $\Lambda\in\s{P}_f(\Gamma)$, then its \emph{volume} (or cardinality) is  $|\Lambda|<\infty$.
 The set $\s{P}_f(\Gamma)$
  is directed, where the direction is by inclusion.  Let $\Lambda_1,\Lambda_2\subseteq\Gamma$ be nonempty and not necessarily finite. The
distance between them is defined by
\begin{equation}\label{eq:dist_set}
d(\Lambda_1,\Lambda_2)\;:=\;\inf\left\{d(\gamma_1,\gamma_2)\in\R\;|\;\gamma_1\in\Lambda_1\,,\gamma_2\in \Lambda_2\right\}\;.
\end{equation}

\smallskip

The notion of a \emph{quasi--local 
$C^*$--algebra} is recalled in \cite[Definition~2.6.3]{Bratteli-Robinson-1}. One of its defining features is the \emph{exact} vanishing of the graded commutators of elements with disjoint supports. The structure introduced below weakens this requirement, demanding only an \emph{asymptotic} vanishing.
\begin{definition}[Almost--quasi--local  (AQL) structure]\label{def:AQL}
The CAR algebra $\rr{A}$ has an AQL structure if there is a  discrete and countable metric space
$(\Gamma,d)$, {equipped with a group $\s T$ of isometries (called \emph{translations}),} and a net of $C^*$-subalgebras $\{\rr{A}_\Lambda\}_{\Lambda\in \s{P}_f(\Gamma)}\subset \rr{A}$ such that:
\begin{itemize}
\item[(1)] the algebras $\rr{A}_\Lambda$ have a common identity $\unit$;
\item[(2)] if $\Lambda_1\subseteq\Lambda_2$ then $\rr{A}_{\Lambda_1}\subseteq\rr{A}_{\Lambda_2}$;
\item[(3)] the \emph{local} subalgebra
\[
\rr{A}_{\mathrm{loc}}\;:=\;\bigcup_{\Lambda\in\s{P}_f(\Gamma)}\rr{A}_\Lambda\;
\]
is dense, \ie $\rr{A}=\overline{\rr{A}_{\mathrm{loc}}}$ where the closure is in the norm topology;
\item[(4)]  for every $\lambda>0$ there is a function $C_\lambda(\Lambda_1,\Lambda_2)>0$ of the pairs $\Lambda_1,\Lambda_2\in\s P_f(\Gamma)$, invariant under the action of $\s T$ on each argument separately, that is 
\[
C_\lambda(\tau\Lambda_1,\Lambda_2)\;=\;C_\lambda(\Lambda_1,\tau\Lambda_2)\;=\;C_\lambda(\Lambda_1,\Lambda_2)\;,\qquad \forall\;\tau\in\s T\;,
\]
 such that
 for all  {homogeneous} $A\in\rr A_{\Lambda_1}$, $B\in\rr A_{\Lambda_2}$,  one has that
\begin{equation}\label{eq:int-ineq_gr}
\bigl\|[A,B]_{\rm gr}\bigr\|\;\leqslant\;C_\lambda(\Lambda_1,\Lambda_2)\,\norm{A}\,\norm{B}\,
\expo{-\lambda\,d(\Lambda_1,\Lambda_2)}\;.
\end{equation}
\end{itemize}
\end{definition}

\smallskip

When only properties (1), (2), and (3) are satisfied, we refer to the resulting structure as a \emph{pre--quasi--local structure}.
It is worth emphasizing that the bound in point~(4) implies that, for \emph{any} rate $\lambda>0$, the graded commutators of elements localized in increasingly distant sets decay {exponentially fast} with their separation. Consequently, while the strict quasi–locality condition of \cite[Definition~2.6.3]{Bratteli-Robinson-1} does not hold in the present setting, the graded commutators become small at large distances {faster than any exponential}. In this precise sense, the present structure relaxes strict quasi–locality to a \emph{fast}, \emph{asymptotic} condition of quasi-–locality. 
A relevant aspect of Definition \ref{def:AQL} is the invariance requirement on $C_\lambda$. Without it, point (4) would be empty, because $\|[A,B]_{\rm gr}\|\leqslant2\|A\|\|B\|$ and one could always take $C_\lambda(\Lambda_1,\Lambda_2)=2\expo{\lambda d(\Lambda_1,\Lambda_2)}$. With it, (4) controls the graded commutators uniformly along the orbits of $\s T$, which is what is needed to prove asymptotic abelianness for translated observables. In lattice systems $C_\lambda$ typically depends only on the volumes $|\Lambda_j|$.
In concrete situations, it is reasonable to expect this quantity to be an (eventually) increasing function of the volumes $|\Lambda_j|$ as well as of the rate parameter $\lambda$. The better the control of $C_\lambda(\Lambda_1,\Lambda_2)$, the sharper the resulting estimates, particularly for thermodynamic purposes. Ideally, one would seek a polynomial dependence on the volumes rather than exponential, or faster growth. In concrete settings, however, obtaining such a sharp volume dependence can be technically challenging and may require a substantial amount of additional analysis.
A final relevant observation is that the definition could, in principle, be generalized to metric spaces $\Gamma$ that are not necessarily countable, with the finiteness of sets expressed in terms of a suitable underlying measure. However, although such a generalization is possible, it is precisely the discrete and countable structure of $\Gamma$ that is crucial for applying techniques such as Lieb–-Robinson bounds to the study of the interacting dynamics.

\smallskip

There is a second relevant structure that we need to introduce. We say that $\R^\nu$ acts by
\emph{translations} on $\rr A$ if there is a map $\R^\nu\ni\bz\mapsto\alpha_\bz\in\Aut(\rr A)$
such that $\bz\mapsto\alpha_\bz(A)$ is norm--continuous for every $A\in\rr A$, and
\[
\alpha_\bz\circ\alpha_{\bz'}\;=\;\vartheta_{\varpi(\bz,\bz')}\circ\alpha_{\bz+\bz'}\,,\qquad
\forall\,\bz,\bz'\in\R^\nu\,.
\]
The latter identity says that the group structure of $\R^\nu$ is lifted to the group
$\ast$--automorphisms $\Aut(\rr A)$ only \emph{projectively}, that is, up to a twist by a \emph{cocycle}
$\varpi$ acting through the \emph{gauge} automorphisms $\vartheta_\theta\in\Aut(\rr A)$ (see
Remark~\ref{rem:cocycle} for details). We then call $(\rr A,\R^\nu,\alpha,\varpi)$ a
\emph{twisted $C^{*}$--dynamical system}.
\begin{definition}[Graded asymptotic abelianness]\label{def:AA}
The twisted $C^{*}$--dynamical system $(\rr A,\R^\nu,\alpha,\varpi)$ is \emph{graded
asymptotically abelian} if, for all homogeneous $A,B\in\rr A$,
\[
\lim_{|\bz|\to\infty}\bigl\|[\alpha_\bz(A),B]_{\rm gr}\bigr\|\;=\;0\;.
\]
\end{definition}

\smallskip

The two structures of Definitions~\ref{def:AQL} and~\ref{def:AA} are not automatically
compatible: whether they are depends on how the translation action $\bz\mapsto\alpha_\bz$ of
$\R^\nu$ interacts with the metric space $(\Gamma,d)$ and with the family $\s P_f(\Gamma)$ of
local regions. When they \emph{are} compatible, however, the system acquires a rich 
structure, in which two distinct notions of locality are intertwined: the \emph{intrinsic}
localization in the degrees of freedom labelled by $\Gamma$ (not necessarily spatial
positions) encoded by the net of local algebras, and the \emph{dynamical} requirement that
\virg{spatially} distant observables be asymptotically decorrelated.

\smallskip

For spin systems or fermionic systems on a lattice, the two structures are perfectly compatible. Indeed, in these cases, $\Gamma\simeq \mathbb{Z}^\nu$ represents the positional degrees of freedom of the system, which transform naturally under the discrete translation group. 
In the continuum, however, the situation is more subtle. In principle, one can define a quasi-local structure by identifying $\Gamma \simeq \mathbb{R}^\nu$, again interpreting its elements solely as positional degrees of freedom, and exploiting the natural action of $\mathbb{R}^\nu$ on itself by translations \cite[Examples 5.2.7 \& 5.2.21]{Bratteli-Robinson-2}. The price to pay, however, is the loss of countability of $\Gamma$, which makes many of the techniques available for discrete systems no longer directly applicable.

\smallskip

The main objective of this work is to show that, for two-dimensional continuum systems in the presence of a uniform magnetic field $B$ perpendicular to the plane, the structures introduced in Definitions~\ref{def:AQL} and~\ref{def:AA} arise naturally and are mutually compatible.
To describe our results, it is useful to introduce some notation. 
First of all the magnetic field fixes an intrinsic  scale via
the \emph{magnetic length} $\lB>0$.
The one-particle Hilbert space is taken to be $\HH:=L^2(\R^2)$
and the relevant CAR algebra is the $C^*$-algebra
$\rr{A}:=\rr{A}(L^2(\R^2))$ defined over this Hilbert space (a CAR algebra with continuous degrees of freedom).
The Hilbert space $\HH$ decomposes as an orthogonal direct sum of infinite-dimensional subspaces $\HH_r$, the so-called \emph{Landau levels}, labelled by the \emph{energy level} quantum number $r\in\N_0:=\{0\}\cup\N$.
Each Landau level $\HH_r$ contains a relevant family of normalized vectors $\{\chi_{\bl,r}\}_{\bl\in\s{L}}$ called \emph{coherent states}. 
The \emph{spatial} variable $\bl\in\R^2$ describes the center of localization of the wavefunction $\chi_{\bl,r}$ and 
the set of labels 
\begin{equation}\label{eq:L}
\s{L}\;\equiv\;\s{L}(\ell_1,\ell_2)\;:=\;\ell_1\Z\times \ell_2\Z\;\subset\;\R^2
\end{equation}
is a two-dimensional lattice, isomorphic to $\Z^2$.
The choice of lattice parameters
is not completely arbitrary, as they are required to satisfy the condition
\begin{equation}\label{eq:threshold}
f_B\;:=\;\frac{\ell_1\ell_2}{\lB^2}\;<\;2\pi\;.
\end{equation}
In fact \eqref{eq:threshold}
 guarantees that the family $\{\chi_{\bl,r}\}_{\bl\in\s{L}}$ is a 
 \emph{frame} for $\HH_r$, that is, an
overcomplete non-orthogonal system of vectors. 
It is worth noting that the dimensionless quantity $f_B$
measures the fraction of \emph{flux quanta} per  unit cell of
$\s{L}$. By combining the spatial and energy degrees of freedom, we can define the countable set
\begin{equation}\label{eq:Xi}
\Gamma\;:=\;\s{L}\times\N_0\;=\;\bigl\{\gamma\equiv(\bl,r)\;\big|\;\bl\in \s{L}\,,\,\,  r\in\N_0\ 
\bigr\}\;
\end{equation}
which turns out to be a discrete metric space with the distance defined by
\begin{equation}\label{eq:Xi-22}
d\left(\gamma,\gamma'\right)\;:=\;\tfrac{1}{\lB}|\bl-\bl'|_1+|r-r'|\;
\end{equation}
where $|\bl|_1:=\ell_1 |n_1|+ \ell_2|n_2|$ for 
$\bl\equiv(\ell_1 n_1, \ell_2 n_2)\in\s{L}$ denotes  the \emph{graph norm}.
With this notation at hand, and assuming condition \eqref{eq:threshold}, it follows that $\{\chi_\gamma\}_{\gamma\in\Gamma}$ is an overcomplete frame for the full space $\HH$ (Theorem \ref{thm:frame}).
The overcompleteness of the frame prevents the vectors from being orthogonal, and strict orthogonality is recovered only in an asymptotic sense. As stated in Proposition \ref{prop:magloc}, the frame is \emph{almost--orthogonal} in the sense that
for every $\lambda>0$, there is a $G_\lambda\geqslant1$, such that
\[\bigl|\ip{\chi_{{\gamma}}}{\chi_{{\gamma}'}}\bigr|
\;\leqslant\;G_\lambda\,\expo{-\lambda\,d(\gamma,\gamma')}\;,\qquad\;\gamma,\gamma'\in\Gamma\,.
\]
It is worth observing that the frame shows an exponential decay at \emph{any} rate $\lambda>0$ (indeed a {decay \emph{faster than any exponential}}), at the cost of increasing the constant $G_\lambda$.

\smallskip

The frame $\{\chi_\gamma\}_{\gamma\in\Gamma}$ of $\HH$ described above is fundamental for endowing the CAR algebra $\rr{A}$ with the appropriate quasi--local structure. This can be done as follows. For each 
$\Lambda\in\s{P}_f(\Gamma)$, set
\begin{equation}\label{eq:H_LAM}
\HH_\Lambda\;:=\;\mathrm{span}\{\chi_\gamma\;|\;\gamma\in\Lambda\}\;\subseteq\; \HH\;.
\end{equation}
This is a subspace of finite dimension  
$\dim(\HH_\Lambda)=|\Lambda|$ (Lemma \ref{lem:findim}).
To each $\HH_\Lambda$ consider the associated finite-dimensional CAR algebra $\rr{A}_\Lambda\equiv\rr{A}(\HH_\Lambda)\simeq{\rm Mat}_{2^{|\Lambda|}}(\C)$. From  $\Lambda_1\subseteq\Lambda_2\subset\Gamma$ it follows that $\HH_{\Lambda_1}\subseteq \HH_{\Lambda_2}\subset\HH$, and in turn $\rr{A}_{\Lambda_1}\subseteq
\rr{A}_{\Lambda_2}\subset\rr{A}$. Moreover, all these algebras share a common unit $\unit$.
Therefore, the net of $C^{*}$--algebras $\{\rr{A}_\Lambda\}_{\Lambda\in\s{P}_f(\Gamma)}$  appears to be well suited for endowing $\rr{A}$ with a quasi-local structure. In fact, also the density property (3) of Definition \ref{def:AQL} is satisfied (Proposition \ref{prop:quasilocal-frame}). 
What fails for strict quasi-locality is the commutativity of observables supported on disjoint sets. Indeed, for disjoint sets $\Lambda_1\cap\Lambda_2=\emptyset$ it is \emph{not} true  that the subspaces $\HH_{\Lambda_1}$ and $\HH_{\Lambda_2}$ are mutually orthogonal, as a consequence of the overcompleteness of the frame. Consequently, if $A\in \rr{A}_{\Lambda_1}$ and $B\in \rr{A}_{\Lambda_2}$ are homogeneous elements, in general one obtains that
${[A,B]_{\rm gr}}\neq 0$ in view of the non orthogonality of the generating spaces   $\HH_{\Lambda_1}$ and $\HH_{\Lambda_2}$.
However, as a consequence of the almost orthogonality of the frame, $\rr{A}$ turns out to be asymptotic quasi--local in the sense of Definition \ref{def:AQL}.
\begin{theorem}[Magnetic AQL structure]\label{theo:main1}
Let $\rr{A}$ be the CAR algebra over $\HH\equiv L^2(\R^2)$. The net of $C^{*}$--algebras $\{\rr{A}_\Lambda\}_{\Lambda\in\s{P}_f(\Gamma)}$ 
induced by the magnetic frame of coherent states $\{\chi_\gamma\}_{\gamma\in\Gamma}$ endows $\rr{A}$ with an AQL structure.
\end{theorem}

\smallskip

The proof of Theorem \ref{theo:main1} is obtained by combining Propositions \ref{prop:quasilocal-frame} and \ref{lem:almost-comm-gr-II}. The former establishes the pre-quasi-local structure, which is the simpler part of the argument, while the latter provides the key estimate needed to complete the proof, namely property (4) in Definition \ref{def:AQL}. It follows by combining a bound of the graded commutators in terms of the angle between the subspaces $\HH_{\Lambda_1},\HH_{\Lambda_2}$ (Proposition \ref{lem:almost-comm-gr-delta}) with the almost-orthogonality of the frame. The resulting prefactor $C_\lambda(\Lambda_1,\Lambda_2)$ is polynomial in the volumes $|\Lambda_j|$, but it also involves the smallest eigenvalues of the Gram matrices of the frame vectors in $\Lambda_j$. As a consequence of the overcompleteness of the frame, these eigenvalues are not bounded away from zero when the regions grow, so that the dependence of $C_\lambda$ on the volume is not controlled. Whether a polynomial dependence can be achieved is stated as an open problem in Section~\ref{sec:bet:est} (Open Problem~\ref{lem:delta-uniform}).
There is one aspect that deserves special consideration. Theorem \ref{theo:main1} shows that the magnetic frame equips $\rr{A}$ with an appropriate local structure. However, the localization is not restricted to the spatial degrees of freedom. Indeed, $\Gamma$ accounts for both the spatial and the (magnetic) energy degrees of freedom. Thus, the localization provided by Theorem \ref{theo:main1} should be understood as a sort of \emph{phase-space localization} rather than as a purely spatial one.

\smallskip

At this point, a structural aspect must be carefully considered. The local structure encoded by
the finite subsets of $\Gamma$, mixing spatial and energy degrees of freedom, is not fully
compatible with the action of the translation group $\R^2$. More precisely, $\R^2$ acts on
$\HH$ through the magnetic translations $t_{\bz}$, with $\bz\in\R^2$,\revtxtmk{ }which preserve the Landau levels
$\HH_r$ (see Section~\ref{sub:coh}). Via the associated Bogoliubov $\ast$-automorphisms $\alpha_\bz$ these induce an
action of $\R^2$ on $\rr A$. Because the magnetic translations compose only projectively, up to
the usual \emph{magnetic cocycle} $\sigma_B(\bz,\bz'):=\sfrac{1}{2\lB^2}(\bz\times\bz')$,
$\R^2$ is represented in $\Aut(\rr A)$ by a projective
representation, and $(\rr A,\R^2,\alpha, {\varpi})${, with $\varpi=-\sigma_B$ (see Remark~\ref{rem:cocycle}),} is a twisted $C^{*}$--dynamical system (see Section \ref{sec:trasl_aut}).
A generic translation, however, is \emph{not} compatible with the local structure carried by
the finite subsets of $\Gamma$. The \emph{lattice} translations $\alpha_\bl$ with $\bl\in\s L$, merely
permute the elements of $\s P_f(\Gamma)$ via the translations $\Lambda\mapsto(\bl,0)+\Lambda$, and in turn preserve
$\rr A_{\mathrm{loc}}$. However, for a generic $\bz\in\R^2$ the translated region $(\bz,0)+\Lambda$
no longer belongs to $\s P_f(\Gamma)$. Thus $\rr A_{\mathrm{loc}}$ is invariant only under the
discrete subgroup $\{\alpha_\bl\}_{\bl\in\s L}$, not under the full continuous group.
On the other hand, the continuous nature of the physical system singles out $\R^2$, rather than
the lattice $\s L$, as the fundamental symmetry group. Reconciling the two structures therefore
calls for an enlargement of $\rr A_{\mathrm{loc}}$ in which the loss of strict locality caused by
a generic translation can be quantitatively controlled. The tool that makes this possible is the
conditional expectation $\n E_\Lambda$ (see Section~\ref{sec:cond-exp}), the trace--preserving,
norm--one projection onto the local algebra $\rr A_\Lambda$. The conditional expectation allows us to introduce the family of norms
\begin{equation}\label{eq:norm_p_0}
\norm{A}_{p}\;:=\;\norm{A}
+\sup_{k\in\N_0}\bigl\|A-\n{E}_{\Lambda_k}(A)\bigr\|\,(1+k)^{p}\;,\qquad
A\in\rr{A}\;,
\end{equation}
where
$
\Lambda_k:=\{\gamma\in\Gamma\;|\;{d(\gamma,0)}\leqslant k\}$
is the ball of radius $k$ centered at the origin $0\in\Gamma$,  and 
defined for every $p\in\N_0$.
These norms 
 weigh how fast an
observable is approximated by its local restrictions. The completion of $\rr A_{\mathrm{loc}}$ with respect to the system $\{\norm{\cdot}_{p}\}_{p\in\N_0}$ is a Fr\'echet $\ast$--algebra $\rr A_\infty\subset \rr A$, the \emph{almost--local} algebra.
The relevance of $\rr A_\infty$ is that the  {lattice} translations restrict to it. 
For every
$\bl\in\s{L}$, the restricted \emph{lattice} translation $\alpha_\bl$ (denoted by the same symbol, with a slight abuse of notation) is a topological $\ast$--automorphism of $\rr A_\infty$
 continuous for the Fr\'echet topology (Lemma~\ref{lemm:aut-fr-lattice}).
 Combined with the AQL structure established in Theorem \ref{theo:main1}, this immediately implies the discrete version of asymptotic abelianness (Proposition \ref{cor:asympt-ab-full}).
  The passage to continuum requires controlling how local elements spread throughout $\Gamma$ under continuous translations.
  The crucial result, established in Lemma~\ref{lemma_trasl_loc}, states that
  $\alpha_\bz:\rr A_{\mathrm{loc}}\to\rr A_\infty$
   for every $\bz\in\R^2$.
  This result, together with a control on the growth of the norms of local elements under compactly bounded translations (Lemma \ref{cor:mono-improve-full}) and standard density arguments (Proposition \ref{prop:cont-asympt-ab-local}), allows us to establish our second main result:
  
\begin{theorem}[Asymptotic abelianness]\label{thm:asympt-ab}
Let $\rr{A}$ be the CAR algebra over $\HH\equiv L^2(\R^2)$ and $\alpha_\bz$ the $\ast$--automorphisms induced by the magnetic translations.
Then, the twisted $C^{*}$--dynamical system $(\rr A,\R^2,\alpha,\varpi)$, with $\varpi=-\sigma_B$, is graded asymptotically abelian and compatible with the  AQL structure.
\end{theorem}

\smallskip

The compatibility with the AQL structure is the content of Lemmas \ref{lemma_trasl_loc} and 
\ref{cor:mono-improve-full}.
It is worth emphasizing that asymptotic abelianness of $\rr A$ under
magnetic translations can already be inferred, at the level of the
generators, from the decay of the one--particle overlaps. Writing
$t_\bz\phi(\br)=\expo{\ii\,\bk(\bz)\cdot\br}\phi(\br-\bz)$ for the magnetic translation, 
one finds $|\ip{\psi}{t_\bz\phi}|\leqslant(|\psi|*|\tilde\phi|)(\bz)\to0$, with $\tilde\phi(\bx):=\phi(-\bx)$, being a convolution of two $L^2$ functions (the magnetic phase plays no role). This forces $[\alpha_\bz(a^\#(\phi)),a^\#(\psi)]_{\rm gr}\to0$ on the
generators of $\rr A$, hence, by a finite graded Leibniz expansion, on every polynomial in the generators, and, by density, on all of $\rr A$. Theorem
\ref{thm:asympt-ab} is therefore elementary by itself. The proof given in this paper, through
Proposition \ref{prop:cont-asympt-ab-local}, has a different purpose. It shows that the AQL structure and the translation action are compatible, in the sense that every $\alpha_\bz$ maps $\rr A_{\rm loc}$ into the almost--local algebra $\rr A_\infty$, with seminorm bounds locally uniform in $\bz$ (Lemmas~\ref{lemma_trasl_loc} and~\ref{cor:mono-improve-full}), a piece of information that the elementary argument does not provide. It is worth noting,
however, that this control falls short of upgrading $\alpha_\bz$ itself to
a continuous automorphism of $\rr A_\infty$ for a generic $\bz\in\R^2$.
Whether this stronger statement holds remains open (see Open
Problem~\ref{open:aut-A-infty}).

\smallskip

To summarize, the magnetic coherent-state frame provides a natural way to \virg{discretize the continuum}. The degrees of freedom of $L^2(\R^2)$ are faithfully encoded by the discrete metric space $(\Gamma,d)$, at the price of a localization that is simultaneously \emph{spatial and energetic}, rather than purely spatial. Strict lattice locality is therefore lost, but it is recovered in an \emph{asymptotic} form. The resulting net of local algebras $\rr A_\Lambda$ generates an almost--local Fr\'echet $\ast$-algebra $\rr A_\infty$ on which continuous magnetic translations act in a controlled manner, giving rise to a twisted $C^{*}$--dynamical system that is graded asymptotically abelian.
This provides a discrete–to-continuum bridge through which techniques developed for quantum lattice systems can be adapted to continuous many--body systems in a magnetic field. The present work constitutes the first step of a broader program. As a next stage, the framework can be used to formulate the free magnetic dynamics and revisit the integer QHE in a setting naturally suited to many--body extensions. The construction can then be extended to genuinely interacting systems, where the Lieb–-Robinson bounds established in \cite{BachmannDeNittis2024} provide the necessary control of information propagation. Ultimately, our aim is to develop this discrete–to-continuum framework towards the study of the \emph{fractional QHE}.
 
 \smallskip
 \noindent
{\bf Structure of the paper.}
 {\bf Section~\ref{sec:frames}} reviews the one--particle Landau operator and introduces the magnetic
translations and the magnetic coherent states. After recalling the notion of a frame, we show
that the coherent states constitute an almost--orthogonal frame for $L^2(\R^2)$, thereby encoding the continuum in the discrete index set $\Gamma$.
{\bf Section~\ref{sec:car}}
 passes to the many--body setting. We build the {\rm CAR} algebra $\rr A$
over $L^2(\R^2)$  (introducing Bogoliubov automorphisms and quasi--free states) and,     by means of the
magnetic frame over $(\Gamma,d)$,
equip it with the AQL structure of
Definition~\ref{def:AQL}, thereby proving Theorem~\ref{theo:main1}.
In
{\bf Section~\ref{Sect:frech_struct}}
we introduce conditional expectations onto the local algebras, which provide a suitable notion of localization despite the lack of tensor-product factorization induced by the deformed anticommutation relations. These conditional expectations lead to a family of localizing norms and to a natural Fréchet subalgebra in which the  asymptotic  abelianness claimed in Theorem \ref{thm:asympt-ab}  can be established.{ Two questions left open by the present methods, namely the dependence of the constants $C_\lambda$ on the volume of the regions and the continuity of $\alpha_\bz$ on $\rr A_\infty$, are formulated as Open Problems~\ref{lem:delta-uniform} and~\ref{open:aut-A-infty}.} Finally, {\bf Appendix~\ref{app_prrof}} contains all the technical results, together with their proofs, so as to keep the main text more readable.

 \smallskip
 
 \noindent
{\bf Acknowledgements.}
The author gratefully acknowledges the support provided by the grant \emph{Fondecyt Regular - 1230032}. The author would also like to thank Sven Bachmann, Rahul Hingorani, Bruno Nachtergaele, Stefan Teufel for numerous stimulating and insightful discussions related to this work.

 \smallskip
 
 \noindent
{\bf Conflict of interest statement.}
 The author certifies that they have no affiliations with or involvement in any organization or entity with any financial interest (such as honoraria; educational grants; participation in speakers' bureaus; membership, employment, consultancies, stock ownership, or other equity interest; and expert testimony or patent-licensing arrangements), or non--financial interest (such as personal or professional relationships, affiliations, knowledge or beliefs) in the subject matter or materials discussed in this manuscript.

\smallskip

\noindent
{\bf Data availability statement.}
No datasets were generated or analyzed during the current study. The scripts used to perform the symbolic and numerical verifications supporting the results of this work are available from the author upon reasonable request.

\smallskip
\noindent
{\bf Use of artificial intelligence statement.}
During the preparation of this work, the author used Claude \emph{Opus 5} (Anthropic) to improve the linguistic quality, grammatical accuracy, and readability of the manuscript. The tool was also employed to cross-check the algebraic consistency and overall correctness of the calculations presented in the various theorems, as well as to support and streamline the bibliographic search process. All AI-generated outputs were subsequently thoroughly reviewed, independently verified, and edited by the author, who takes full responsibility for the mathematical accuracy, academic integrity, and final scientific content of the published article.

\section{One-particle magnetic dynamics}
\label{sec:frames}

In this section we review  the one--particle theory of the Landau operator and set up the
tools used throughout the paper. After recalling the spectral decomposition of the Landau
Hamiltonian into its Landau levels, we introduce the magnetic translations and the associated
system of magnetic coherent states. We then recall the notion of a \emph{frame} and show that the
coherent states form an \emph{almost--orthogonal} frame for $L^2(\R^2)$. In this way  the continuum degrees of freedom
are faithfully encoded by the discrete index set given by the metric space  $(\Gamma,d)$.
The material presented in this section is essentially standard, with the exception of a few technical results concerning decay estimates and overlap properties of the frame elements. The proofs of all the results stated in this section are deferred to Appendix~\ref{ap:proof_S1}.
\subsection{The Landau Hamiltonian}\label{sub:landau}
We work on the one-particle Hilbert space $\HH:=L^2(\R^2)$ endowed with the usual scalar product
\[
\ip{\phi}{\psi}\;:=\;\int_{\R^2}\dd^2\br\;\overline{\phi(\br)}\;
\psi(\br)\;.
\] 
Throughout,
$\br\equiv(x_1,x_2)\in\R^2$ and $\mathbf a\times\mathbf b:=a_1b_2-a_2b_1$ is the scalar planar
cross product for given $\mathbf a,\mathbf b\in\R^2$. 
The algebra of bounded operators on $\HH$ will be denoted with $\bb{B}(\HH)$ and the group of unitary operators with $\bb{U}(\HH)$.
We will systematically use Dirac notation for rank-one operators. More precisely, for $\psi_1,\psi_2\in\HH$ we denote by
$\ketbra{\psi_1}{\psi_2}$ the operator defined by
$\ketbra{\psi_1}{\psi_2}:\phi\mapsto \ip{\psi_2}{\phi}\psi_1$ for every $\phi\in\HH$. 
The presence of a perpendicular homogeneous magnetic field fixes two scales: the \emph{magnetic length} $\lB>0$ and the  \emph{quantum of energy} $\eB>0$. 
The limit of a \emph{strong} magnetic field corresponds to
 $\lB\to0$ and $\eB\to\infty$. This represents a \emph{singular limit} for the whole theoretical framework developed in this work.

\smallskip

The dimensionless \emph{kinetic momenta} are defined as
\[
\mathrm K_1\;:=\;-\ii\lB\,\frac{\partial}{\partial x_1}-\frac{x_2}{2\lB}\;,\qquad
\mathrm K_2\;:=\;-\ii\lB\,\frac{\partial}{\partial x_2}+\frac{x_1}{2\lB}\;,
\]
and the companion  \emph{dual momenta} are
\[
\mathrm G_1\;:=\;-\ii\lB\,\frac{\partial}{\partial x_2}-\frac{x_1}{2\lB}\;,\qquad
\mathrm G_2\;:=\;-\ii\lB\,\frac{\partial}{\partial x_1}+\frac{x_2}{2\lB}\;.
\]
These four operators are essentially self--adjoint on the space of Schwartz functions ${S}(\R^2)$ and satisfy the commutation relations
\[
[\mathrm K_1,\mathrm K_2]\;=\;-\ii\,\unit\;=\;[\mathrm G_1,\mathrm G_2]\;,
\qquad [\mathrm K_i,\mathrm G_j]\;=\;0\;.
\] 
The \emph{Landau Hamiltonian} is defined as
\[h_{B}\;:=\;\frac{\eB}{2}(\mathrm K_1^2+\mathrm K_2^2)\;.
\]
Introducing the two commuting
harmonic pairs
\[
\mathfrak a^{\pm}\;:=\;\frac{1}{\sqrt2}\bigl(\mathrm K_1\pm \ii\,\mathrm K_2\bigr)\;,
\qquad
\mathfrak b^{\pm}\;:=\;\frac{-1}{\sqrt2}\bigl(\mathrm G_1\pm \ii\,\mathrm G_2\bigr)\;,
\]
one has that
\[
[\am,\ap]\;=\;\unit\;=\;[\bmm,\bp]\;,\qquad [\mathfrak a^\pm,\mathfrak b^\pm]\;=\;0,
\]
and
\[
 h_{B}\;=\;\eB\bigl(\ap\am+\tfrac12\bigr)\;.
\]
As a consequence
\[\spec( h_{B})\;=\;\left.\left\{\epsilon_r:=\eB(r+\tfrac12)\;\right|\;r\in\N_0\right\}\;,
\]
each level being infinitely degenerate. 

\smallskip

With the normalized fundamental eigenstate
\[
\psi_{0,0}(\br)\;:=\;\frac{1}{\sqrt{2\pi}\,\lB}\expo{-\frac{|\br|^2}{4\lB^2}}\;,
\]  
characterized by
$\am\psi_{0,0}=\bmm\psi_{0,0}=0$, the vectors
\[
\psi_{r,m}\;:=\;\frac{1}{\sqrt{r!\,m!}}(\ap)^{r}(\bp)^{m}\psi_{0,0}\;,\qquad r,m\in\N_0
\] 
form an orthonormal basis of $\HH$ which diagonalizes the Landau Hamiltonian, namely $ h_{B}\psi_{r,m}=\epsilon_r\psi_{r,m}$. 
Let $\Pi_{r,m}$ be the one dimensional projection along $\psi_{r,m}$ defined by $\Pi_{r,m}\phi:=\ip{\psi_{r,m}}{\phi}\psi_{r,m}$. The 
 \emph{$r$--th Landau projection} is $\Pi_r:=\sum_{m\in\N_0}\Pi_{r,m}$, and it projects on \emph{$r$--th Landau level} $\HH_r:=\Pi_r\HH$. One has   that $\HH=\bigoplus_{r\in\N_0}\HH_r$.

\subsection{Magnetic translations and coherent states}\label{sub:coh}

For $\bz\in\R^2$ the magnetic translation $t_{\bz}$ is the unitary
\[
t_{\bz}\;:=\;\expo{-\tfrac{\ii}{\lB^2}\,({\bz}\times\bR)},\qquad
\bR\;:=\;(-\lB\,\mathrm G_1,\ \lB\,\mathrm G_2)\;,
\]
where
 $\bR$ is known as \emph{guiding {center} vector}. The unitary $t_{\bz}$
acts as 
\[
(t_{\bz}\phi)(\br)\;=\;\expo{\frac{\ii}{2\lB^2}(\br\times\bz)}\,
\phi(\br-\bz)\;,\qquad \phi\in\HH\;.
\] 
Written through the ladder operators one has,
\begin{equation}\label{eq:Tb}
t_{\bz}\;=\;\expo{\left(\frac{z_1+\ii z_2}{\sqrt2\lB}\right)\,\bp
-\left(\frac{z_1-\ii z_2}{\sqrt2\lB}\right)\,\bmm}\;=\;
\expo{-\frac{|\bz|^2}{4\lB^2}}\,\expo{\left(\frac{z_1+\ii z_2}{\sqrt2\lB}\right)\,\bp}\,\expo{-\left(\frac{z_1-\ii z_2}{\sqrt2\lB}\right)\,\bmm}\;,
\end{equation}
where the second equality is a consequence of the Baker–Campbell–Hausdorff identity, along with the fact that the $\psi_{r,m}$ are a set of analytic vectors for $\bp$ and $\bmm$. It follows that $[t_{\bz}, h_{B}]=0$ and $t_{\bz}\HH_r=\HH_r$ for every $\bz\in\R^2$ and $r\in\N_0$. The
family $\{t_{\bz}\}_{\bz\in\R^2}$ obeys the Weyl relations
\begin{equation}\label{eq:weyl}
t_{\bz}t_{\bz'}\;=\;\expo{-\frac{\ii}{2\lB^2}(\bz\times\bz')}
t_{\bz+\bz'}\;,\quad
t_{\bz}^{*}\;=\;t_{-\bz}\;,
\end{equation}
which imply
\[t_{\bz}t_{\bz'}t_{\bz}^{*}t_{\bz'}^{*}\;=\;\expo{-\frac{\ii}{\lB^2}(\bz\times\bz')}\,\unit\;.
\]
\begin{lemma}\label{Lem:01}
It holds true that
\[
\ip{\psi_{r,0}}{t_{\bz}\psi_{r,0}}\;=\;\expo{-\frac{|\bz|^2}{4\lB^2}}
\]
for every $\bz\in\R^2$ and $r\in\N_0$.
\end{lemma}

\smallskip

The vector 
\[
\chi_{\bz,r}\;:=\;t_{\bz}\,\psi_{r,0}\;=\;\frac{1}{\sqrt{r!}}(\ap)^{r}\,t_{\bz}\,\psi_{0,0}
\]
will be called the \emph{(magnetic) coherent state} of center $\bz$ and energy $\epsilon_r$. The next two results will provide the main properties of the vectors  $\chi_{\bz,r}$. The first one concerns the 
spatial localization of the coherent states.

\begin{lemma}[Ring localization of coherent states]\label{lemm:gau_loc}
For every $\bz\in\R^2$ and every $r\in\N_0$ the magnetic coherent state $\chi_{\bz,r}$ obeys
\begin{equation}\label{eq:ring-bound}
\bigl|\chi_{\bz,r}(\br)\bigr|\;\leqslant\;\frac{1}{\sqrt{2\pi}\,\lB\,(1+r)^{\frac14}}\,\expo{
-\frac{\bigl(|\br-\bz|-\sqrt{2r}\,\lB\bigr)^{2}}{4\lB^{2}}}\,,\qquad\br\in\R^2\,.
\end{equation}
The bound is \emph{sharp} at
$r=0$.
\end{lemma}
 
 \smallskip

    It is worth emphasizing that in \eqref{eq:ring-bound} both the
    prefactor $\sfrac{1}{\sqrt{2\pi}\lB}$ and the \emph{decay rate} 
$\sfrac{1}{4\lB^{2}}$    
   are uniform in the energy level $r\in\N_0$.
   The peak height decays only like
$(1+r)^{-\sfrac{1}{4}}$, the mass concentrates on the cyclotron ring of radius $\sqrt{2r}\lB$ centered at $\bz$, and the Gaussian tail away from the ring decays at the fixed rate given above.

\begin{remark}[Limitations of the center--Gaussian estimate]\label{rem:peak-height}
Instead of \eqref{eq:ring-bound}, it is possible to bound $\chi_{\bz,r}$   by a Gaussian \emph{centered} at $\bz$,
\[
\bigl|\chi_{\bz,r}(\br)\bigr|\;\leqslant\;\frac{C_{\sigma,r}}{r^{\frac{1}{4}}}\,\expo{-\sigma\,\frac{|
\br-\bz|^{2}}{4\lB^{2}}}\,,\qquad\sigma\in[0,1)\,.
\]
This bound is perfectly  valid, but has the cost of a constant that grows exponentially in the level,
\[
C_{\sigma,r}\;\underset{r\to\infty}{\sim}\;\frac{1}{(2\pi)^{3/4}\lB}\,\expo{\frac r2\log\frac1{1
-\sigma}}\,.
\]
Moreover  no such bound survives at $\sigma=1$. This explosion is a \emph{geometric artifact}, not a
feature of $\chi_{\bz,r}$. The mass of the level--$r$ state lies on the ring $\rho_{\max}=\sqrt{
2r}\,\lB$, whereas the estimate centers its Gaussian at $\bz$. Reaching the ring against a fixed
rate $\sigma$ from the center costs $\expo{\sfrac{\sigma\rho_{\max}^{2}}{4\lB^{2}}}=\expo{\sfrac{\sigma r}{2}}$,
exactly the growth of $C_{\sigma,r}$. No $r$--independent constant can absorb it. The ring bound
\eqref{eq:ring-bound} removes the artifact by placing the decay where the state actually
concentrates.
\hfill$\blacktriangleleft$
\end{remark}

\smallskip

The second result follows from the size of the overlaps 
between coherent states which is 
 diagonal in the energy variable and Gaussian in the centers.
 \begin{lemma}[Gaussian {overlaps}]\label{lem:overlap}
 It holds true that
\[\ip{\chi_{\bz,r}}{\chi_{\bz',r'}}
\;=\;\delta_{r,r'}\;
\expo{\frac{\ii}{2\lB^2}(\bz\times\bz')}\;
\expo{-\frac{|\bz-\bz'|^2}{4\lB^2}}\;
\]
for all $\bz,\bz'\in\R^2$ and $r,r'\in\N_0$
\end{lemma}

\subsection{Almost-orthogonal frames}\label{sub:frames}
We briefly review the notion of a frame and some of its fundamental properties. For a comprehensive treatment, we refer the reader to \cite{Christensen-frames}.

\smallskip

Let us recall that a countable family $\{\xi_\gamma\}_{\gamma\in \Gamma}$ in a separable Hilbert space $\HH$ is a
\emph{frame} if there exist constants $0<A\leqslant B<\infty$ (the frame bounds) with
\[A\,\norm{\psi}^2\;\le\;\sum_{\gamma\in \Gamma}\bigl|\ip{\xi_\gamma}{\psi}\bigr|^2\;\le\;B\,\norm{\psi}^2\;,
\qquad \forall\,\psi\in\HH\;.
\]
The associated \emph{frame operator}
\begin{equation}\label{S_eq}
S\psi\;:=\;\sum_{\gamma\in \Gamma}\ip{\xi_\gamma}{\psi}\xi_\gamma\;,
\qquad \forall\,\psi\in\HH
\end{equation}
is bounded, self--adjoint, positive and
invertible, with $A\,\unit\le S\le B\,\unit$. Every $\psi\in\HH$ is recovered from the
\emph{canonical dual} $\{S^{-1}\xi_\gamma\}_{\gamma\in \Gamma}$ through
the reconstruction formula
\begin{equation}\label{eq:rec_can}
\psi\;=\;\sum_{\gamma\in \Gamma}\ip{S^{-1}\xi_\gamma}{\psi}\,\xi_\gamma\;=\;\sum_{\gamma\in \Gamma}s_\gamma(\psi)\,\xi_\gamma\;.
\end{equation}
The numbers $s_\gamma(\psi)$ are called \emph{frame coefficients} and
$\{s_\gamma(\psi)\}_{\gamma\in \Gamma}\in\ell^2(\Gamma)$. The frame is a Riesz basis if and only if  it is
$\ell^2$--independent. Otherwise it is \emph{overcomplete}.

\smallskip

Let $\s{P}_f(\Gamma)$ be the \emph{set  of finite subsets} of $\Gamma$. 
For 
 $\Lambda\in\s{P}_f(\Gamma)$  define the \emph{truncated} reconstruction operator
\[I_\Lambda\;:=\;\sum_{\gamma\in\Lambda}\ketbra{\xi_\gamma}{S^{-1}\xi_\gamma}\;\in\;\bb{B}(\HH).
\]
The next result will play a relevant role in the following sections.
\begin{lemma}\label{lemma:reconstr_op}
When $\Lambda\nearrow\Gamma$ both $I_\Lambda\to\unit$ and $I_\Lambda^*\to\unit$ in the strong operator topology. Moreover
\[
\left\|I_\Lambda\right\|\;=\left\|I_\Lambda^*\right\|\;\leqslant\;\sqrt{\frac{B}{A}}
\]
where $0<A\leqslant B<\infty$ are the frame bounds.
\end{lemma}

\begin{remark}[Increasing exhaustive sequences]
By $\lim_{\Lambda\nearrow\Gamma}$ we mean that the limit must exist along any
\emph{increasing}, \emph{exhaustive} sequence
$\{\Lambda_n\;|\;n\in\N\}\subset\s{P}_f(\Gamma)$. Increasing means
$\Lambda_n\subset\Lambda_{n+1}$, and exhaustive means that for every $\gamma\in\Gamma$
there is $n\in\N_0$ with $\gamma\in\Lambda_n$.
\hfill$\blacktriangleleft$
\end{remark}

\smallskip

Let  $(\Gamma,d)$ be  a discrete and countable metric space endowed with a distance $d:\Gamma\times\Gamma\to[0,\infty]$.
We will assume the \emph{uniform summability} property
\begin{equation}\label{eq:unif-summ-cos}
m_\epsilon\;:=\;\sup_{\gamma\in\Gamma}\left(\sum_{\gamma'\in\Gamma}\expo{-\epsilon\,d(\gamma,\gamma')}\right)\;<\;\infty \;,
\qquad \forall\,\epsilon\;>\;0\;.
\end{equation}
It is worth noting that $m_\epsilon\to \infty$ when $\epsilon\to 0$.

\begin{definition}[Asymptotically--  and almost--orthogonal frame]\label{def:loc-frame}
Let $(\Gamma,d)$ be discrete, countable and with the uniform summability  property. A frame $\{\xi_\gamma\}_{\gamma\in \Gamma}$ of normalized vectors, \ie $\|\xi_\gamma\|=1$ for all $\gamma\in\Gamma$, is
\emph{asymptotically--orthogonal} if there are constants $G\geqslant1$ and $\lambda>0$ with
\begin{equation}\label{eq:loc}
\bigl|\ip{\xi_\gamma}{\xi_{\gamma'}}\bigr|\;\leqslant\;G\,\expo{-\lambda\,d(\gamma,\gamma')},
\qquad\forall\,\gamma,\gamma'\in \Gamma\;.
\end{equation}
The frame is called \emph{almost--orthogonal} if it is {asymptotically-orthogonal} with respect to every $\lambda>0$ with a pre--factor $G_\lambda\geqslant1$ which can depend on $\lambda$. 
\end{definition}

\smallskip

The condition $G\geqslant 1$ (resp. $G_\lambda\geqslant 1$) is required for compatibility with the normalization assumption on the frame. The adjective \emph{asymptotically–orthogonal} is justified by \eqref{eq:loc}: the family $\{\xi_\gamma\}_{\gamma\in\Gamma}$ is 
complete by definition of a frame, while its
off--diagonal overlaps vanish asymptotically in the separation of the indices. In the \emph{almost–orthogonal} case the vanishing
is {\emph{faster than any exponential}}. 
 Consequently, for $\lambda$ sufficiently large, an almost–orthogonal frame behaves increasingly like an orthogonal family, and in the limiting picture like an orthonormal Riesz basis, which is complete but not overcomplete. The overcompleteness of the frame is therefore encoded in the prefactor $G_\lambda$, which must increase with $\lambda$.

\smallskip

An important property of an almost--orthogonal frame is that the matrix elements of the frame operator $S$, as well as those of all its powers, decay exponentially. More precisely, the following result, established in \cite[Lemma 6.4]{BachmannDeNittis2024}, holds.
\begin{lemma}[Almost--diagonality of $S^p$]\label{alm_ort_Sp}
Let $\{\xi_\gamma\}_{\gamma\in\Gamma}$ be an asymptotically--orthogonal frame as in
Definition~\ref{def:loc-frame}, with rate $\lambda>0$, frame constant $G$, and frame
operator $S$. Then for every $p\in\N_0$,
\begin{equation}\label{eq:Sp-loc-00}
\bigl|\ip{\xi_\gamma}{S^{p}\xi_{\gamma'}}\bigr|\;\leqslant\;G_{p,\epsilon}\,\expo{-(\lambda-
\delta)\,d(\gamma,\gamma')}\,,\qquad\forall\,\gamma,\gamma'\in\Gamma\,,
\end{equation}
for every $0<\epsilon<\delta<\lambda$, where $G_{p,\epsilon}:=G^{p+1}m_\epsilon^{p}$.
If the frame is almost--orthogonal (\ie its overlaps decay at every rate
$\lambda>0$), then
\begin{equation}\label{eq:Sp-loc-000}
\bigl|\ip{\xi_\gamma}{S^{p}\xi_{\gamma'}}\bigr|\;\leqslant\;G_{p,\epsilon,\lambda}\,\expo{-
\lambda\,d(\gamma,\gamma')}\,,\qquad\forall\,\gamma,\gamma'\in\Gamma\,,
\end{equation}
for every $\lambda>0$, and $S^{p}$ is said to be \emph{almost--diagonal} with respect to the
frame.
\end{lemma}

\smallskip

In other words, the off--diagonal matrix elements of $S^{p}$ decay at a rate arbitrarily close
to that of the frame, the discrepancy being quantified by $\delta$. The price for a small
$\delta$ is a large multiplicative constant. As $\delta\to0$ one has $\epsilon\to0$, hence
$G_{p,\epsilon}\to\infty$.
The bound \eqref{eq:Sp-loc-00} is the content of \cite[Lemma~6.4]{BachmannDeNittis2024}. The
almost--diagonality \eqref{eq:Sp-loc-000} follows by applying \eqref{eq:Sp-loc-00} with frame
rate $\lambda+\delta$ in place of $\lambda$. This is legitimate precisely because the frame is
almost--orthogonal, hence localized at the rate $\lambda+\delta$ as well. The price is the rate--dependent constant
$G_{p,\epsilon,\lambda}:=(G_{\lambda+\delta})^{p+1}m_\epsilon^{p}$, where $G_{\lambda+\delta}$
is the frame constant of Definition~\ref{def:loc-frame} at rate $\lambda+\delta$.

\smallskip 
 
Perhaps the most important, and in some sense most surprising, fact is that the     asymptotic--diagonality  extends to the inverse powers of $S$. The almost--diagonality, however, is more delicate.
The central part of the following result was first established in \cite{Jaffard}, with a detailed proof provided in \cite[Proposition 2.2 \& Example 4]{BachmannDeNittis2024}. The statement below, however, gives a more detailed formulation and includes some additional information.

\smallskip

\begin{proposition}[Almost--diagonality of $S^{-p}$]\label{prop:inherited}
Let $\{\xi_\gamma\}_{\gamma\in\Gamma}$ be an asymptotically--orthogonal frame as in
Definition~\ref{def:loc-frame}, with rate $\lambda>0$, frame constant $G$, and frame
operator $S$. Then, for every $p\in\N$ there are $A_p>0$ and $0<\lambda_p<\sfrac{\lambda}{2}$ such that
\begin{equation}\label{eq:Sp-loc}
\bigl|\ip{\xi_\gamma}{S^{-p}\xi_{\gamma'}}\bigr|\;\leqslant\;A_p\,\expo{-\lambda_p\,d(\gamma,
\gamma')}\,,\qquad\forall\,\gamma,\gamma'\in\Gamma\,.
\end{equation}
If the frame is almost--orthogonal \textup{(}localized at every rate $\lambda>0$\textup{)} and
\begin{equation}\label{eq:growth-cond}
\ln (G_\lambda)\;=\;o(\lambda)\qquad(\lambda\to\infty)\,,
\end{equation}
then $S^{-p}$ is \emph{almost--diagonal} in the sense that for every $\lambda>0$ there is $A_{p,\lambda}>0$ with
\begin{equation}\label{eq:Sp-loc01}
\bigl|\ip{\xi_\gamma}{S^{-p}\xi_{\gamma'}}\bigr|\;\leqslant\;A_{p,\lambda}\,\expo{-\lambda\,d(\gamma,
\gamma')}\,,\qquad\forall\,\gamma,\gamma'\in\Gamma\,.
\end{equation}
\end{proposition}
 
\smallskip

The previous result should be contrasted with the case of the positive powers. For $S^{p}$
the asymptotic orthogonality of the frame implies the almost--diagonality of $S^{p}$
\emph{directly}. The degradation of the rate is additive, so describing the frame at a higher
rate recovers any target rate at the sole cost of a larger constant. For the inverse powers
$S^{-p}$ this is no longer the case. The asymptotic--orthogonality of the frame still yields
the fixed--rate localization \eqref{eq:Sp-loc}, but the almost--diagonality
\eqref{eq:Sp-loc01} does \emph{not} follow directly from it. It requires, in addition, the
growth condition \eqref{eq:growth-cond}. This condition is not innocent. As we shall see in
the next section, it \emph{fails} in our case of interest (see Remark \ref{rem:dual-magnetic}). It remains open
whether \eqref{eq:growth-cond} is an intrinsic obstruction to the almost--diagonality of
$S^{-p}$, or merely a limitation of the proof technique of
Proposition~\ref{prop:inherited}. In any case, what we can conclude is that, in the absence of condition \eqref{eq:growth-cond},
an almost--orthogonal frame $\{\xi_\gamma\}_{\gamma\in\Gamma}$ gives rise to a dual frame $\{S^{-1}\xi_\gamma\}_{\gamma\in\Gamma}$
that is merely \emph{asymptotically orthogonal}. Its overlaps decay exponentially at a single
fixed rate $\lambda_\ast <\sfrac{\lambda}{2}$ (the rate of $S^{-2}$), rather than at every rate like the original frame. 

\smallskip

 Interestingly, also the matrix elements of the fractional powers $S^{\pm\sfrac{1}{2}}$ are known to exhibit exponential decay, with a fixed maximal decay rate. Results of this type are relevant for establishing quantitative localization and decay estimates for the associated \emph{canonical tight frame} $\{\nu_\gamma\}_{\gamma\in\Gamma}$, defined by $\nu_\gamma:=S^{-\sfrac{1}{2}}\xi_\gamma$, see \cite[Theorem 5.3.4]{Christensen-frames}.

\subsection{The magnetic coherent--state frame}\label{sub:magframe}
Consider the countable set $\Gamma:=\s{L}\times\N_0$ defined by \eqref{eq:L} and \eqref{eq:Xi}. This is a discrete metric space with respect to the distance \eqref{eq:Xi-22}.
Introducing the pseudo-norm
\[
\|\gamma\|_1\;\equiv\;\|(\bl,r)\|\;:=\;\tfrac{1}{\lB}|\bl|_1+|r|\;,
\]
the distance can be expressed as
\[d\left(\gamma,\gamma'\right)\;=\;\|\gamma-\gamma'\|_1
\;=\;\tfrac{1}{\lB}|\bl-\bl'|_1+|r-r'|\;
\]
where one has to take into account that $\Gamma$ is only a semigroup, and hence the difference of two elements 
$\gamma\equiv(\bl,r)$ and $\gamma'\equiv(\bl',r')$
must be defined as $\gamma-\gamma'\equiv(\bl-\bl',|r-r'|)$.
It is worth noting that $d$ possesses a form of translation invariance, which
determines the geometric properties of $(\Gamma,d)$. This fact is at the basis of the next result.
\begin{lemma}\label{lem:summ-dim}
Let $(\Gamma,d)$ be   the metric space  defined by \eqref{eq:Xi} and \eqref{eq:Xi-22}. Then
\begin{equation}\label{eq:geo_gamma1}
m_\epsilon
\;=\;{\coth\!\Big(\tfrac{\epsilon}{2}\Big)}\;
\coth\!\Big(\tfrac{\epsilon\,\ell_1}{2\lB}\Big)\,
\coth\!\Big(\tfrac{\epsilon\,\ell_2}{2\lB}\Big)
\;<\;\infty\
\end{equation}
for every $\epsilon>0$. Moreover
\begin{equation}\label{eq:geo_gamma2}
\sup_{\gamma\in\Gamma}\bigl|\bigl\{\gamma'\in\Gamma\;\big|\;d(\gamma,\gamma')<R\bigr\}\bigr|
\;\leqslant\;\underbrace{3\bigl(\tfrac{2\lB}{\ell_1}+1\bigr)\bigl(\tfrac{2\lB}{\ell_2}+1\bigr)}_{=:\kappa_B}\,R^3
\end{equation}
for every $R\geqslant 1$.
\end{lemma}

\smallskip

Equation~\eqref{eq:geo_gamma1} shows that $(\Gamma,d)$ has the uniform summability
property.
It is worth observing that $m_\epsilon\sim(\sfrac{{8}}{f_B})\epsilon^{-3}$
as $\epsilon\to0$ where the dimensionless quantity
$f_B$ is defined in \eqref{eq:threshold}.
 Equation~\eqref{eq:geo_gamma2} shows that $(\Gamma, d)$ is moreover
\emph{(translation) homogeneous} of dimension $\nu=3$, meaning that the ball--counting bound being
uniform in the {center} $\gamma$ grows cubically with its radius $R$. The main role of  equation~\eqref{eq:geo_gamma2} is the following. Let $\Lambda\in\s{P}_f(\Gamma)$ be a finite subset and define its \emph{diameter} as
\[
{\rm diam}(\Lambda)\;:=\;\sup_{\gamma,\gamma'\in \Lambda}d(\gamma,\gamma')\;.
\]
Then, the volume of $\Lambda$ can be bounded by
\begin{equation}\label{eq:vol_boun}
|\Lambda|\;\leqslant\;\rr{v}_B(\Lambda)\;:=\;\kappa_B\left(1+{\rm diam}(\Lambda)\right)^3\;.
\end{equation}

\smallskip

The metric space $(\Gamma, d)$ is the  index set of lattice sites
and energy levels  (endowed with  its natural metric)  suitable for 
the labeling of the (magnetic) coherent states introduced in Section \ref{sub:coh}. 
From now on, we will use the short notation $\chi_\gamma\equiv\chi_{\bl, r}$ for every $\gamma\equiv(\bl,r)\in \Gamma$.
The completeness of the sampled coherent states  $\{\chi_\gamma\}_{\gamma\in\Gamma}$ is governed by the density of
$\s{L}$ relative to the magnetic length $\ell_B$. The threshold follows from a
classical criterion for lattices of Gabor type vectors \cite{Bargmann1971,Perelomov1971,Bacry}. The frame properties 
are related to results in
\cite{Daubechies,Janssen}.
The following result follows directly from the work in the aforementioned papers as presented in \cite[Section 5.3]{BachmannDeNittis2024} (see also the related proof in Appendix \ref{ap:proof_S1}). 
\begin{theorem}[Condition for a magnetic coherent–state frame]
\label{thm:frame}
The  family of magnetic coherent–states $\{\chi_\gamma\}_{\gamma\in\Gamma}$ is a frame if and only if \eqref{eq:threshold} holds.
In that case, it is overcomplete.
\end{theorem}

\begin{remark}[Flux quantum threshold]\label{rem:threshold}
Condition \eqref{eq:threshold} says that the fundamental cell 
of $\s{L}$
carries less than one
flux quantum.  At the threshold value (one
flux quantum per unit cell)
 the family is complete but not a frame (a von Neumann lattice). Above the threshold, the family is incomplete. 
 As the magnetic field increases ($\ell_B \to 0$) the lattice $\s{L}$ must become increasingly dense ($\ell_i \to 0$), in order to satisfy the threshold conditions.
We take \eqref{eq:threshold} to hold from now on. \hfill$\blacktriangleleft$
\end{remark}

\smallskip

Given a $\Lambda\in\s{P}_f(\Gamma)$, let $\HH_\Lambda$ be the subspace defined by \eqref{eq:H_LAM}.
By convention, we define $\HH_\emptyset:=\{0\}$ as the $0$-dimensional vector space.
A priori
one only knows that $\dim(\HH_\Lambda)\leqslant|\Lambda|$, since the coherent states
form an overcomplete family. It turns out, perhaps surprisingly, that this is not the case. In fact the
dimension is always maximal, $\dim(\HH_\Lambda)=|\Lambda|$ meaning  that the redundancy of the frame
is a genuinely infinite--dimensional phenomenon. Every finite subfamily is linearly
independent, and only infinite $\ell^2$--combinations can produce linear relations.

\begin{lemma}[Finite linear independence]\label{lem:findim}
For every $\Lambda\in\s{P}_f(\Gamma)$ one has that $\dim(\HH_\Lambda)=|\Lambda|$.
\end{lemma}

\smallskip

Let $S$ be the frame operator defined by the magnetic frame $\{\chi_\gamma\}_{\gamma\in\Gamma}$. It turns out that $S$, as well as all bounded Borel functions of $S$, are diagonal with respect to the decomposition into Landau levels. More precisely, for every $f\in L^\infty(\mathbb{R})$, let $f(S)$ be the bounded operator on $\HH$ obtained through the functional calculus. The collection of all these operators forms the commutative von Neumann algebra $\bb{M}(S)$ generated by $S$.
\begin{proposition}\label{pr:comm}
Let  $\Pi_r$ be the {$r$--th Landau projection}. Then
\[
[f,\Pi_r]\;=\;0\;,\qquad \forall\, f\in\bb{M}(S)\,
\]
for every $r\in\N_0$.
\end{proposition}

\smallskip

  An immediate consequence of the latter result is that also the Landau Hamiltonian $ h_{B}$ commutes with $\bb{M}(S)$. In particular it holds true that $[S^p, h_{B}]=0$ for every $p\in\Z$.
    
\smallskip

The following result describes perhaps the most important property of the magnetic frame $\{\chi_\gamma\}_{\gamma\in\Gamma}$. 
\begin{proposition}[Almost-orthogonality of the magnetic frame]\label{prop:magloc}
Under \eqref{eq:threshold},  $\{\chi_\gamma\}_{\gamma\in\Gamma}$ is an 
almost-orthogonal frame for $\HH$. More precisely,
for every $\lambda>0$ one has that
\begin{equation}\label{eq:magloc}
\bigl|\ip{\chi_{\gamma}}{\chi_{\gamma'}}\bigr|
\;\leqslant\;G_\lambda\,\expo{-\lambda\,d(\gamma,\gamma')}\;
\end{equation}
with $G_\lambda:=\expo{\,2\lambda^{2}}\geqslant1$.
\end{proposition}

\smallskip

It is worth observing that the exact behavior established in
Lemma~\ref{lem:overlap} (strict locality in the energy index and Gaussian decay
in the spatial index) is \emph{stronger} than mere almost-orthogonality.
Nevertheless, almost-orthogonality is flexible enough to yield general results. On the other hand, a remnant of the {decay faster than any exponential} is reflected in the fact that the decay rate $\lambda>0$ can be chosen arbitrarily large, at the cost of increasing (super-exponentially)
the constant $G_\lambda$.

\begin{remark}[Asymptotic--orthogonality of the dual magnetic frame]\label{rem:dual-magnetic}
Proposition \ref{prop:magloc} shows that the magnetic frame $\{\chi_\gamma\}_{\gamma\in\Gamma}$ is almost--orthogonal, in the sense that the overlap decays exponentially at any prescribed rate $\lambda>0$. However, the dependence of the constant $G_\lambda$ on $\lambda$ prevents this estimate from satisfying condition \eqref{eq:growth-cond}. More precisely, in the present case one has
$\ln( G_\lambda)=2\lambda^2$. As a consequence, it 
 can be expected   that the dual magnetic frame $\{S^{-1}\chi_\gamma\}_{\gamma\in
\Gamma}$ is only \emph{asymptotically--orthogonal}, at a single \emph{maximal} rate $\lambda_\ast$, rather than at every
rate. Its overlaps are governed by the second inverse power,
\[
\ip{S^{-1}\chi_\gamma}{S^{-1}\chi_{\gamma'}}\;=\;\ip{\chi_\gamma}{S^{-2}\chi_{\gamma'}}\,,
\]
so the relevant rate is $\lambda_\ast:=\lambda_2$, the output rate of \eqref{eq:Sp-loc} for
$p=2$. By the gap estimate of Proposition~\ref{prop:inherited}, $\lambda_\ast<\sfrac{\lambda}{2}$ at
every description rate $\lambda$, and since the magnetic frame is localized at every rate, the
best achievable value is the finite maximum
\[
\lambda_\ast\;=\;\max_{\lambda>0}\,\lambda_2(\lambda)\,.
\]
{The following argument is only \emph{formal} and provides an order--of--magnitude heuristics}. To estimate  $\lambda_\ast$ one has to look at the proof of  Proposition \ref{prop:magloc} and observe that the maximum is  
attained where the increasing branch $\theta\lesssim\lambda$ meets the decreasing branch
$\ln(\sfrac{1}{r_2})\,E_{2,\epsilon,\theta,\delta}\sim \ln(\sfrac{1}{r_2})/(8\lambda)$ (the latter
because $\ln d_{2,\epsilon}\sim 4\lambda^{2}$ for the Gaussian frame). This crossover
gives the order--of--magnitude estimate
\[
\lambda_\ast\;\asymp\;\sqrt{\frac{1}{8}\,\ln\left(\frac{1}{r_2}\right)}\,,\qquad r_2\;:=\;1-\bigl(\tfrac{s}{\norm{S}}\bigr)^{2}
\,.
\]
This is a fixed positive constant depending only on the 
spectral bounds $s,\norm{S}$ of the frame operator $S$. The magnetic dual frame is therefore asymptotically--orthogonal at this fixed rate $\lambda_\ast$, but not
almost--diagonal. \hfill$\blacktriangleleft$
\end{remark}

\smallskip

The dual frame inherits a form of exponential spatial  localization from the
ring localization of the coherent states described in Lemma \ref{lemm:gau_loc}.

\begin{lemma}[Cylindrical localization   of the dual frame]\label{lem:dual-localization}
The dual magnetic frame $\{S^{-1}\chi_\gamma\}_{\gamma\in\Gamma}$ is exponentially localized
outside   the cyclotron ring. More precisely, there is a rate $\lambda_
\ast>0$ \textup{(}the same as in Remark~\ref{rem:dual-magnetic}\textup{)} such that for every
$\gamma\equiv(\bl,r)\in\Gamma$,
\[
\bigl|(S^{-1}\chi_\gamma)(\br)\bigr|\;\leqslant\;\frac{C_\ast}{(1+r)^{\frac14}}\,\expo{-\frac{\lambda_\ast}{2}\left[\tfrac{|\br-\bl
|}{\lB}-\sqrt{2r}\right]_+}\,,\qquad\forall\,\br\in
\R^2\,,
\]
with $[x]_+:=\max\{x,0\}$, and $C_\ast$ depending only on $\lambda_
\ast$ and the geometry of $(\Gamma,d)$.
 \end{lemma}
 
\smallskip

Both the constant $C_\ast$ and the rate $\lambda_\ast$ are independent of the Landau level $r$. Specifically, the dual frame function $S^{-1}\chi_{(\bl,r)}$ is uniformly bounded within a cyclotron cylinder of radius $\sqrt{2r}\,\lB$ around $\bl$  and amplitude $C_\ast(1+r)^{-\sfrac14}$, while decaying exponentially outside it as the distance from the ring increases. This transition from Gaussian  decay of the frame to exponential decay of the dual frame is intrinsic to the dual structure. In fact, the spatial localization is governed by the decay rate of the inverse frame matrix rather than that of the underlying frame vectors. Finally, for any fixed $r\in\N_0$, Lemma~\ref{lem:dual-localization} provides a rigorous justification for condition~(ii) in \cite[Proposition 3.7]{BachmannDeNittis2024}.

\smallskip

The preceding result provides a quantitative estimate of the overlap between the elements of the magnetic frame and those of its canonical dual. We will use the following technical result in the subsequent analysis.

\begin{lemma}[Dual--coherent overlap]\label{lem:dual-coh-overlap}
For every $\bl\in\s L$, $\bz\in\R^2$ and $r,r'\in\N_0$, 
there is a rate $\lambda_
\ast>0$ \textup{(}the same as in Remark~\ref{rem:dual-magnetic}\textup{)} such that
\[
\bigl|\ip{S^{-1}\chi_{(\bl,r)}}{\chi_{(\bz,r')}}\bigr|\;\leqslant\;\delta_{r,r'}\;D_\ast\;\expo{-\frac
{\lambda_\ast}{2}\left[\tfrac{|\bl-\bz|}{\lB}-2\sqrt{2r}\right]_+}\,,
\]
with $[x]_+=\max\{x,0\}$, and $D_\ast>0$ depending only on $\lambda_
\ast$ and the geometry of $(\Gamma,d)$.
\end{lemma}
 
 \smallskip
 
     It is worth noting that 
 the amplitude in the result above does \emph{not} decay as $(1+r)^{-\sfrac{1}{4}}$. Although both the dual and the coherent
state have peak height $\sim(1+r)^{-\sfrac{1}{4}}$, their overlap integrates over the cyclotron ring, whose
area grows like $\sqrt{1+r}$, exactly restoring a level--uniform constant.

\section{CAR algebra and magnetic local structures}\label{sec:car}

The next main goal is to lift the one-body framework developed in Section \ref{sec:frames} to the many-body setting. This will be achieved by constructing the CAR algebra $\rr{A}$ over $\HH$ and endowing it with the quasi-local structure induced by the magnetic frame of coherent states over $(\Gamma,d)$. 
The main result of this section is to establish that $\rr{A}$ carries the AQL structure introduced in Definition \ref{def:AQL}, thereby proving Theorem \ref{theo:main1}.
The main references for the arguments presented in this section are \cite[Chapter 5.2]{Bratteli-Robinson-2} and \cite{ArakiMoriya2003}. The proofs of all the results stated here are deferred to Appendix \ref{ap:proof_S2}.

\subsection{The CAR algebra}\label{sub:car}
The construction of the universal $C^*$-algebra of the \emph{canonical anti-commutation relations}, in short  \emph{CAR algebra} over $\HH=L^2(\R^2)$ is described by  \cite[Theorem 5.2.5]{Bratteli-Robinson-2}. This algebra, denoted here with $\rr{A}\equiv\rr{A}(\HH)$ (we systematically suppress the symbol $\HH$ whenever no ambiguity can arise), 
is 
the unique, up to $\ast$-isomorphism, unital $C^{*}$--algebra generated by the
symbols $\{a(\phi),a^{\dag}(\phi)|\phi\in\HH\}$ subject to the condition that  $\phi\mapsto a(\phi)$ is
antilinear, $a^{\dag}(\phi)=a(\phi)^{*}$, and the anti-commutation relations
\begin{equation}\label{eq:CAR}
\{a(\phi),a^{\dag}(\psi)\}=\ip{\phi}{\psi}\,\unit,\qquad
\{a(\phi),a(\psi)\}=0,\qquad \forall\,\phi,\psi\in\HH\;.
\end{equation}
These relations force $\norm{a(\phi)}=\norm{\phi}$, 
namely the generators of
 $\rr{A}$ are bounded. Moreover  $\rr{A}$ is \emph{simple} and \emph{separable} (as a consequence of the separability of $\HH$).
    Finally, it is worth recalling that $\rr{A}$ has the structure of a  \emph{UHF (uniformly hyperfinite) algebra} (see \cite[Example 2.6.12]{Bratteli-Robinson-1} and Remark \ref{rk:uhf}).
When necessary, we will denote by $\mathfrak{A}_0$ the \emph{algebraic {\rm CAR} algebra}, \ie the dense unital $*$--subalgebra consisting of
finite linear combinations of the identity $\unit$ and of monomials  in the generators.
 
\smallskip

A  very relevant structure in $\rr{A}$ is provided by the \emph{Bogoliubov
transformations}. These are the $*$--automorphisms of the CAR algebra that act
linearly on the generators, mixing creation and annihilation operators while
preserving the canonical anticommutation relations. Let $\Aut(\rr{A})$ denote the group of $*$--automorphisms of $\rr{A}$. Concretely, let $u$ be a bounded
\emph{linear} and {$v$} a bounded \emph{antilinear} operator on $\HH$ satisfying
\[
{v^{*}}u+u^{*}v\;=\;0\;=\;uv^{*}+vu^{*}\;,\qquad
u^{*}u+v^{*}v\;=\;\unit\;=\;uu^{*}+vv^{*}\;.
\]
Then there is a \emph{unique} $*$--automorphism $\alpha_{u,v}\in\Aut(\rr{A})$ with
\[
\alpha_{u,v}\bigl(a(\phi)\bigr)\;:=\;a(u\phi)+a^{\dag}(v\phi)\;,\qquad \phi\in\HH\;,
\]
whose inverse is given by 
\[
\alpha_{u,v}^{-1}(a(\phi))\;:=\;a(u^{*}\phi)+a^{\dag}(v^{*}\phi)\;.
\] 
The four
relations above for $u$ and $v$ are precisely what guarantees that $\alpha_{u,v}$ maps the CAR relations
\eqref{eq:CAR} into themselves.

\smallskip

The special case $v=0$, which implies that $u$ must be unitary, recovers the
group of  Bogoliubov \emph{quasi--free} automorphisms.
More precisely, for every unitary $u$ on $\HH$ there is a \emph{unique} 
$*$--automorphism $\alpha_{u}\in\Aut(\rr{A})$
 with $\alpha_u(a(\phi))=a(u\phi)$. The
assignment $\bb{U}(\HH)\ni u\mapsto\alpha_u\in \Aut(\rr{A})$ is a group homomorphism, namely $\alpha_u\circ \alpha_{u'}=\alpha_{uu'}$. The continuity property is described in the next result.
\begin{proposition}[Continuity of Bogoliubov dynamics]\label{lem:strong-cont}
Let $\bb{I}\ni \imath\mapsto u_{\imath}\in \bb{U}(\HH)$ be
 a strongly continuous family of unitaries
and 
$\alpha_{u_\imath}\in\Aut(\rr{A})$   the associated Bogoliubov automorphisms. Then, for every $A\in\rr{A}$, the map
$\bb{I}\ni\imath\mapsto{\alpha_{u_\imath}(A)\in\rr A}$ is {continuous in the norm topology of $\rr A$}.
\end{proposition}

\smallskip

As a consequence of the result above, if $\R\ni t\mapsto u_t\in \bb{U}(\HH)$ is a strongly
continuous one--parameter group on $\HH$, then $\R\ni t\mapsto\alpha_{u_t}\in
\Aut(\rr{A})$ is a one--parameter group of $\ast$--automorphisms which is
pointwise continuous (\ie norm--continuous when evaluated at any element). In other words, a dynamics $u_t$ on $\HH$ lifts to a dynamics on $\rr{A}$  through the associated 
Bogoliubov \emph{quasi--free automorphisms} $\alpha_{u_t}$. This fact will play a central role and will be invoked repeatedly throughout the remainder of this paper.

\smallskip

A special {quasi--free} automorphism is the \emph{parity automorphism} implemented by $-\unit$, and denoted with $\sigma$. In concrete this is the unique $*$--automorphism specified by $\sigma(a(\phi))=-a(\phi)$ for every $\phi\in\HH$. 
As anticipated in Section \ref{intro} $\sigma$ is involutive and defines 
\emph{even} and \emph{odd} elements of $\rr{A}$, as well as the \emph{even}
$C^{*}$--subalgebra $\rr{A}^{+}\subset\rr{A}$ and the \emph{odd} Banach subspace
$\rr{A}^{-}\subset\rr{A}$.
 Both $\rr{A}^{\pm}$ are conveniently described in terms of the degree of the field
polynomials. Since $\sigma$ acts on a monomial $a^{\#}(\phi_1)\cdots a^{\#}(\phi_k)$
(where each $a^{\#}$ is either $a$ or $a^{\dag}$) by
\[
\sigma\bigl(a^{\#}(\phi_1)\cdots a^{\#}(\phi_k)\bigr)
=(-1)^{k}\,a^{\#}(\phi_1)\cdots a^{\#}(\phi_k)\;,
\]
such a monomial is even when its degree $k$ is even, and odd when $k$ is odd. As the
polynomials in the $a^\sharp(\phi)$ are norm--dense in $\rr{A}$, and the
projections $A\mapsto A^{\pm}$ defined by \eqref{eq:proj:pm}  are continuous, every even
$A\in\rr{A}^{+}$ (resp. every odd
$A\in\rr{A}^{-}$) is a norm limit of even--degree polynomials (resp. odd--degree polynomials). Consequently
$\rr{A}^{+}$ is the $C^{*}$--subalgebra generated by the even--degree monomials,
while $\rr{A}^{-}$ is the closed linear span of the odd--degree monomials.

\smallskip

\begin{remark}[Bogoliubov automorphisms and parity]\label{rk:B-A-parity}
It is worth noting that every Bogoliubov automorphisms preserve the canonical parity decomposition of the CAR algebra. More precisely, let $\alpha_{u,v}\in\Aut(\rr{A})$ be the Bogoliubov automorphism associated with the pair $u,v$. 
A direct computation shows that
\[
\begin{aligned}
\sigma\bigl(\alpha_{u,v}(a(\phi))\bigr)
\;&=\;
\sigma\bigl(a(u\phi)+a^\dag(v\phi)\bigr)\;=
\;-a(u\phi)-a^\dag(v \phi)\\
&=\;
-\alpha_{u,v}(a(\phi))\;=\;
\alpha_{u,v}\bigl(-a(\phi)\bigr)\;=\;
\alpha_{u,v}\bigl(\sigma(a(\phi))\bigr).
\end{aligned}
\]
The same computation applies to $a^\dag(\phi)$. Since the CAR algebra is generated by the 
$a^\sharp(\phi)$, one concludes that
$\alpha_{u,v}\circ\sigma=\sigma\circ\alpha_{u,v}$.  
\hfill$\blacktriangleleft$\end{remark}
 
\subsection{Bogoliubov dynamics and gauge transformations}\label{sec:gaug}

Let $h$ be a self--adjoint operator on the one--particle space $\HH$ with dense
domain $\s{D}(h)$. Since  Stone's theorem asserts that $\R\ni t\mapsto u_t:=\expo{\ii th}\in\bb{U}(\HH)$ is a strongly
continuous one--parameter group of unitary operators, Proposition~\ref{lem:strong-cont}
applies and the assignment
\[\tau_t\bigl(a(\phi)\bigr)\;:=\;a\bigl(\expo{\ii th}\phi\bigr)\;,\qquad
\tau_t\bigl(a^{*}(\phi)\bigr)\;:=\;a^{*}\bigl(\expo{\ii th}\phi\bigr)\;,\qquad \phi\in\HH\;,
\]
extends to a one--parameter group of $\ast$--automorphisms $\tau_t\in\Aut(\rr{A})$,
pointwise norm--continuous in $t$. We refer to it as the \emph{Bogoliubov (quasi--free)
dynamics} generated by $h$.
The \emph{Liouvillian} (or infinitesimal generator) of $\tau_t$ is
\[
\n{L}(A)\;:=\;-\,\ii\,\lim_{t\to0}\frac{\tau_{t}(A)-A}{t}\;,\qquad A\in\rr{D}(\n{L})
\subseteq\rr{A}\;,
\]
so that $\tau_t=\expo{\ii t\n{L}}$, the limit being taken in the norm topology and the
natural domain $\rr{D}(\n{L})$ consisting of all $A$ for which it exists. 
It turns out that $\rr{D}(\n{L})$ is a norm--dense $\ast$--subalgebra of $\rr{A}$, and that
$\n{L}$ is closed (hence closable) as the generator of a strongly continuous one--parameter group of
$\ast$--automorphisms. It is worth emphasizing that in the $C^*$-algebraic literature like \cite{Bratteli-Robinson-1,Bratteli-Robinson-2}, one often prefers to work with the \emph{derivation} $\delta=\ii\n{L}$ rather than with the Liouvillian $\n{L}$ itself.
If $\phi\in\s{D}(h)$, then $a(\phi)$
and $a^{\dag}(\phi)$ belong to $\rr{D}(\n{L})$, and a direct computation shows that
\[
\n{L}\bigl(a(\phi)\bigr)\;=\;-\,\ii\,a\bigl(\ii\,h\,\phi\bigr)\;=\;-\,a(h\,
\phi)\;,\quad
\n{L}\bigl(a^{\dag}(\phi)\bigr)\;=\;-\,\ii\,a^{\dag}\bigl(\ii\,h\,\phi\bigr)\;=\;
\,a^{\dag}(h\,\phi)\;.
\]
Let us recall that a subspace $\rr{D}_0\subseteq\rr{D}(\n{L})$ is a \emph{core} for the
closable operator $\n{L}$ if the closure of the restriction $\n{L}|_{\rr{D}_0}$ equals
$\n{L}$ or, equivalently, if $\rr{D}_0$ is dense in $\rr{D}(\n{L})$ for the graph norm
$\norm{A}_{\n{L}}:=\norm{A}+\norm{\n{L}(A)}$.
In particular, any $\tau$--invariant subspace $\rr{D}_0\subset\rr{D}(\n{L})$,
such that $\rr{D}_0$ is dense in $\rr{A}$,
 is a core for $\n{L}$
 \cite[Corollary~3.1.7]{Bratteli-Robinson-1}.

\smallskip

A particularly simple application of the above discussion is the \emph{group of gauge
transformations}. This is the Bogoliubov dynamics on $\rr{A}$ generated by the one--body Hamiltonian
$h=\unit$. The one--body dynamics $\R\ni\theta\mapsto\expo{\ii\theta}\unit\in\bb{U}(\HH)$ lifts to the
CAR  dynamics $\R\ni\theta\mapsto\vartheta_\theta\in\Aut(\rr{A})$, defined on the generators
by
\[\vartheta_\theta\bigl(a(\phi)\bigr)\;=\;\expo{-\ii\theta}\,a(\phi)\;,\qquad
\vartheta_\theta\bigl(a^{\dag}(\phi)\bigr)\;=\;\expo{\ii\theta}\,a^{\dag}(\phi)\;,\qquad
\phi\in\HH\;,
\]
and 
{the} sign in the phase arises from the antilinearity of $a(\cdot)$. One has $\vartheta_\theta\vartheta_{\theta'}=\vartheta_{\theta+
\theta'}$ and $\vartheta_0=\Id$. Moreover $\theta\mapsto\vartheta_\theta$ is $2\pi$--periodic.
Thus $\theta\mapsto\vartheta_\theta$ defines a representation of $\n{U}(1)$ in $\Aut(\rr{A})$.
Its image is called  the group of gauge automorphisms, or \emph{gauge group}, and denoted
by $\n{G}(\rr{A})\subset\Aut(\rr{A})$.
An element $A\in\rr{A}$ is \emph{gauge invariant} if $\vartheta_\theta(A)=A$ for all
$\theta\in\R$. Under $\vartheta_\theta$ a monomial with $p$ creation and $q$ annihilation
operators is multiplied by $\expo{\ii\theta(p-q)}$. Hence $A$ is gauge invariant if and only
if it lies in the $C^{*}$--subalgebra generated by the \emph{charge--balanced} monomials,
those with equally many creation and annihilation operators. This subalgebra is denoted
$\rr{A}^{\mathrm{g}}$, the \emph{gauge--invariant} subalgebra. It is contained in the even
subalgebra. In fact $\vartheta_\pi=\sigma$ implies that $\rr{A}^{\mathrm{g}}\subset\rr{A}^{+}$. The
inclusion is proper.  For instance $a(\phi)a(\psi)$ is even but transforms as
$\vartheta_\theta(a(\phi)a(\psi))=\expo{-2\ii\theta}\,a(\phi)a(\psi)$.
The action of the Liouvillian (or generator) of the gauge group, denoted with $\n{L}_{\rm g}$, is obtained by the general formula with $h=\unit$. One has that
\[
\n{L}_{\rm g}\bigl(a(\phi)\bigr)\;=\;-\,a(\phi)\;,\qquad
\n{L}_{\rm g}\bigl(a^{\dag}(\phi)\bigr)\;=\;a^{\dag}(\phi)\;,\qquad
\phi\in\HH\;.
\]
Since, as observed above, $\vartheta_\theta$ maps monomials to scalar multiples of themselves, it follows by linearity that $\vartheta_\theta(\rr{A}_0)=\rr{A}_0$ for every $\theta\in\R$. 
 In other words, the algebraic CAR algebra $\rr{A}_0$ is invariant under the gauge group. Since $\rr{A}_0$ is also dense in $\rr{A}$, it follows  by the general discussion above that
 $\rr{A}_0$ is a core for $\n{L}_{\rm g}$.


\subsection{Quasi--free states, canonical trace and Fock vacuum}\label{sec:Q-F}
The \emph{states} of $\rr{A}$ are the linear functionals $\omega: \rr{A} \to \n{C}$
such that $\omega(A^*A)\geqslant 0$ for all $A\in\rr{A}$ (\emph{positivity}) and 
$\omega(\unit)=1$ (\emph{normalization}). The space of states of  $\rr{A}$ will be denoted with $\s{S}(\rr{A})$. It is a topological space equipped with the $\ast$-weak topology.

\smallskip

Quasi--free states play an important role in the theory of CAR algebras and we refer to  \cite{Araki1971} or 
\cite[Chapter 17]{DerezinskiGerard} for a comprehensive treatment of the subject.
A state $\omega\in\s{S}(\rr{A})$ is called \emph{gauge--invariant quasi--free} if all odd correlators vanish, and
the even ones factorize by the \emph{fermionic Wick rule}
\[
\omega\bigl(a^{\dag}(\phi_1)\cdots a^{\dag}(\phi_m)\,a(\psi_n)\cdots a(\psi_1)\bigr)
\;=\;\delta_{m,n}\,\det\bigl[\omega(a^{\dag}(\phi_i)a(\psi_j))\bigr]_{i,j=1}^{n}\;.
\]
By definition, such a state vanishes on monomials that are not charge-balanced, \ie  those containing a distinct number of creation and annihilation operators. This is the reason for the  
gauge--invariance as clarified in Section \ref{sec:gaug}.
It turns out that a gauge--invariant quasi--free state is  completely determined by its \emph{two--point functions}
$\omega(a^{\dag}(\phi)a(\psi))$, or equivalently by its \emph{covariance}
$p\in\bb{B}(\HH)$, the unique operator $0\leqslant p\leqslant\unit$ with
\[
\ip{\psi}{p\phi}\;:=\;\omega(a^{\dag}(\phi)a(\psi))\;,\qquad \forall\; \phi,\psi\in\HH.
\]
Conversely every $0\leqslant p\leqslant\unit$ defines a unique gauge--invariant
quasi--free state, first through the two--point functions and then, by the Wick rule,
on monomials of any degree.  We write $\omega_p$ for the state associated
with the covariance $p$, and conversely.

\smallskip

\smallskip

A state $\tau \in \s{S}(\mathfrak{A})$ is \emph{tracial} if $\tau(AB) = \tau(BA)$ for all
$A,B \in \mathfrak{A}$. Being \emph{UHF} of type $2^{\infty}$, the algebra $\mathfrak{A}$ is the
inductive limit of full matrix algebras $\mathrm{Mat}_{2^{n}}(\mathbb{C})$, each of which carries a
unique tracial state given by the normalized trace. These are mutually compatible and
determine a unique tracial state on $\mathfrak{A}$, which we denote by $\omega_{\mathrm{tr}}$.
It turns out that $\omega_{\mathrm{tr}}$ coincides with the gauge--invariant quasi--free state
with covariance ${p} := (\sfrac{1}{2})\mathbf{1}$. From the definition of a gauge--invariant quasi--free state and the CAR, one obtains
\[
\omega_{{p}}\bigl(a^{\dag}(\phi)a(\psi)\bigr) \;=\; \frac{1}{2}\langle\psi, \phi\rangle \;=\; \omega_{{p}}\bigl(a(\psi)a^{\dag}(\phi)\bigr) \;, \qquad \forall\; \phi,\psi \in \HH \;.
\]
Since $\omega_{{p}}$ vanishes on unbalanced monomials, and the cyclic permutation of an unbalanced monomial remains unbalanced, one infers the cyclicity of $\omega_{p}$ on monomials of degree two. The passage to monomials of even degree is achieved by applying Wick's rule. Linearity implies that ${\omega_p}(XY) = {\omega_p}(YX)$ holds for all $X, Y \in \mathfrak{A}_0$. Finally, the extension of the cyclicity to the full algebra $\mathfrak{A}$ is performed by density, using the continuity of ${\omega_p}$ (being a state) and the continuity of the multiplication. As a result, one gets that ${\omega_p}(AB) = {\omega_p}(BA)$ for every $A,B \in \mathfrak{A}$. Thus, ${\omega_{p}}$ is tracial, and by the uniqueness of the trace on the simple algebra $\mathfrak{A}$, it is \emph{the} trace, $\omega_{\mathrm{tr}} = {\omega_{p}}$.

\smallskip

Another relevant state is the \emph{Fock} (or \emph{vacuum}) state
$\omega_0$ on $\rr{A}$, characterized by
\[
\omega_0\bigl(a^{\dag}(\phi)\,B\bigr)\;=\;0\;=\;\omega_0\bigl(B\,a(\phi)\bigr)\,,\qquad\forall
\,\phi\in\HH\,,\ \forall\,B\in\rr{A}\,,
\]
that is, by the requirement that all the annihilation operators annihilate the vacuum
in the associated GNS representation. In
particular $\omega_0\bigl(a^{\dag}(\phi)a(\psi)\bigr)=0$ for every $\phi,\psi\in\HH$, so
$\omega_0$ is the gauge--invariant quasi--free state with vanishing symbol $p=0$.

\begin{remark}[Invariance under the Bogoliubov dynamics]\label{rem:inv_tr_0} 
Let $\R\ni t\mapsto {\tau_t}\in \Aut(\rr{A})$ be the one--parameter group of Bogoliubov quasi–free automorphisms
associated to   a strongly
continuous one--parameter group $\R\ni t\mapsto u_t\in \bb{U}(\HH)$.
Both the trace state $\omega_{\mathrm{tr}}$ and the {Fock} state $\omega_0$ are invariant under the dynamics. Let $\omega_p$ be a gauge–invariant quasi–free state with covariance $p$. For every $\phi,\psi\in\HH$, and every $t\in\R$, one has that 
\[
\begin{aligned}
\omega_p\left({\tau_t}\left(a^{\dag}(\phi)a(\psi)\right)\right)\;&=\;
\omega_p\left(a^{\dag}(u_t\phi)a(u_t\psi)\right)\\
&=\;\ip{u_t\psi}{pu_t\phi}\;=\;\ip{\psi}{u_t^*pu_t\phi}\;.
\end{aligned}
\]
Therefore, if  $u_t^*pu_t=p$ one ends with $\omega_p\circ {\tau_t}=\omega_p$. This condition is satisfied both for $\omega_{\mathrm{tr}}$ (with covariance $p = (\sfrac{1}{2})\mathbf{1}$) and   $\omega_{0}$ (with covariance $p = 0$).
\end{remark} 
 
\subsection{Frame field operators and the quasi--local structure}\label{sub:qloc}
The abstract construction of a quasi--local algebra is described in detail in \cite[Chapter 2.6]{Bratteli-Robinson-1}, and the proof that the CAR algebra has an intrinsic structure of quasi--local algebra is provided in \cite[Proposition 5.2.6.]{Bratteli-Robinson-2}. In this section we will provide a quasi--local presentation for $\rr{A}$ subordinated to the magnetic coherent–state frame 
$\{\chi_\gamma\}_{\gamma\in\Gamma}$
described in Section \ref{sub:magframe}. 

\smallskip

To each index $\gamma\in\Gamma$ we attach the \emph{frame field operators}
\[a_{\gamma}\;:=\;a(\chi_{\gamma})\;,\qquad
a^{*}_{\gamma}\;:=\;a^{\dag}(\chi_{\gamma})\;.
\]
Because the frame is overcomplete, these do \emph{not} satisfy the exact CAR.
 Instead \eqref{eq:CAR} and Lemma~\ref{lem:overlap} give the \emph{deformed} relations
\begin{equation}\label{eq:deformed}
\begin{aligned}
\{a_{\gamma},a^{*}_{\gamma'}\}
\;&=\;\ip{\chi_{\gamma}}{\chi_{\gamma'}}\,\unit
\;=\;\delta_{r,r'}\,\expo{\frac{\ii}{2\lB^2}(\bl\times\bl')}
\expo{-\frac{|\bl-\bl'|^2}{4\lB^2}}\,\unit\;,\\
\{a_{\gamma},a_{\gamma'}\}\;&=\;0,
\end{aligned}
\end{equation}
where the right--hand side of the first relation 
is obtained by recalling the coordinate representation $\gamma\equiv(r,\bl)$  and $\gamma'\equiv(r',\bl')$ for the elements of $\Gamma$.

\smallskip

For each $\Lambda\in\s{P}_f(\Gamma)$ consider the  
subspace $\HH_\Lambda$
defined by \eqref{eq:H_LAM}.
It has finite dimension  $|\Lambda|$ in view of Lemma \ref{lem:findim}. Let
\[
\rr{A}_\Lambda\;\equiv\;\rr{A}(\HH_\Lambda)\;:=\;C^{*}\bigl(\{a_\gamma\;|\;\gamma\in\Lambda\}\bigr)\;\subseteq\;\rr{A}
\]
be  the associated finite-dimensional CAR subalgebra
generated by the frame field operators indexed by
$\Lambda$.
Since $\Lambda\subseteq\Lambda'$ implies $\rr{A}_\Lambda\subseteq
\rr{A}_{\Lambda'}$, the family $\{\rr{A}_\Lambda\}_{\Lambda\in\s{P}_f(\Gamma)}$ is a
net of $C^{*}$--algebras, with a common unit $\unit$ and ordered by inclusion. We set
\[
\rr{A}_{\mathrm{loc}}\;:=\;\bigcup_{\Lambda\in\s{P}_f(\Gamma)}\rr{A}_\Lambda\;.
\]

\begin{proposition}[Pre-quasi--local structure]
\label{prop:quasilocal-frame}
The net $\{\rr{A}_\Lambda\}_{\Lambda\in\s{P}_f(\Gamma)}$ exhausts the CAR algebra,
\[
\overline{\rr{A}_{\mathrm{loc}}}\;=\;\rr{A}\;,
\]
and thus endows $\rr{A}$ with the structure of a pre--quasi--local algebra.
\end{proposition}

\smallskip

\begin{remark}[UHF structure]\label{rk:uhf}
The result above explicitly shows that $\rr{A}$ is {a} UHF algebra. In fact $\rr{A}$ is the norm-closure of an increasing sequence of finite-dimensional full matrix algebras. 
In fact, in view of Lemma \ref{lem:findim} and \cite[Proposition 5.2.5 (2)]{Bratteli-Robinson-2} one has that
 $\rr{A}_\Lambda\simeq {\rm Mat}_{2^{|\Lambda|}}(\C)$ 
is  the full matrix algebra  over $\C^{2^{|\Lambda|}}$.
\hfill$\blacktriangleleft$
\end{remark}

\begin{remark}[Almost-orthogonality vs. quasi-local structure]\label{rem:approx-loc}
Comparing the quasi--local structure of $\rr{A}$ described in this section with
\cite[Definition~2.6.3]{Bratteli-Robinson-1}, one sees that properties (1), (2) and
(3) are satisfied, but not the disjointness condition (4). 
    Let us focus a little more on this point
 Let $\Lambda_1,\Lambda_2\in\s{P}_f(\Gamma)$ be such that
$\Lambda_1\cap\Lambda_2=\emptyset$. Since the family $\{\chi_\gamma\}_{\gamma\in\Gamma}$
is not orthogonal, the subspaces $\HH_{\Lambda_1}$ and $\HH_{\Lambda_2}$ are in
general not mutually orthogonal and the standard anti--commutation
condition $\{\rr{A}_{\Lambda_1}^{-},\rr{A}_{\Lambda_2}^{-}\}=\{0\}$ fails as well.
Here we used the notation
$\rr{A}_{\Lambda}^{\pm}:=\rr{A}_{\Lambda}\cap\rr{A}^{\pm}$. As a very simple
example, consider $\Lambda_1:=\{\gamma_1\}$ and $\Lambda_2:=\{\gamma_2\}$ with
$\gamma_1\neq\gamma_2$ but with the same energy index. Then
$
\{a_{\gamma_1},a^{*}_{\gamma_2}\}\neq0
$
in view of \eqref{eq:deformed}.
The commutation condition $[\rr{A}_{\Lambda_1}^{+},\rr{A}_{\Lambda_2}^{-}]=\{0\}$
fails too. Let $\Lambda_1:=\{\gamma_0,\gamma_1\}$ and $\Lambda_2:=\{\gamma_2\}$
with $\gamma_0,\gamma_1,\gamma_2$ pairwise distinct but with the same energy index.
Then, by exploiting the Leibniz identity, one gets
\begin{equation}\label{eq:come_o}
[a_{\gamma_0}a_{\gamma_1},a^{*}_{\gamma_2}]
\;=\;a_{\gamma_0}\{a_{\gamma_1},a^{*}_{\gamma_2}\}
-\{a_{\gamma_0},a^{*}_{\gamma_2}\}\,a_{\gamma_1}\;\neq\;0\;,
\end{equation}
since the anti--commutators are again non--vanishing, in {general}, in view of \eqref{eq:deformed}.
Finally, the last commutation condition
$[\rr{A}_{\Lambda_1}^{+},\rr{A}_{\Lambda_2}^{+}]=\{0\}$ fails too in view of the identity
\[
[a_{\gamma_0}a_{\gamma_1},a^{*}_{\gamma_2}a^{*}_{\gamma_3}]
\;=\;a^{*}_{\gamma_2}\,[a_{\gamma_0}a_{\gamma_1},a^{*}_{\gamma_3}]
+[a_{\gamma_0}a_{\gamma_1},a^{*}_{\gamma_2}]\,a^{*}_{\gamma_3}\;\neq\;0
\]
 together with \eqref{eq:come_o}. We will see in the Section \ref{sec:asy_grad_dec} that these exact commutation conditions can be replaced by the asymptotic condition expressed by property (4) of Definition \ref{def:AQL}. \hfill$\blacktriangleleft$
\end{remark}

We conclude this section with a result establishing the compatibility between the local structure induced by the algebras $\rr{A}_\Lambda$ and the gauge group described in Section \ref{sec:gaug}.
\begin{proposition}\label{prop:core-general-0}
The dense $\ast$-algebra $\rr{A}_{\mathrm{loc}}$
is invariant under the gauge group
    and therefore
it is a core for
the associated Liouvillian $\n{L}_{\rm g}$.
\end{proposition}

\subsection{Asymptotic decay of graded commutators}\label{sec:asy_grad_dec}
Let $\Lambda_1,\Lambda_2\subseteq\Gamma$ be nonempty, not necessarily finite. The
distance between them is defined by \eqref{eq:dist_set}.
Since the metric \eqref{eq:Xi-22} takes values in a discrete set, it is bounded away
from zero on distinct points, \ie $d(\gamma,\gamma')\geqslant d_0:=\min\{1,\ell_1/\lB,
\ell_2/\lB\}>0$ for every $\gamma\neq\gamma'$. Consequently, for $\Lambda_1,\Lambda_2\in
\s{P}_f(\Gamma)$ the induced set distance  is attained by a minimum (instead of an infimum) and satisfies
$
d(\Lambda_1,\Lambda_2)=0$ if and only if $\Lambda_1\cap\Lambda_2\neq\emptyset$. On the other hand 
$
d(\Lambda_1,\Lambda_2)\geqslant d_0>0$ if and only if $\Lambda_1\cap\Lambda_2=\emptyset
$.
Let us recall the notation $a^{\#}_\gamma\in\{a_\gamma,a^{*}_\gamma\}$  for the generic field operator \emph{located} at $\gamma$.
A (monic) \emph{monomial of
degree $k$} $a^{\#}_{\gamma_1}\cdots a^{\#}_{\gamma_k}$   is \emph{supported in} $\Lambda$ if and only if $\{\gamma_1,\ldots,\gamma_k\}\subseteq\Lambda$.

\smallskip

We now present the results leading to property (4) of Definition \ref{def:AQL}, which together with Proposition \ref{prop:quasilocal-frame}
establishes the AQL structure for the CAR algebra $\rr{A}$. The control of the graded commutators rests on a geometric quantity which measures the relative position of the subspaces $\HH_{\Lambda_1},\HH_{\Lambda_2}\subset\HH$ spanned by the frame vectors localized in $\Lambda_1$ and $\Lambda_2$. Since the frame is not orthogonal, these subspaces are in general not orthogonal even if $\Lambda_1\cap\Lambda_2=\emptyset$ and their orthogonality is recovered only asymptotically (Lemma~\ref{lem:delta-gram} below).

\smallskip

For $\Lambda_1,\Lambda_2\in\s P_f(\Gamma)$ let $\HH_{\Lambda_1},\HH_{\Lambda_2}\subset\HH$  be the associated
subspaces of dimension $|\Lambda_1|$ and $|\Lambda_2|$, respectively, as proved in Lemma~\ref{lem:findim}. Let $\pi_1$ and $\pi_2$ denote the orthogonal
projections of $\HH$ onto $\HH_{\Lambda_1}$ and $\HH_{\Lambda_2}$, respectively. Set
\[
\delta\;\equiv\;\delta(\Lambda_1,\Lambda_2)\;:=\;\bigl\|\pi_{2}\!\restriction_{\HH_{\Lambda_1}}\bigr\|\;=\;\|\pi_{2}\pi_{1}\|\;.
\]
The quantity $\delta\in[0,1]$ measures the coupling of the two subspaces. If $\delta=0$ they are
mutually orthogonal, and if $\delta=1$ some vector of $\HH_{\Lambda_1}$ lies entirely in
$\HH_{\Lambda_2}$. Note that
$\delta(\Lambda_1,\Lambda_2)=\delta(\Lambda_2,\Lambda_1)$, since
$\|\pi_2\pi_1\|=\|(\pi_2\pi_1)^*\|=\|\pi_1\pi_2\|$.

\begin{proposition}[Subspace coupling control]\label{lem:almost-comm-gr-delta}
For all $\Lambda_1,\Lambda_2\in\s P_f(\Gamma)$ and all homogeneous $A\in\rr A_{\Lambda_1}$,
$B\in\rr A_{\Lambda_2}$,
\begin{equation}\label{eq:almost-comm-gr-delta}
\bigl\|[A,B]_{\rm gr}\bigr\|\;\leqslant\;4\pi\,\sqrt{|\Lambda_1|\,|\Lambda_2|}\;\norm A\,\norm B\;\delta(\Lambda_1,\Lambda_2)\;.
\end{equation}
\end{proposition}

\smallskip

The coupling $\delta$ is controlled by the almost--orthogonality of the frame, at the price of the Gram matrix of the frame vectors. For $\Lambda\in\s P_f(\Gamma)$ let $G_\Lambda:=[\langle\chi_\gamma,\chi_{\gamma'}\rangle]_{\gamma,\gamma'\in\Lambda}$ be the Gram matrix, which is positive definite by Lemma~\ref{lem:findim}. Set $m(\Lambda):=\lambda_{\min}(G_\Lambda)>0$ for the minimal eigenvalue of $G_\Lambda$.

\begin{lemma}\label{lem:delta-gram}
Let $\Lambda_1,\Lambda_2\in\s P_f(\Gamma)$ and let $G_\lambda$ be as in Proposition~\ref{prop:magloc}. Then:
\begin{enumerate}[label=\emph{(\roman*)}]
\item for every $\lambda>0$,
\[
\delta(\Lambda_1,\Lambda_2)\;\leqslant\;G_\lambda\,\frac{\sqrt{|\Lambda_1||\Lambda_2|}}{\sqrt{m(\Lambda_1)\,m(\Lambda_2)}}\;\expo{-\lambda\,d(\Lambda_1,\Lambda_2)}\;;
\]
\item $m(\Lambda+\gamma_\bl)=m(\Lambda)$ for every $\bl\in\s L$, with $\gamma_\bl:=(\bl,0)$.
\end{enumerate}
\end{lemma}

\begin{proposition}[Almost graded--commutators]\label{lem:almost-comm-gr-II}
Let $\Lambda_1,\Lambda_2\in\s P_f(\Gamma)$ and $A\in\rr A_{\Lambda_1}$, $B\in\rr A_{\Lambda_2}$ homogeneous. For
every $\lambda>0$,
\begin{equation}\label{eq:almost-comm-gr-II}
\bigl\|[A,B]_{\rm gr}\bigr\|\;\leqslant\;C_\lambda(\Lambda_1,\Lambda_2)\,\norm A\,\norm B\,\expo{-\lambda\,d(\Lambda_1,\Lambda_2)}\;,
\end{equation}
with
\[
C_\lambda(\Lambda_1,\Lambda_2)\;:=\;4\pi\,G_\lambda\,\frac{|\Lambda_1|\,|\Lambda_2|}{\sqrt{m(\Lambda_1)\,m(\Lambda_2)}}\;.
\]
Moreover $C_\lambda(\Lambda_1+\gamma_\bl,\Lambda_2)=C_\lambda(\Lambda_1,\Lambda_2+\gamma_\bl)=C_\lambda(\Lambda_1,\Lambda_2)$ for every $\bl\in\s L$. 
\end{proposition}

\smallskip

Let $\s T$ be the group of isometries of $\Gamma$ generated by the lattice shifts
$\tau_\bl:\gamma\mapsto\gamma+\gamma_\bl$ with $\bl\in\s L$  (hence $\s T\simeq\s L$).
Proposition~\ref{lem:almost-comm-gr-II} then shows that property~(4) of
Definition~\ref{def:AQL} holds with respect to $\s T$. Together with
Proposition~\ref{prop:quasilocal-frame}, this proves
Theorem~\ref{theo:main1}.

\begin{remark}[From homogeneous to generic elements]\label{rk:hom_gen}
For odd elements $A\in\rr A_{\Lambda_1}^-$ and $B\in\rr A_{\Lambda_2}^-$, inequality
\eqref{eq:almost-comm-gr-II} translates directly into a bound on the anticommutator, since $[A,B]_{
\rm gr}=\{A,B\}$. When at least one of the two elements is even, the homogeneity hypothesis can be
dropped altogether. Indeed, split $A\in\rr A_{\Lambda_1}$ into its even and odd parts $A=A^++A^-$. If $B\in\rr A_{\Lambda_2}^+$ is even, then 
\[
[A,B]\;=\;[A^+,B]+[A^-,B]\;=\;[A^+,B]_{\rm gr}+[A^-,B]_{\rm gr}\,.
\]
By applying \eqref{eq:almost-comm-gr-II} to the two terms, together with $\norm{A^\pm}\leqslant\norm{
A}$, gives
\[
\bigl\|[A,B]\bigr\|\;\leqslant\;2\,C_\lambda(\Lambda_1,\Lambda_2)\,\norm A\,\norm B\,\expo{-\lambda\,d(\Lambda_1,\Lambda_2)}\,.
\]
The factor $2$ can be eventually absorbed into $C_\lambda(\Lambda_1,\Lambda_2)$.
The same argument applies if $A$ is even rather than $B$.
\hfill$\blacktriangleleft$
\end{remark}

\subsection{Volume dependence of the constant}\label{sec:bet:est}
By Proposition~\ref{lem:almost-comm-gr-II} the constant $C_\lambda(\Lambda_1,\Lambda_2)$ depends on the regions through the volumes $|\Lambda_j|$ (polynomially) and through the factors $m(\Lambda_j)^{-\sfrac12}$, where $m(\Lambda)$ denotes the smallest eigenvalue of the Gram matrix of $\{\chi_\gamma\}_{\gamma \in\Lambda}$. For a fixed pair of regions, or when one region is moved by lattice translations, as in the proofs of Section~\ref{Sect:frech_struct}, this is all that is needed, because $C_\lambda$ is translation invariant. For growing regions, instead, the dependence on the volume is not controlled. Indeed, the frame is overcomplete, hence the synthesis operator has a non--trivial kernel, {or} equivalently,  there is $0\neq c\in\ell^2(\Gamma)$ with $\sum_\gamma c_\gamma\chi_\gamma=0$. Truncating $c$ to $\Lambda\nearrow\Gamma$ gives $\bigl\|\sum_{\gamma\in\Lambda}c_\gamma\chi_\gamma\bigr\|\to0$ while $\|c\unit_\Lambda\|_{\ell^2(\Gamma)}\to\|c\|_{\ell^2(\Gamma)}>0$. This implies that $\inf_\Lambda m(\Lambda)=0$. 
At the current stage of our investigation, it is not known whether \(m(\Lambda)\) decays exponentially or polynomially in \(|\Lambda|\).
 We state this as an open problem.
\begin{openproblem}[Uniform subspace coupling]\label{lem:delta-uniform}
Do there exist constants $\mu_*>0$ and $C_*>0$, depending only on the frame
data, and a polynomial $F$ such that
\[
\delta(\Lambda_1,\Lambda_2)\;\leqslant\;C_*\,F(|\Lambda_1|)\,F(|\Lambda_2|)\,
\expo{-\mu_*\,d(\Lambda_1,\Lambda_2)}\,,\qquad\forall\,\Lambda_1,\Lambda_2\in\s P_f(\Gamma)\;?
\]
\end{openproblem}

\smallskip

By Proposition~\ref{lem:almost-comm-gr-delta}, a positive answer would give property~(4) of Definition~\ref{def:AQL} with a prefactor $C_\lambda$ depending polynomially on the volumes. The determination of the optimal growth of $F$, should such an $F$ exist, would be a substantial result, with potentially fruitful applications to the study of interacting magnetic systems.


\section{Almost--locality and Fr\'echet structure}
\label{Sect:frech_struct}

We introduce in this section an important tool for the study of locality properties of the CAR algebra $\rr{A}$, namely the \emph{conditional expectation} onto the local algebras.
It is worth emphasizing that the frame-field operators satisfy \emph{deformed} canonical anticommutation relations \eqref{eq:deformed}. As a consequence, operators localized at distinct indices do not, in general, anticommute, and the local algebras $\rr{A}_\Lambda$ do \emph{not} admit a tensor-product factorization across disjoint index sets. The construction developed below completely bypasses this obstruction, relying only on the finite-dimensional matrix structure of the local algebras and their canonical trace.
The conditional expectation  allows us to introduce a family of \emph{localizing norms}. These norms define a Fréchet subalgebra of $\rr{A}$ that turns out to be the natural framework for the study of asymptotic abelianness.
The proofs of all the results stated in this section are deferred to Appendix \ref{ap:proof_S3}.

\subsection{Conditional expectation}
\label{sec:cond-exp}

\smallskip
On the lattice the conditional expectation is built from the tensor factorization across a
region and its complement, via a parity twirl onto the even complement and a slice map
\cite[Thm.~4.7]{ArakiMoriya2003}. The overcompleteness of the frame obstructs the
\emph{geometric} complement, but not the \emph{orthogonal} one. Writing
$\HH=\HH_\Lambda\oplus\HH_\Lambda^{\perp}$ with $\HH_\Lambda^{\perp}$ the Hilbert
orthogonal complement, the same construction goes through verbatim, and the resulting map
is the trace--preserving conditional expectation determined by the unique tracial state
$\omega_{\mathrm{tr}}$ introduced in Section \ref{sec:Q-F}.

\smallskip

Recall from Lemma~\ref{lem:findim} that for every $\Lambda\in\s{P}_f(\Gamma)$ the
subspace $\HH_\Lambda$ spanned by $\{\chi_\gamma\}_{\gamma\in\Lambda}$ has
dimension $|\Lambda|$, so that the local algebra
\[
\rr{A}_\Lambda\;:=\;\rr{A}(\HH_\Lambda)\;=\;C^{*}\bigl(\{a_\gamma\;|\;\gamma\in\Lambda\}\bigr)
\;\cong\;M_{2^{|\Lambda|}}(\C)
\]
is a full matrix algebra, and $\tau_\Lambda:=2^{-|\Lambda|}\mathrm{Tr}$ is
its unique, \emph{faithful}, normalized trace. By using the orthogonal complement of $\HH_\Lambda$ one can introduce the \emph{orthogonal} subalgebra $\rr{A}_\Lambda^\bot \equiv \rr{A}(\HH_\Lambda^\bot)$. Because $\HH_\Lambda$ and $\HH_\Lambda^\bot$ are
genuinely orthogonal, the {\rm CAR} algebra factorizes as a \emph{graded} (fermionic)
tensor product
\[\rr{A}\;\simeq\;\rr{A}_\Lambda\,\hat{\otimes}\,\rr{A}_\Lambda^{\perp}\;,
\]
in which \emph{only} the even part of each factor commutes with the entire other factor, and the odd parts of the two factors \emph{anti}-commute with one another (\cf \cite[Definition 2.6.3]{Bratteli-Robinson-1}).
\begin{remark}[Orthogonal vs. geometric complement]
This orthogonal $\HH_\Lambda^{\perp}$ must
not be confused with the frame complement $\HH_{\Lambda^c}$ which is 
spanned by the complementary vectors
$\{\chi_\gamma\;|\;\gamma\notin\Lambda\}$. In fact vectors in $\HH_\Lambda^{\perp}$
overlap with vectors in $\HH_{\Lambda^c}$, the two spaces are not orthogonal, 
and there is \emph{no} splitting of the {form}
$\rr{A}_\Lambda \hat{\otimes}\rr{A}_{\Lambda^{c}}$ for the full algebra $\rr{A}$. \hfill$\blacktriangleleft$
\end{remark}

\begin{definition}[Conditional expectation]\label{def:cond-exp}
For $\Lambda\in\s{P}_f(\Gamma)$, the \emph{conditional expectation onto $\rr{A}_\Lambda$
with respect to $\omega_{\mathrm{tr}}$} is the unique linear map
$\n{E}_\Lambda:\rr{A}\to\rr{A}_\Lambda$ satisfying
\begin{equation}\label{eq:cond-exp-defining}
\omega_{\mathrm{tr}}\bigl(A\,B\bigr)\;=\;\omega_{\mathrm{tr}}\bigl(\n{E}_\Lambda(A)\,B\bigr)
\;,\qquad\forall\,A\in\rr{A}\,,\ \forall\,B\in\rr{A}_\Lambda\;.
\end{equation}
\end{definition}

\smallskip

From \eqref{eq:cond-exp-defining} it immediately follows that $\n{E}_\emptyset(A)=\omega_{\mathrm{tr}}(A)\unit$.
 Existence and the properties of the conditional expectation are summarized in the following result, whose proof closely follows the argument in \cite[Theorem~4.7]{ArakiMoriya2003}.

\begin{proposition}\label{prop:cond-exp}
For every $\Lambda\in\s{P}_f(\Gamma)$ the map $\n{E}_\Lambda$ of
Definition~\ref{def:cond-exp} exists and is unique. It is a unital, completely positive
projection of norm one onto $\rr{A}_\Lambda$, restricting to the identity on
$\rr{A}_\Lambda$, \ie $\n{E}_\Lambda|_{\rr{A}_\Lambda}=\mathrm{id}$. Moreover:
\begin{enumerate}[label=\emph{(\roman*)}]
\item it is a bimodule map, $\n{E}_\Lambda(BAC)=B\,\n{E}_\Lambda(A)\,C$ for
$B,C\in\rr{A}_\Lambda$;
\item  it is tower--compatible \emph{on nested regions}, namely if $\Lambda_1\subseteq\Lambda_2$ in $\s{P}_f(\Gamma)$, then
\[
\n{E}_{\Lambda_1}\circ\n{E}_{\Lambda_2}\;=\;\n{E}_{\Lambda_2}\circ\n{E}_{\Lambda_1}\;=\;\n{E}_{\Lambda_1}\;;
\]
\item it preserves the gauge--invariant subalgebra,
$\n{E}_\Lambda(\rr{A}^{\mathrm{g}})\subseteq\rr{A}_\Lambda^{\mathrm{g}}$.
\end{enumerate}
\end{proposition}

\smallskip

The tower property is stated only for \emph{nested} regions, and this restriction cannot be removed
in general. The identity $\n E_{\Lambda_1}\circ\n E_{\Lambda_2}=\n E_{\Lambda_1\cap\Lambda_2}$
fails for non--nested regions as {shown} in Remark \ref{rk:non-nested} below.
First we need the following  {result, which makes} explicit the action of $\n E_\Lambda$ on the field generators.
\begin{lemma}\label{lem:E-gen}
Let $\Lambda\in\s P_f(\Gamma)$ and let $\pi_\Lambda$ be the orthogonal projection of $\HH$ onto $\HH_\Lambda$.
For every $\phi\in\HH$,
\[
\n E_\Lambda\bigl(a(\phi)\bigr)=a(\pi_\Lambda\phi)\,,\qquad
\n E_\Lambda\bigl(a^\dag(\phi)\bigr)=a^\dag(\pi_\Lambda\phi)\,,
\]
and consequently $\bigl\|(\mathrm{id}-\n E_\Lambda)\bigl(a^\sharp(\phi)\bigr)\bigr\|=\norm{(\unit-\pi_\Lambda)\phi}$.
\end{lemma}

\smallskip

\begin{remark}[Failure of the tower property]\label{rk:non-nested}
Let $\gamma_1\neq\gamma_2$ be two sites of level $r=0$, and $\Lambda_i=\{\gamma_i\}$, so that
$\Lambda_1\cap\Lambda_2=\emptyset$. Then $\ip{\chi_{\gamma_1}}{\chi_{\gamma_2}}\neq0$ in view of the Gaussian overlap.
By Lemma~\ref{lem:E-gen}, $\n E_{\Lambda_1}\bigl(\n E_{\Lambda_2}(a_{\gamma_2})\bigr)=\n E_{\Lambda_1}(a_{\gamma_2})
=\overline{\ip{\chi_{\gamma_1}}{\chi_{\gamma_2}}}\,a_{\gamma_1}\neq0$. On the other hand, 
$\n E_{\Lambda_1\cap\Lambda_2}(a_{\gamma_2})=\n E_\emptyset(a_{\gamma_2})=\omega_{\mathrm{tr}}(a_{\gamma_2})\unit=0$.
The same computation shows that $\n E_\Lambda(a_\gamma)=a(\pi_\Lambda\chi_\gamma)\neq0$ in general for
$\gamma\notin\Lambda$, because the frame is not orthogonal.\hfill$\blacktriangleleft$
\end{remark}

\smallskip

All the uses of the tower property
in the following of this work will be restricted to 
 nested regions like balls of increasing radius, or a ball contained in a larger one.

\smallskip

The property (iii) of Proposition \ref{prop:cond-exp} is a consequence of the fact that the conditional expectation commutes with the gauge group. The next result is a direct generalization of this fact. 

\begin{lemma}\label{lem:dyn-commutes}
Let $\alpha\in\Aut(\rr{A})$ be a $\ast$-automorphism such that $\omega_{\mathrm{tr}}\circ\alpha=\omega_{\mathrm{tr}}$, and let $\Lambda,\Lambda'\in \s{P}_f(\Gamma)$ be such that $\alpha(\rr{A}_\Lambda)=\rr{A}_{\Lambda'}$. Then
\[\alpha\circ\n{E}_{\Lambda}\;=\;
\n{E}_{\Lambda'}\circ\alpha\;.
\]
In particular, if $\alpha(\rr A_\Lambda)=\rr A_\Lambda$ then $\alpha$ commutes with $\n E_\Lambda$.
\end{lemma}

\smallskip

The lemma applies to the gauge automorphisms $\vartheta_\theta$ and to the parity (with $\Lambda'=\Lambda$), and to the lattice translations $\alpha_\bl$ (with $\Lambda'=\Lambda+\gamma_\bl$, see \eqref{eq:lattice-cov-main} and Remark~\ref{rem:inv_tr_0}). The proof is a computation. Set $X:=\alpha(\n E_\Lambda(A))\in\rr A_{\Lambda'}$. For every $B\in\rr A_{\Lambda'}$ one has $\alpha^{-1}(B)\in\rr A_\Lambda$, and therefore
\[
\begin{aligned}
\omega_{\mathrm{tr}}(XB)\;&=\;\omega_{\mathrm{tr}}\bigl(\alpha(\n E_\Lambda(A)\,\alpha^{-1}(B))\bigr)\\
&=\;\omega_{\mathrm{tr}}\bigl(\n E_\Lambda(A)\,\alpha^{-1}(B)\bigr)\;=\;\omega_{\mathrm{tr}}\bigl(A\,\alpha^{-1}(B)\bigr)\;=\;\omega_{\mathrm{tr}}\bigl(\alpha(A)\,B\bigr)\;,
\end{aligned}
\]
where we used the invariance of $\omega_{\mathrm{tr}}$ ({second and fourth equality}) and the defining relation \eqref{eq:cond-exp-defining} of $\n E_\Lambda$ (third equality). By the uniqueness in Proposition~\ref{prop:cond-exp}, $X=\n E_{\Lambda'}(\alpha(A))$, which is the claim.

\subsection{Almost--local observables}

Once a conditional expectation is available, the \emph{almost--locality} of an observable can be quantified through its approximation by localizations, using suitable norms to control the corresponding localization errors. This approach is not new. The notion of almost--local observables was introduced in the context of algebraic quantum field theory in \cite{ArakiHaag1967}, later adapted to the lattice setting in \cite{BachmannDybalskiNaaijkens2016}, and developed further in \cite{Ogata2021,KapustinSopenko2022}.
 
\begin{definition}[Almost--local seminorms]\label{def:qloc-norm}
Let $\Lambda_k(\gamma_0):=\{\gamma\in\Gamma\;|\;d(\gamma,\gamma_0)\leqslant k\}$ be the
balls of radius ${k\in\N_0}$, and  centered at $\gamma_0\in\Gamma$,  for the metric $d$ in \eqref{eq:Xi-22}. For $p\in\N_0$ 
  and $A\in\rr{A}$ set
\begin{equation}\label{eq:norm_p}
\norm{A}_{p,\gamma_0}\;:=\;\norm{A}
+\sup_{k\in\N_0}\bigl\|A-\n{E}_{\Lambda_k(\gamma_0)}(A)\bigr\|\,(1+k)^{p}\;.
\end{equation}
Let $\rr{A}_p\subset\rr{A}$ be the set of  $A\in \rr{A}$ with $\norm{A}_{p,\gamma_0}
<\infty$ for a given $p$ and some (hence every) $\gamma_0$. Also define the family of \emph{almost-local observables} $\rr{A}_\infty:=\bigcap_{p\in\N_0}\rr{A}_p$.
\end{definition}
 
\smallskip

\begin{lemma}\label{lemm:eq_npr_p} 
Every $\rr{A}_p$, with $p\in\N_0$, is a Banach $\ast$-algebra {under the equivalent norm $5\norm{\cdot}_p$}
whose definition is
independent on the choice of base point $\gamma_0\in\Gamma$ in the norm \eqref{eq:norm_p}.
\end{lemma}

 \smallskip
 In view of the result above, one may define $\rr{A}_p$ and
$\rr{A}_\infty$ using the special choice $\gamma_0=0$ in \eqref{eq:norm_p}, that is with the norms defined in \eqref{eq:norm_p_0}.

\begin{theorem}[Fr\'echet structure of $\rr{A}_\infty$]\label{lem:loc-in-Ainfty}
The collection $\{\norm{\cdot}_p\}_{p\in\N_0}$ is
increasing family of norms, and the space
\[
\rr{A}_\infty\;=\;\bigl\{A\in\rr{A}\;\big|\;\norm{A}_p<\infty\ \forall\;
p\in\N_0\bigr\}
\]
is a \emph{Fr\'echet $\ast$--algebra} with respect to the topology induced by this family. More precisely it is complete,
$\ast$--invariant with $\norm{A^{*}}_p=\norm{A}_p$, and submultiplicative in the sense
that
\[
\norm{AB}_{p}\;\leqslant\;5\,\norm{A}_{p}\,\norm{B}_{p}\;,\qquad A,B\in\rr{A}_\infty\;,\
p\in\N_0\;.
\]
Moreover $\rr{A}_{\mathrm{loc}}\subset\rr{A}_\infty$, and $\rr{A}_\infty$ is the closure
of $\rr{A}_{\mathrm{loc}}$ in the Fr\'echet topology. Finally, the truncations converge at the
explicit rate
\[
\norm{A-\n{E}_{\Lambda_k}(A)}_{p}\;\leqslant\;\frac{2}{1+k}\,\norm{A}_{p+1}\;,
\qquad A\in\rr{A}_\infty\;.
\]
\end{theorem}

\smallskip

An immediate consequence of the result above is that $\rr{A}_\infty$ is closed under commutators since
\[
\norm{[A,B]}_p\;\leqslant\;10\,\norm{A}_p\norm{B}_p\;,
\]
 and the bracket
$[\,\cdot\,,\,\cdot\,]:\rr{A}_\infty\times\rr{A}_\infty\to\rr{A}_\infty$ is jointly
continuous for the Fr\'echet topology.


\subsection{Translations and discrete action}\label{sec:trasl_aut}
The magnetic translations $\{t_{\bz}\}_{\bz\in\R^2}$ lift to $\rr{A}$ as a family of
Bogoliubov $\ast$--automorphisms,
\[
\alpha_\bz\bigl(a(\phi)\bigr)\;:=\;a\bigl(t_{\bz}\phi\bigr)\;,\qquad
\alpha_\bz\bigl(a^\dag(\phi)\bigr)\;:=\;a^\dag\bigl(t_{\bz}\phi\bigr)\;,\qquad \forall\,\phi\in\HH\;.
\]
Since each $t_{\bz}$ is unitary, $\alpha_\bz\in\Aut(\rr{A})$, and by
{Proposition}~\ref{lem:strong-cont} together with the strong continuity of $\bz\mapsto t_{\bz}$,
the map $\bz\mapsto\alpha_\bz(A)$ is norm continuous for every $A\in\rr{A}$. Thus
$(\rr{A},\R^2,\alpha,{\varpi})${, with $\varpi=-\sigma_B$,} is a \emph{twisted $C^{*}$--dynamical system}. The meaning of the adjective \virg{twisted} is the content of the next remark. 

\begin{remark}[Projective composition and the magnetic cocycle]\label{rem:cocycle}
It is worth noticing that the mapping $\R^2\ni\bz\mapsto\alpha_\bz\in\Aut(\rr{A})$
is not a genuine group representation, but a \emph{projective} representation,
twisted by the gauge group through the \emph{magnetic cocycle} $\sigma_B$. From the
composition law \eqref{eq:weyl} of the magnetic translations and the antilinearity of
$\phi\mapsto a(\phi)$, one obtains
\[
\alpha_\bz\circ\alpha_{\bz'}\;=\;\vartheta_{{-}\sigma_B(\bz,\bz')}\circ\alpha_{\bz+\bz'}\;,
\qquad\sigma_B(\bz,\bz'):=\tfrac{1}{2\lB^2}(\bz\times\bz')\;,
\]
where $\vartheta$ is the gauge group. The scalar phase provided by the cocycle is the
genuine algebraic imprint of the magnetic field. The action descends to an ordinary
representation of the lattice $\s{L}$ precisely when
$\sigma_B(\bz,\bz')\in2\pi\Z$ for all lattice vectors $\bz,\bz'$, \ie under the
flux--quantization condition $f_B\in4\pi\,\Z$. This condition, however, is
incompatible with the frame property of the coherent states established in
Theorem~\ref{thm:frame}. In other words, the presence of the cocycle $\sigma_B$ is
inherent to our construction.
\hfill$\blacktriangleleft$
\end{remark}

On the frame field operators the action translates the underlying coherent state.
Since $\chi_{\bl,r}=t_\bl\psi_{r,0}$, and
\begin{equation}\label{eq:t-chi}
t_{\bz}\chi_{\bl,r}\;=\;\expo{-\frac{\ii}{2\lB^2}(\bz\times\bl)}\,t_{\bl+\bz}\psi_{r,0}
\;=\;\expo{-\frac{\ii}{2\lB^2}(\bz\times\bl)}\,\chi_{\bl+\bz,\,r}\;,
\end{equation}
where $ \chi_{\bl+\bz,r}$  is the coherent state centered at
$\bl+\bz$. The latter is a lattice site if and only if $\bz\in \s{L}$, and in such a case
\[
\alpha_\bz(a_{\bl,r})\;=\;\expo{\ii\,\sigma_B(\bz,\bl)}\,a_{\bl+\bz,r}
\;=\;\vartheta_{{-}\sigma_B(\bz,\bl)}\bigl(a_{\bl+\bz,r}\bigr)\;.
\]
 In particular $\alpha_\bz$
preserves the Landau levels and translates the spatial center. For lattice
translations $\bl\in \s{L}$ it permutes the frame sites, so that
\begin{equation}\label{eq:lattice-cov-main}
\alpha_\bl(\rr{A}_\Lambda)\;=\;\rr{A}_{\Lambda+\bl}\,,\qquad|\Lambda+\bl|\;=\;|\Lambda|\,,
\end{equation}
and the action is compatible with the
local structure. This fact underlies the following result.

  \begin{lemma}\label{lemm:aut-fr-lattice}
For every $A\in\rr A_p$, $\bl\in\s L$ and $\gamma_0\in\Gamma$ the equality
\begin{equation}\label{eq:lattice-cov}
\norm{\alpha_\bl(A)}_{p,\gamma_0}\;=\;\norm{A}_{p,\gamma_0-\gamma_\bl}
\end{equation}
holds, with $\gamma_\bl:=(\bl,0)$. As a consequence $\alpha_\bl\in\mathrm{Aut}(\rr A_\infty)$ for
every $\bl\in\s L$.
\end{lemma}

\smallskip

    We emphasize that in the claim of Lemma \ref{lemm:aut-fr-lattice} the symbol ${\rm Aut}(\rr{A}_\infty)$ denotes the group of $\ast$-automorphisms of $\rr{A}_\infty$ continuous
with respect to the Fréchet topology. 

\smallskip

In view of the preceding result, the discrete translation group $\s L$ acts on $\rr A_\infty$ by $\ast$-automorphisms. Combined with Proposition~\ref{lem:almost-comm-gr-II}, this immediately yields a discrete version of asymptotic abelianness. For purely local elements, the graded commutators vanish  {faster than any exponential}.

\begin{lemma}\label{cor:asympt-ab-lattice}
Let $A,B\in\rr A_{\rm loc}$
homogeneous. Then, 
\[
\lim_{\substack{\bl\in\s L\\ |\bl|\to\infty}}\bigl\|[\alpha_\bl(A),B]_{\rm gr}\bigr\|\;=\;0\,,
\]
and the convergence is faster than ${\rm e}^{-\lambda|\bl|}$ for every $\lambda>0$.
\end{lemma}

Lemma~\ref{cor:asympt-ab-lattice} provides an explicit {bound decaying faster than any exponential} on $[\alpha_\bl(A),B]_{\rm gr}$ for purely local elements $A$ and $B$. A standard density argument removes the locality restriction, at the expense of losing quantitative information on the rate of decay. 

\begin{proposition}[Discrete asymptotic abelianness]\label{cor:asympt-ab-full}
It holds true that 
\[
\lim_{\substack{\bl\in\s L\\ |\bl|\to\infty}}\bigl\|[\alpha_\bl(A),B]_{\rm gr}\bigr\|\;=\;0\,,
\]
for every homogeneous $A,B\in\rr A$.
\end{proposition} 

\smallskip

Combining Proposition~\ref{cor:asympt-ab-full} with Remark~\ref{rk:hom_gen}, we obtain
$
[\alpha_\bl(A),B]\to 0$,
for every $ A,B\in\rr A$,
whenever at least one of the two elements is even. If both $A$ and $B$ are odd, then instead
$
\{\alpha_\bl(A),B\}\to 0$. Thus, the asymptotic commutation relations take the expected graded form in all cases.

\subsection{Continuous asymptotic abelianness}
For generic $\bz\in\R^2$ the  automorphism $\alpha_\bz$
does not preserve the pure local structure, and also the symbol $\Lambda+\bz$ \emph{does not} make sense.
However, surprisingly, $\alpha_\bz$  sends pure local elements  into 
     almost-local elements.  This fact  relies on the composition property of translations. Any $\bz\in\R^2$ admits a unique
decomposition $\bz=\bz_L+\bz_c$, with $\bz_L\in\s L$ and $\bz_c\in[0,\ell_1)\times[0,\ell_2)$; in
particular $|\bz_c|<\ell_0:=(\ell_1^2+\ell_2^2)^{\sfrac12}$. From the composition property of the
magnetic translations, including the magnetic cocycle, one derives
\begin{equation}\label{eq:dec_trasl}
\alpha_\bz\;=\;\vartheta_{\sigma_B(\bz_L,\bz_c)}\circ\alpha_{\bz_L}\circ\alpha_{\bz_c}\,.
\end{equation}
By {Lemma~\ref{lem:dyn-commutes}}, $\vartheta_\theta$ commutes with $\n E_\Lambda$, so that
$\norm{\vartheta_\theta(A)}_{p,\gamma_0}=\norm{A}_{p,\gamma_0}$ for all $\theta\in\R$. In combination
with Lemma~\ref{lemm:aut-fr-lattice}, this yields
\[
\norm{\alpha_\bz(A)}_{p,\gamma_0}\;=\;\norm{\alpha_{\bz_c}(A)}_{p,\gamma_0-\gamma_{\bz_L}}\,,
\]
which is meaningful whenever $\alpha_{\bz_c}(A)\in\rr A_p$. The following result shows that this is indeed the case when $A\in \rr A_{\mathrm{loc}}$.
\begin{lemma}\label{lemma_trasl_loc}
For every  $\bz\in\R^2$ it holds that   $\alpha_\bz:\rr A_{\mathrm{loc}}\to\rr A_\infty$.
\end{lemma}

\smallskip

We stress that Lemma~\ref{lemma_trasl_loc} provides only a \emph{membership} statement, as its proof does not yield an estimate of the form
\begin{equation}\label{eq:estim_norm}
\norm{\alpha_\bz(A)}_{p}\;\leqslant\; C_p\norm{A}_{q}\;,
\end{equation}
for some $q=q(p)\geqslant p$ and $C_p>0$, which is required to establish the Fr\'echet continuity of $\alpha_\bz$ and, consequently, to conclude that $\alpha_\bz\in\mathrm{Aut}(\rr A_\infty)$.
The next result is a quantitative version, for a \emph{fixed} local element, of the membership statement of Lemma~\ref{lemma_trasl_loc}. In the spirit of \eqref{eq:estim_norm},  it shows that $\alpha_\bz(A)$ has finite almost--local seminorms $\norm\cdot_{p,\gamma_0}$ for every $p\in\N_0$, with a bound in terms of the operator norm of $A$ alone, uniform for $|\bz|<\ell_0$. The constant $C_p(\Lambda,\gamma_0)$ depends on the base point $\gamma_0$ and 
{on the support $\Lambda$ of $A$, through its cardinality and through the smallest eigenvalue of the Gram matrix of $\{\chi_\gamma\}_{\gamma\in\Lambda}$.} For this reason it is not exactly an estimate of the form \eqref{eq:estim_norm}. It is nevertheless sufficient for the  purposes of this work, and in particular for the proof of Theorem~\ref{thm:asympt-ab}.

\begin{lemma}\label{cor:mono-improve-full}
Let $A\in\rr A_{\Lambda}$ be a 
local element  with $\Lambda\in\s P_f(\Gamma)$. For every $p\in\N_0$ and $\gamma_0\in\Gamma$ there is 
a constant $C_p(\Lambda,\gamma_0)>0$
 depending on $p$, on
$\Lambda$ and $\gamma_0$ such that
\begin{equation}\label{eq:ineq_z_unif}
\norm{\alpha_{\bz}(A)}_{p,\gamma_0}\;\leqslant\;C_p(\Lambda,\gamma_0)\,\norm{A}\;,
\end{equation}
for every $\bz\in\R^2$ with $|\bz|<\ell_0$.
\end{lemma}

\smallskip

 The above result provides a fundamental building block for the proof of asymptotic abelianness.

\begin{proposition}[Continuous asymptotic abelianness for local elements]\label{prop:cont-asympt-ab-local}
Let $A,B\in\rr A_{\mathrm{loc}}$ be homogeneous elements. Then
\[
\lim_{|\bz|\to\infty}\bigl\|[\alpha_\bz(A),B]_{\rm gr}\bigr\|\;=\;0\;,\qquad\bz\in\R^2\,,
\]
the limit being taken over the full continuous group $\R^2$.
\end{proposition}

\smallskip

The last result expresses the compatibility between graded asymptotic abelianness and the almost--local structure. The decay $[\alpha_\bz(A),B]_{\rm gr}\to0$ for local $A,B$ is obtained through the norms $\norm{\cdot}_{p,\gamma_0}$ of $\rr A_\infty$. Theorem \ref{thm:asympt-ab} then follows by extending the statement from $\rr A_{\mathrm{loc}}$ to all of $\rr A$ via a density argument (we recall that it also follows from the elementary argument sketched in Section \ref{intro}).

\smallskip

The results above (Lemmas  \ref{lemma_trasl_loc} and \ref{cor:mono-improve-full}) do not show that $\alpha_\bz$ is itself a continuous automorphism of $\rr A_\infty$ for a generic $\bz\in\R^2$. They control $\alpha_\bz(A)$ only for a fixed local element $A$, with a constant depending on its support and with the operator norm of $A$ on the right-hand side, whereas Fr\'echet continuity requires a bound of the form \eqref{eq:estim_norm}. Note that a bound with $\norm A$ on the right-hand side and a constant independent of the support could not hold: by density it would give $\alpha_\bz(A)\in\rr A_\infty$ for every $A\in\rr A$, that is $\rr A_\infty=\rr A$. However, this is false since one can explicitly show that $\rr A_\infty\neq\rr A$. This leads to the following problem.

\smallskip

\begin{openproblem}[Continuous automorphisms of the almost--local algebra]
\label{open:aut-A-infty}
Is $\alpha_\bz$, for every $\bz\in\R^2$, a continuous automorphism of the
Fr\'echet algebra $\rr A_\infty$? Equivalently: for every $p\in\N_0$, do there exist $q\geqslant p$ and $C_p>0$ such that
$\norm{\alpha_\bz(A)}_{p}\leqslant C_p\norm{A}_{q}$ for every $A\in\rr A_{\mathrm{loc}}$?
\end{openproblem}
The equivalence follows from the density of $\rr A_{\mathrm{loc}}$ in $\rr A_\infty$, the completeness of $\rr A_\infty$ and the fact that $\alpha_\bz$ is isometric on $\rr A$. For the lattice translations an estimate of this form holds, with $q=p$ and a different base point (Lemma~\ref{lemm:aut-fr-lattice}).

\appendix

\section{Proofs and technical results}
\label{app_prrof}

\subsection{Proofs and auxiliary material for  Section \ref{sec:frames}}\label{ap:proof_S1}

\proof[Proof of Lemma \ref{Lem:01}]
From the second equality in \eqref{eq:Tb} one gets
\[
\ip{\psi_{r,0}}{t_{\bz}\psi_{r,0}}\;=\;\expo{-\frac{|\bz|^2}{4\lB^2}}\left\langle\psi_{r,0},E_{\bz}(\bp)E_{\bz'}(\bmm)\psi_{r,0}\right\rangle
\]
where $E_{\bz}(\bp)$ and $E_{\bz'}(\bmm)$ denote the exponentials of  
$\bp$ and $\bmm$, respectively, and $\bz'\equiv(z_1,-z_2)$ for {consistency} of notation. Observing that $E_{\bz'}(\bmm)\psi_{r,0}=\psi_{r,0}$ and 
$\ip{\psi_{r,0}}{E_{\bz}(\bp)\psi_{r,0}}=\ip{\psi_{r,0}}{\psi_{r,0}}$,
one obtains the result.
\qed

\smallskip

\begin{proof}[Proof of Lemma \ref{lemm:gau_loc}]
The zero-angular-momentum Landau eigenfunctions are,
\[\psi_{r,0}(\br) \;=\; \psi_{0,0}(\br)\,\frac{\ii^r}{\sqrt{r!}}\left(\frac{x_1+\ii x_2}{\lB\sqrt2}\right)^{r}\;.
\]
By magnetic translation covariance, $|\chi_{\bz,r}(\br)| = |\psi_{r,0}(\br-\bz)|$, and writing $\rho:=|\br-\bz|$, one gets
\[
|\chi_{\bz,r}(\br)| \;=\; f_r(\rho)\;:=\;\frac{1}{\sqrt{2\pi}\,\lB\sqrt{r!}}\left(\frac{\rho}{\sqrt2\,\lB}\right)^{r}\expo{-\frac{\rho^2}{4\lB^2}}\;.
\]
Let $r\neq0$. Writing $t:=\sfrac{\rho}{\sqrt2\,\lB}$, the exponent $r\log (t)-\sfrac{t^2}{2}$ is maximized at $t_{\max}=
\sqrt r$, namely  $\rho_{\max}=\sqrt{2r}\,\lB$, where
\[
f_r(\rho_{\max})\;=\;\frac{1}{\sqrt{2\pi}\lB}\,\frac{r^{\frac{r}{2}}\expo{-\frac{r}{2}}}{\sqrt{r!}}\;
\leqslant\;\frac{1}{\sqrt{2\pi}\lB}\,\frac{1}{(2\pi r)^{\frac{1}{4}}}\;=\;\frac{(2\pi)^{-\frac{3}{4}}}{\lB\,
r^{\frac{1}{4}}}\,,
\]
the inequality being by the Stirling's bound $r!\geqslant\sqrt{2\pi r}\,(\sfrac{r}{\rm e})^{r}$. For the shape,
set $u:=\sfrac{(t-\sqrt r)}{\sqrt r}=\sfrac{(\rho-\rho_{\max})}{\rho_{\max}}$. Then
\[
\log\left[\frac{f_r(\rho)}{f_r(\rho_{\max})}\right]\;=\;r\bigl[\log(1+u)-u-\tfrac{u^2}{2}\bigr]\;\leqslant\;-
\frac{r\,u^{2}}{2}\;=\;-\frac{(\rho-\rho_{\max})^{2}}{4\lB^{2}}\,,
\]
where we used $\log(1+u)-u\leqslant0$ for all $u>-1$.
Combining everything together, we obtain
\[
f_r(\rho)\;\leqslant\;\frac{(2\pi)^{-\frac{3}{4}}}{\lB\,
r^{\frac{1}{4}}}\,\expo{-\frac{(\rho-\sqrt{2r}\,\lB)^{2}}{4\lB^{2}}}
\;\leqslant\;\frac{1}{\sqrt{2\pi}\,\lB\,(1+r)^{\frac14}}\,\expo{-\frac{(\rho-\sqrt{2r}\,\lB)^{2}}{4\lB^{2}}}
\]
where the second inequality follows from $2\pi\geqslant1+\sfrac{1}{r}$
which implies
\[
\frac{(2\pi)^{-\frac{3}{4}}}{\lB\,
r^{\frac{1}{4}}}\;\leqslant\;\frac{\left(1+\tfrac{1}{r}\right)^{-\frac{1}{4}}}{\sqrt{2\pi}\,\lB\,
r^{\frac{1}{4}}}\;=\;\frac{1}{\sqrt{2\pi}\,\lB\,(1+r)^{\frac14}}\;.
\]
Observing that for  $r=0$ one has that
\[
|\chi_{\bz,0}(\br)|\; =\; |\psi_{0,0}(\br-\bz)|\;=\; f_0(\rho)\;:=\;
\frac{1}{\lB\sqrt{2\pi}}\,\expo{-\frac{\rho^2}{4\lB^2}}
\]
one concludes the validity of \eqref{eq:ring-bound} for every $r\in\N_0$.
\end{proof}

\smallskip

\begin{proof}[Proof of Lemma \ref{lem:overlap}]
Using the relations \eqref{eq:weyl} one has
\[
\ip{\chi_{\bz,r}}{\chi_{\bz',r'}}
\;=\;\ip{\psi_{r,0}}{t_{-\bz}t_{\bz'}\psi_{r',0}}
\;=\;\expo{\frac{\ii}{2\lB^2}(\bz\times\bz')}
\ip{\psi_{r,0}}{t_{\bz'-\bz}\psi_{r',0}}\;.
\]
By \eqref{eq:Tb}, $t_{\bz'-\bz}$ preserves each Landau level, so
$t_{\bz'-\bz}\psi_{r',0}\in\HH_{r'}$. The orthogonality of distinct levels
gives the factor $\delta_{r,r'}$. For $r=r'$ the diagonal value
$\ip{\psi_{r,0}}{t_{\bz'-\bz}\psi_{r,0}}$ is computed in Lemma \ref{Lem:01}.
\end{proof}

\smallskip

\proof[Proof of Lemma \ref{lemma:reconstr_op}]
From the reconstruction formula  \eqref{eq:rec_can} one infers that $I_\Lambda\psi\to\psi$ when $\Lambda\nearrow\Gamma$ for every $\psi\in\HH$. This means
$I_\Lambda\to\unit$ in the strong operator topology. Observing that, by self--adjointness of $S^{-1}$,
\[
I_\Lambda^{*}S^{-1}\psi
=\sum_{\gamma\in\Lambda}\ip{\xi_\gamma}{S^{-1}\psi}\,S^{-1}\xi_\gamma
=\sum_{\gamma\in\Lambda}\ip{S^{-1}\xi_\gamma}{\psi}\,S^{-1}\xi_\gamma
=S^{-1}I_\Lambda\psi\;,
\]
one obtains $I_\Lambda^{*}S^{-1}=S^{-1}I_\Lambda$, that is
$I_\Lambda^{*}=S^{-1}I_\Lambda S$. Since $S^{-1}$ is bounded and $I_\Lambda\to\unit$
strongly, for every $\psi\in\HH$ one has $I_\Lambda^{*}\psi=S^{-1}I_\Lambda(S\psi)\to
S^{-1}(S\psi)=\psi$. Hence also $I_\Lambda^{*}\to\unit$ in the strong operator topology as
well.
Write $I_\Lambda=D_\Lambda\circ C$ as the composition of the partial synthesis  operator $D_\Lambda:\ell^2(\Gamma)\to\HH$ defined by
$D_\Lambda(c):=\sum_{\gamma\in\Lambda}c_\gamma\xi_\gamma$, and the dual analysis operator $C:\HH\to\ell^2(\Gamma)$ defined
by $(C\psi)_\gamma:=\ip{S^{-1}\xi_\gamma}{\psi}$. 
Then $\norm{I_\Lambda}\leqslant\norm{D_\Lambda}\norm{C}$. Let  $J_\Lambda$ be the orthogonal projection of
$\ell^2(\Gamma)$ onto $\ell^2(\Lambda)$, namely
$J_\Lambda e_\gamma=e_\gamma$ if   $\gamma\in\Lambda$, and $J_\Lambda
=0$ otherwise, where $\{e_\gamma\}_{\gamma\in\Gamma}$ is the canonical orthonormal basis of
$\ell^2(\Gamma)$. Then $D_\Lambda=DJ_\Lambda$, with $D$ the full synthesis  operator, and in turn 
$\norm{D_\Lambda}\leqslant \norm{D}\leqslant\sqrt{B}$
in view of  \cite[Theorem 3.2.3]{Christensen-frames}.
On the other hand $C$ is the dual of the full synthesis operator with respect to the canonical dual frame which has  {upper frame bound $A^{-1}$}. Then using the arguments in \cite[Lemma 5.1.5]{Christensen-frames}
one gets $\norm{C}\leqslant\sfrac{1}{\sqrt{A}}$. This proves the norm bound for $I_\Lambda$. The equality $\|I_\Lambda\|=\|I_\Lambda^*\|$ is a general result for bounded operators on a Hilbert space.\qed

\smallskip

\proof[Proof of Proposition \ref{prop:inherited}]
In \cite[Proposition~2.2 \& Example~4]{BachmannDeNittis2024} the fixed--rate case is
proved. For fixed frame constants $\lambda>0$ and $G=G_\lambda>0$ as in
Definition~\ref{def:loc-frame}, the bound \eqref{eq:Sp-loc} holds with the following constants.
Let $s:=\min\sigma(S)>0$ and $\norm{S}$ be the spectral bounds of $S$. Fix free parameters $0<\epsilon<\delta<\lambda$ and $0<\theta
<\lambda-\delta$, and put
\[
r_p\;:=\;1-\Bigl(\tfrac{s}{\norm{S}}\Bigr)^{p}\in(0,1)\,,\quad d_{p,\epsilon}\;:=\;1+\Bigl(
\tfrac{G\,m_\epsilon}{\norm{S}}\Bigr)^{p}\,,\quad E_{p,\epsilon,\theta,\delta}\;:=\;\frac{
\lambda-\delta-\theta}{\ln\!\bigl(\sfrac{d_{p,\epsilon}}{r_p}\bigr)}\,.
\]
Then \eqref{eq:Sp-loc} holds with
\[
\lambda_{p,\epsilon,\theta,\delta}\;:=\;\min\Bigl\{\theta,\ \ln\!\bigl(\sfrac{1}{r_p}\bigr)E_{p,
\epsilon,\theta,\delta}\Bigr\}\,,\qquad A_p\;:=\;\frac{2\,G}{\norm{S}^{p}\,(1-r_p)}\,.
\]
We note that $d_{p,\epsilon}$ and $A_p$ differ by a factor $G$ from the expressions
in \cite[Proposition~2.2]{BachmannDeNittis2024},  {because the frame constant is allocated differently between the two. The constants above are those that follow from the proof below.} Optimizing over
the free parameters,
\[
\lambda_p\;:=\;\sup_{\substack{0<\epsilon<\delta<\lambda\\ 0<\theta<\lambda-\delta}}\lambda_{p,
\epsilon,\theta,\delta}
\]
is the best rate the method provides.
Set $\beta:=\ln(\sfrac{1}{r_p})/\ln(\sfrac{d_{p,\epsilon}}{r_p})$. Since
$d_{p,\epsilon}>1$ one has  $\beta\in(0,1)$, and
$\ln(\sfrac{1}{r_p})E_{p,\epsilon,\theta,\delta}=\beta(\lambda-\delta-\theta)$. Thus
\[
\lambda_{p,\epsilon,\theta,\delta}\;=\;\min\bigl\{\theta,\ \beta(\lambda-\delta-\theta)\bigr\}
\;\leqslant\;\frac{\beta}{1+\beta}\,(\lambda-\delta)\;<\;\frac12(\lambda-\delta)\;<\;\frac
\lambda2\,,
\]
the first inequality being the maximum of $\min\{\theta,\beta(\lambda-\delta-\theta)\}$ over
$\theta$, attained at $\theta=\beta(\lambda-\delta)/(1+\beta)$, and the second following from
$\beta<1$. Taking the supremum over the admissible parameters gives $\lambda_p<\sfrac
\lambda2<\lambda$, with gap $0<\lambda-\lambda_p<\sfrac\lambda2$. The bound is
intrinsic to the method. Enlarging $\theta$ shrinks the numerator $\lambda-\delta-\theta$ of $E_{p,\epsilon,
\theta,\delta}$, and the two entries of the minimum cannot be made large simultaneously.
 Now let the frame be localized at every rate, and write $\lambda_p
(\lambda)$ for the optimal output rate above, viewed as a function of the frame rate
$\lambda$. The spectral quantities $s,\norm{S},r_p$ do not depend on $\lambda$, the sole
dependence being through $G_\lambda$ in $d_{p,\epsilon}$. For $\lambda$ large choose the admissible triple $\theta=\delta=\tfrac
\lambda4$, $\epsilon=\tfrac\lambda8$. Since $m_\epsilon$ is non--increasing with $m_\epsilon\to1$
as $\epsilon\to\infty$, there are $M_0\geqslant1$ and $\lambda_0$ with $m_{\lambda/8}\leqslant
M_0$ for $\lambda\geqslant\lambda_0$. Hence, for such $\lambda$,
\[
\ln\!\Bigl(\tfrac{d_{p,\epsilon}}{r_p}\Bigr)\;\leqslant\;p\ln (G_\lambda)+C\,,\qquad C\;\equiv\;C(p,M_0,
\norm{S},r_p)\,,
\]
and, with $\lambda-\delta-\theta=\tfrac\lambda2$,
\[
\lambda_p(\lambda)\;\geqslant\;\lambda_{p,\frac{\lambda}{8},\frac{\lambda}{4},\frac{\lambda}{4}}\;\geqslant\;\min\Bigl\{\tfrac\lambda4,\ \frac{\ln(\sfrac{1}{r_p})}{2}\,\frac{\lambda}{
p\ln (G_\lambda)+C}\Bigr\}\,.
\]
Under \eqref{eq:growth-cond} one has that $\ln (G_\lambda)/\lambda\to0$, so both entries diverge and
$\lambda_p(\lambda)\to\infty$.
 Finally, given $\lambda>0$, pick $\lambda_\ast$ with $\lambda_p
(\lambda_\ast)\geqslant\lambda$. Applying \eqref{eq:Sp-loc} with the frame  rate
$\lambda_\ast$  one gets
\[
\bigl|\ip{\xi_\gamma}{S^{-p}\xi_{\gamma'}}\bigr|\;\leqslant\;A_p(\lambda_\ast)\,\expo{-\lambda_p
(\lambda_\ast)\,d(\gamma,\gamma')}\;\leqslant\;A_{p,\lambda}\,\expo{-\lambda\,d(\gamma,\gamma')}
\,,
\]
which is \eqref{eq:Sp-loc01} with $A_{p,\lambda}:=A_p(\lambda_\ast)$.
\qed

\smallskip

\begin{proof}[Proof of Lemma \ref{lem:summ-dim}]

Fix $\gamma=(\bl,r)\in\Gamma$. Using \eqref{eq:Xi-22} for  $d(\gamma,\gamma')$ one gets
$$
\sum_{\gamma'\in\Gamma}\expo{-\epsilon\,d(\gamma,\gamma')}\;=\;s_\epsilon(r)
\left(\sum_{n_1\in\Z}\expo{-\epsilon\frac{\ell_1}{\lB}|n_1|}\right)
\left(\sum_{n_2\in\Z}\expo{-\epsilon\frac{\ell_2}{\lB}|n_2|}\right)\;,\qquad
$$
where
\[
s_\epsilon(r)\;:=\;\sum_{r'\in\N_0}\expo{-\epsilon|r-r'|}\;.
\]
The energy level variable has to be treated with care. Since $r'\geqslant0$, the map $r'\mapsto|r-r'|$ is one--to--one only for $r=0$, and two--to--one
(for $|r-r'|\leqslant r$) otherwise. Splitting the sum in $r'\leqslant r$ and $r'>r$, one obtains
$$
s_\epsilon(r)\;=\;\frac{1-\expo{-\epsilon(r+1)}}{1-\expo{-\epsilon}}+\frac{\expo{-\epsilon}}{1-\expo{-\epsilon}}
\;=\;\coth\!\Big(\frac\epsilon2\Big)-\frac{\expo{-\epsilon(r+1)}}{1-\expo{-\epsilon}}\;.
$$
Hence $s_\epsilon(r)$ is strictly increasing in $r$, with $s_\epsilon(0)=(1-\expo{-\epsilon})^{-1}$ and
\[
\sup_{r\in\N_0}s_\epsilon(r)\;=\;\lim_{r\to\infty}s_\epsilon(r)\;=\;\coth\!\Big(\frac\epsilon2\Big)\;.
\] 
 The remaining two factors are computed using
\[
\sum_{n\in\Z}\expo{-a|n|}\;=\;\frac{2}{1-\expo{-a}}-1\;=\;\coth\left(\frac{a}{2}\right)\;,
\]
and they do not depend on $\gamma$. Taking the supremum over $\gamma$ gives \eqref{eq:geo_gamma1}. For  \eqref{eq:geo_gamma2}
fix $\gamma_\ast=(\bl_\ast,r_\ast)$ and $R\geqslant1$. A point
$\gamma'=(\bl',r')$ with $d(\gamma_\ast,\gamma')<R$ satisfies \emph{both} $|r'-r_\ast|<R$ and
$\tfrac{1}{\lB}|\bl'-\bl_\ast|_1<R$. Therefore
\[
N_\ast(R)\;\leqslant\;E_\ast(R)\,D_\ast(R)
\]
where $N_\ast(R):=\left.|\{\gamma'\in\Gamma\,\right|\,d(\gamma,\gamma')<R\}|$,
\[
D_\ast(R)\;:=\;\left.\left|\left\{\bl'\in \s{L}\;\right|\;|\bl'-\bl_\ast|_1<\lB R\right\}\right|\;\leqslant\;\left(\frac{2\lB R}{\ell_1}+1\right)\left(\frac{2\lB R}{\ell_2}+1\right)
\]
(the inequality follows  looking at $\ell_j|n_j'-n_j^0|<\lB R$
for each
axis), and finally
\[
E_\ast(R)\;:=\;\left.\left|\left\{r'\in\N_0\;\right|\;|r'-r_\ast|<R\right\}\right|\;\leqslant\;2R+1
\] 
(constraint $r'\geqslant0$ can
only reduce the count). For $R\geqslant1$ one has $2R+1\leqslant 3R$ and
$\tfrac{2\lB R}{\ell_j}+1\leqslant\bigl(\tfrac{2\lB}{\ell_j}+1\bigr)R$, whence 
\[
N_\ast(R)\;\leqslant\;3\bigl(\tfrac{2\lB}{\ell_1}+1\bigr)\bigl(\tfrac{2\lB}{\ell_2}+1\bigr)\,R^3\;.
\]
Since the bound for $N_\ast(R)$ is uniformly in $\gamma_\ast$ one obtains the result.
\end{proof}

\smallskip

\proof[Proof of Theorem \ref{thm:frame}]
Condition \eqref{eq:threshold} is necessary and sufficient for $\{\chi_{\bl,0}\}_{\bl\in\s{L}}$ to form an overcomplete frame for the lowest Landau level $\HH_0$ (see \cite[Theorem 5.6]{BachmannDeNittis2024}).
Now,
 for $r\in\N_0$, let $v_r:={\sqrt{r!}}^{-1}(\ap)^r$. Its restriction
$v_r|_{\HH_0}:\HH_0\to\HH_r$ is unitary. Equivalently, $v_r\Pi_0$ is a partial
isometry with initial space $\HH_0$ and final space $\HH_r$. 
Therefore,
if $\{\chi_{\bl,0}\}_{\bl\in \s{L}}$ is a frame of $\HH_0$ with bounds
$A\leqslant B$, then for every $r\in\N_0$ the family
$\{\chi_{\bl,r}\}_{\bl\in \s{L}}$, with $\chi_{\bl,r}:=v_r\chi_{\bl,0}$, is a frame of $\HH_r$ with the \emph{same}
bounds $A,B$. Since $\HH=\bigoplus_{r\in\N_0}\HH_r$, then $\{\chi_\gamma\}_{\gamma\in\Gamma}$ is a frame for $\HH$ with bounds $A,B$.
\qed

\smallskip

\begin{proof}
[Proof of Lemma \ref{lem:findim}]
For $\Lambda=\emptyset$ the claim is tautological. Then we can assume $\Lambda\neq\emptyset$.
Write $\Lambda_r:=\{\bl\in \s{L}|(\bl,r)\in\Lambda\}$. Since coherent
states of different Landau levels are orthogonal (Lemma~\ref{lem:overlap}), one has
the orthogonal decomposition $\HH_\Lambda=\bigoplus_{r}\HH_{\Lambda_r}$, whence
$\dim(\HH_\Lambda)=\sum_r\dim(\HH_{\Lambda_r})$. It thus suffices to prove
$\dim(\HH_{\Lambda_r})=|\Lambda_r|$ for each fixed $r$. 
Let $v_r$
by the interlevel unitary, introduced in the proof of Theorem \ref{thm:frame}.
Since $v_r\chi_{\bl,0}=\chi_{\bl,r}$, linear independence at
level $r$ is equivalent to linear independence at level $0$. We are therefore reduced
to showing that, for pairwise distinct lattice points $\bl_1,\dots,\bl_n$, the vectors
$\chi_{\bl_1,0},\dots,\chi_{\bl_n,0}$ are linearly independent in $\HH_0$.
Recall the explicit form of the lowest–level coherent states 
\[
\chi_{\bl,0}(\br)
\;=\;c\;
\expo{\frac{\ii}{2\lB^2}(\br\times\bl)}\;
\expo{-\frac{|\br-\bl|^2}{4\lB^2}}\;
\]
where $c>0$ is the common normalization constant.
Expanding $|\br-\bl|^2=|\br|^2-2\,\br\cdot\bl+|\bl|^2$ and separating the factors
that do not depend on $\br$, this reads
\[
\chi_{\bl,0}(\br)
\;=\;\psi_{0,0}(\br)\underbrace{\,\expo{-\frac{|\bl|^2}{4\lB^2}}}_{=:g(\bl)}\;
\expo{\,\boldsymbol{k}(\bl)\cdot\br}\;,
\qquad
\boldsymbol{k}(\bl)\;:=\;\frac{1}{2\lB^2}\bigl(\bl+\ii\,\bl^{\perp}\bigr)\in\C^{2},
\]
where $\bl^{\perp}:=(\,\ell_2 n_2,\,-\ell_1 n_1\,)$   so that $\br\times\bl=\br\cdot\bl^{\perp}$, and $g(\bl)\neq0$. Thus the whole
$\bl$--dependence of the $\br$--variable is carried by the exponential of the
\emph{linear} form $\boldsymbol{k}(\bl)\cdot\br$.
The map $\bl\mapsto\boldsymbol{k}(\bl)$ is
injective (it is sufficient to consider the real part), so
the vectors $\boldsymbol{k}(\bl_1),\dots,\boldsymbol{k}(\bl_n)\in\C^2$ are pairwise
distinct. Suppose now that $\sum_{j=1}^{n}\alpha_j\,\chi_{\bl_j,0}=0$ in $\HH_0$. 
Dividing by the
common nowhere--vanishing factor $\psi_{0,0}(\br)$
one gets
\begin{equation}\label{eq:rest_c}
\sum_{j=1}^{n}\beta_j\,\expo{\boldsymbol{k}(\bl_j)\cdot\br}\;=\;0\;,
\qquad\forall\,\br\in\R^{2}\;.
\end{equation}
with
$\beta_j:=\alpha_jg(\bl_j)$.
For each pair $i<j$ consider the vector ${\boldsymbol r}_{i,j}:= \bl_i-\bl_j$. There is a finite number of these vectors. Therefore there exists a  unit vector ${\boldsymbol w}\in\R^2$ which is not perpendicular to any of the ${\boldsymbol r}_{i,j}$, namely 
such that ${\boldsymbol r}_{i,j}\cdot {\boldsymbol w}\neq 0$ for all {$1\leqslant i<j\leqslant n$}. For such a $\boldsymbol w$ the real
numbers $\bl_j\cdot\boldsymbol w$ are pairwise distinct, hence so are
$\mu_j:=\boldsymbol{k}(\bl_j)\cdot\boldsymbol w$ (they already differ in their real
parts).
Restricting the vanishing condition \eqref{eq:rest_c}
to the line $\br=t\boldsymbol{w}$ with  $t\in\R$, gives
\[
\sum_{j=1}^{n}\beta_j\,\expo{\,\mu_j t}\;=\;0\quad\forall\,t\in\R\;,
\qquad \mu_j:=\boldsymbol{k}(\bl_j)\cdot\boldsymbol{w}\ \text{ pairwise distinct}\;.
\]
Applying $\sfrac{d^{k}}{dt^{k}}$ and evaluating at $t=0$ yields
$\sum_{j}\beta_j\mu_j^{k}=0$ for $k=0,\dots,n-1$, a Vandermonde system with distinct
nodes $\mu_j$, hence invertible. Therefore $\beta_1=\dots=\beta_n=0$. Since
$g(\bl)\neq0$ also $\alpha_1=\dots=\alpha_n=0$. This proves that
$\chi_{\bl_1,0},\dots,\chi_{\bl_n,0}$ are linearly independent.
\end{proof}

\smallskip

\proof[Proof of Proposition \ref{pr:comm}]
It suffices to show that $[S,\Pi_r]=0$ for an arbitrary $r\in\N_0$, since the general result follows by linearity and continuity in the strong operator topology. A direct computation provides
\[
\begin{aligned}
S(\Pi_r\psi)\;&=\;\sum_{\gamma\in \Gamma}\ip{\chi_\gamma}{\Pi_r\psi}\chi_\gamma\;=\;\sum_{\gamma\in \Gamma}\ip{\Pi_r\chi_\gamma}{\psi}\chi_\gamma\;=\;\sum_{\bl\in \s{L}}\ip{\chi_{\bl,r}}{\psi}\chi_{\bl,r}\\
&=\;\sum_{\bl\in \s{L}}\ip{\chi_{\bl,r}}{\psi}\Pi_r\chi_{\bl,r}\;=\;\Pi_r\left(\sum_{\gamma\in \Gamma}\ip{\chi_\gamma}{\psi}\chi_\gamma\right)\;=\;\Pi_r(S\psi)\;,
\end{aligned}
\]
and since it holds for an arbitrary $\psi\in\HH$ one infers that 
$S\Pi_r=\Pi_r S$.
\qed

\smallskip

\begin{proof}[Proof of Proposition \ref{prop:magloc}]
Lemma~\ref{lem:overlap} provides that
\[
\bigl|\ip{\chi_{\gamma}}{\chi_{\gamma'}}\bigr|\;=\;\delta_{r,r'}\,\expo{-\frac{|\bl-\bl'|^2}{4\lB^2}}\;\leqslant\;\expo{-\lambda|r-r'|}\,\expo{-\frac{|\bl-\bl'|^2}{4\lB^2}}\;.
\]
since $\delta_{r,r'}\leqslant\expo{-\lambda|r-r'|}$  for every $\lambda>0$.
Using the standard inequality $|\bl|_1\leqslant\sqrt{2}\,|\bl|$
valid in $\R^2$, one gets that 
$|\bl-\bl'|^{2}\geqslant\sfrac12|\bl-\bl'|_1^{2}$. Writing
$\rho:=\sfrac{1}{\lB}|\bl-\bl'|_1$,
one obtains
\[
\expo{-\frac{|\bl-\bl'|^{2}}{4\lB^2}}
\;\leqslant\;\expo{-\frac{|\bl-\bl'|_1^{2}}{8\lB^2}}
\;=\;\expo{-\frac{\rho^{2}}{8}}
\;\leqslant\;\expo{\,2\lambda^{2}}\,\expo{-\lambda\rho}\;,
\]
for every $\lambda>0$, since
 a Gaussian is dominated by exponential. More precisely in the last step
one uses $\sup_{\rho\geqslant0}\bigl(-\tfrac{\rho^2}{8}+\lambda\rho\bigr)
=2\lambda^{2}$ to show that $\expo{\,2\lambda^{2}}$ is the correct maximization. Multiplying the two bounds and recalling
the definition of $d$ one obtains the \eqref{eq:magloc}
with
 $G_\lambda=\expo{\,2\lambda^{2}}$. 
 \end{proof}

\smallskip

\begin{proof}[Proof of Lemma \ref{lem:dual-localization}]
By applying \eqref{S_eq} to the vector $S^{-2}\chi_\gamma$, with $\gamma\equiv
(\bl,r)\in\Gamma$, one gets
\[
S^{-1}\chi_\gamma\;=\;\sum_{\gamma'\in\Gamma}[S^{-2}]_{\gamma'\gamma}\,\chi_{\gamma'}\,,\qquad
[S^{-2}]_{\gamma'\gamma}:=\ip{\chi_{\gamma'}}{S^{-2}\chi_\gamma}\,.
\]
 By $[S, h_B]=0$ (Proposition \ref{pr:comm}) the sum runs over
$\gamma'=(\bl',r)$ at the fixed level $r$.
The coefficients
$[S^{-2}]_{\gamma'\gamma}=\ip{\chi_{\gamma'}}{S^{-2}\chi_\gamma}$ are exactly the matrix
elements of $S^{-2}$ in the frame. By Proposition~\ref{prop:inherited} with $p=2$ they decay
exponentially,
\[
\bigl|[S^{-2}]_{\gamma'\gamma}\bigr|\;\leqslant\;A_2\,\expo{-\lambda_\ast\,d(\gamma,\gamma')}\;\leqslant\;
A_2\,\expo{-\lambda_\ast\,\frac{|\bl-\bl'|}{\lB}}\,,
\]
at the rate $\lambda_\ast=\lambda_2$ of Remark~\ref{rem:dual-magnetic}. The second {inequality} follows from the constraint $r=r'$ for the energy levels
which gives  $d(\gamma,\gamma')={\lB}^{-1}|\bl-\bl'|_1$
along with $|\bl-\bl'|=\sqrt{(\bl-\bl')^2}\leqslant |\bl-\bl'|_1$.
 Therefore
 \[
 \begin{aligned}
\bigl|(S^{-1}\chi_\gamma)(\br)\bigr|\;&\leqslant\;A_2\sum_{\bl'\in \s{L}}\expo{-\lambda_\ast\,\frac{|\bl-\bl'|}{\lB}}\,\bigl|\chi_{\bl',r}(\br)\bigr|\;.
\end{aligned}
\]
Introducing the bound in Lemma \ref{lemm:gau_loc} one gets
\[
\bigl|(S^{-1}\chi_\gamma)(\br)\bigr|\;\leqslant\;\frac{A_2}{\sqrt{2\pi}\,\lB\,(1+r)^{\frac14}}\,\s{G}(\br)
\]
where
\[
\s{G}(\br)\;:=\;\sum_{\bl'\in \s{L}}\expo{-\lambda_\ast\frac{|\bl-\bl'|}{\lB}}\;\expo{
-\frac{\bigl(|\br-\bl'|-\sqrt{2r}\,\lB\bigr)^{2}}{4\lB^{2}}}
\]
It remains to bound
$\s G(\br)$ uniformly in $r$.
  Assume first $|\br-\bl|\geqslant\sqrt{2r}\,\lB$ and set
\[
d\;:=\;\frac{|\br-\bl|}{\lB}-\sqrt{2r}\;\geqslant\;0\,,\qquad q\;:=\;\frac{|\br-\bl'|-\sqrt{2r}\,
\lB}{\lB}\,.
\]
From $|\bl-\bl'|\geqslant|\br-\bl|-|\br-\bl'|=\lB d-\lB q$ one gets
\[
\expo{-\lambda_\ast\frac{|\bl-\bl'|}{\lB}}\;=\;\expo{-\frac{\lambda_\ast}{2}\frac{|\bl-\bl'|}{\lB}
}\,\expo{-\frac{\lambda_\ast}{2}\frac{|\bl-\bl'|}{\lB}}\;\leqslant\;\expo{-\frac{\lambda_\ast}{2}d
}\,\expo{\frac{\lambda_\ast}{2}q}\,\expo{-\frac{\lambda_\ast}{2}\frac{|\bl-\bl'|}{\lB}}\,.
\]
 Since the Gaussian inside the sum that  defines $\s{G}(\br)$ is 
 $\expo{-\frac{q^
2}{4}}$, 
   the $q$--dependent factors inside the sum combine into $\expo{\frac{\lambda_\ast}{2}q-\frac14q^2}$. Since
\[
\frac{\lambda_\ast}{2}q-\frac14q^{2}\;=\;\frac{\lambda_\ast^{2}}{2}-\frac18\bigl(q-2\lambda_\ast
\bigr)^{2}-\frac18q^{2}\;\leqslant\;\frac{\lambda_\ast^{2}}{2}-\frac18q^{2}\,,
\]
one obtains $\expo{\frac{\lambda_\ast}{2}q-\frac14q^2}\leqslant\expo{\frac{\lambda_\ast^2}{2}}\expo{-\frac{q^2}{8}}\leqslant\expo{\frac{\lambda_\ast^2}{2}}$, where in the last inequality $\expo{-\frac{q^2}{8}}\leqslant 1$ has been used. Therefore
\[
\s G(\br)\;\leqslant\;\expo{\frac{\lambda_\ast^{2}}{2}}\,\expo{-\frac{\lambda_\ast}{2}d}\sum_{\bl'\in\s L}
\expo{-\frac{\lambda_\ast}{2}\frac{|\bl-\bl'|}{\lB}}\,.
\]
With the same proof of Lemma \ref{lem:summ-dim}, and $|\bl-\bl'|\geqslant\sfrac{|\bl-\bl'|_1}{\sqrt{2}}$ one can prove that
 \[
 \sum_{\bl'\in\s L}
\expo{-\frac{\lambda_\ast}{2}\frac{|\bl-\bl'|}{\lB}}\;\leqslant\;
\coth\left(
\frac{\lambda_\ast\ell_1}{4\sqrt2\,\ell_B}
\right)
\coth\left(
\frac{\lambda_\ast\ell_2}{4\sqrt2\,\ell_B}
\right)\;=:\;C_\ast'\;. 
 \]
 Therefore, one obtains that  
 \[
\bigl|(S^{-1}\chi_\gamma)(\br)\bigr|\;\leqslant\;\frac{A_2\,\expo{\frac{\lambda_\ast^{2}}{2}}\,C_\ast'}{\sqrt{2\pi}\,\lB\,(1+r)^{\frac14}}\,\expo{-\frac{\lambda_\ast}{2}\left(\frac{|\br-\bl|}{\lB}-\sqrt{2r}\right)}\;,\qquad |\br-\bl|\geqslant\sqrt{2r}\,\lB\;.
\]
For
$|\br-\bl|<\sqrt{2r}\,\lB$ one can use the uniform bound
\[
\s{G}(\br)\;\leqslant\;\sum_{\bl'\in \s{L}}\expo{-\lambda_\ast\frac{|\bl-\bl'|}{\lB}}\;\leqslant\;\sum_{\bl'\in\s L}
\expo{-\frac{\lambda_\ast}{2}\frac{|\bl-\bl'|}{\lB}}\;\leqslant\;C_\ast'\leqslant\;\expo{\frac{\lambda_\ast^{2}}{2}}\,C_\ast'\;.
\]
Then, by setting
\[
C_\ast\;:=\;\frac{A_2\,\expo{\frac{\lambda_\ast^{2}}{2}}\,C_\ast'}{\sqrt{2\pi}\,\lB\,}
\]
one obtains the inequality in the claim.
\end{proof}

\smallskip

\begin{proof}[Proof of Lemma \ref{lem:dual-coh-overlap}]
The operator $S^{-1}$ preserves the Landau levels \textup{(}Proposition~\ref{pr:comm}\textup{)},
so the overlap vanishes unless $r'=r$; assume $r'=r$. By $|\ip{\,\cdot\,}{\cdot}|\leqslant\langle|
\,\cdot\,|,|\cdot|\rangle$ and the localizations of Lemmas~\ref{lemm:gau_loc}
and~\ref{lem:dual-localization},
\begin{equation}\label{eq:overlap-int}
\bigl|\ip{S^{-1}\chi_{(\bl,r)}}{\chi_{(\bz,r)}}\bigr|\;\leqslant\;\frac{C_\ast}{\sqrt{2\pi}\,\lB\,(
1+r)^{1/2}}\,\s I\,,
\end{equation}
where
\[
\s I\;:=\;\int_{\R^2}\!\expo{-\frac{\lambda_\ast}{2}(\frac{|\bx-\bl|}{\lB}-\sqrt
{2r})_+}\expo{-\frac{(|\bx-\bz|-\sqrt{2r}\lB)^2}{4\lB^2}}\dd^2\bx\,.
\]
Set $w:=|\bx-\bz|-\sqrt{2r}\lB$ and $d_0:=\sfrac{|\bl-\bz|}{\lB}-2\sqrt{2r}$. The reverse triangle
inequality 
 $|\bx-\bl|\geqslant|\bl-\bz|-|\bx-\bz|$ 
gives $\sfrac{|\bx-\bl|}{\lB}-\sqrt{2r}\geqslant d_0-\sfrac{w}{\lB}$, and in polar
coordinates about $\bz$ \textup{(}$\dd^2\bx=2\pi\rho\,\dd\rho$ and $\rho=\sqrt{2r}\lB+w$\textup{)},
\[
\s I\;\leqslant\;2\pi\int_{-\sqrt{2r}\lB}^{\infty}\dd w\;\!\bigl(\sqrt{2r}\lB+w\bigr)\,\expo{-\frac{\lambda
_\ast}{2}[d_0-w/\lB]_+}\,\expo{-\frac{w^2}{4\lB^2}}\,.
\]
Since $[x]_+\geqslant x$, one has $[d_0-w/\lB]_+\geqslant d_0-w/\lB$, so that, completing the
square, one gets
\[
\expo{-\frac{\lambda_\ast}{2}\left[d_0-\frac{w}{\lB}\right]_+}\expo{-\frac{w^2}{4\lB^2}}\;\leqslant\;\expo{-\frac{
\lambda_\ast}{2}d_0}\,\expo{\lambda_\ast^2/4}\,\expo{-\frac{(w-\lambda_\ast\lB)^2}{4\lB^2}}\,.
\]
Substituting $u:=w-\lambda_\ast\lB$, bounding $\sqrt{2r}\lB+w\leqslant\sqrt{2r}\lB+\lambda_\ast\lB+
|u|$, and using 
\[\int_\R\dd u\;\expo{-\frac{u^2}{4\lB^2}}\;=\;2\sqrt\pi\lB\;,\qquad \int_\R\dd u\;|u|\expo{-\frac{u^2}{4\lB^2}}\;=\;4
\lB^2
\]
 together with $\sqrt{2r}\leqslant\sqrt2\sqrt{1+r}$, one ends with
\[\begin{aligned}
\s I\;&\leqslant\;2\pi\,\expo{\lambda_\ast^2/4}\expo{-\frac{\lambda_\ast}{2}d_0}\bigl[2\sqrt\pi\lB^2(
\sqrt{2r}+\lambda_\ast)+4\lB^2\bigr]\\
&\leqslant\;K_\ast'\,\lB^2\sqrt{1+r}\;\expo{-\frac{\lambda_\ast
}{2}d_0}\,,
\end{aligned}
\]
with $K_\ast:=\sqrt2+\lambda_\ast+\sfrac{2}{\sqrt\pi}$ and $K_\ast':=4\pi^{\sfrac{3}{2}}K_\ast\expo{\sfrac{\lambda_
\ast^2}{4}}$. Bounding instead $[d_0-w/\lB]_+\geqslant0$ (\ie setting $\lambda_\ast=0$ in
the two displays above) gives $\s I\leqslant K_\ast'\lB^2\sqrt{1+r}$. As $\min\{\expo{-
\frac{\lambda_\ast}{2}d_0},1\}=\expo{-\frac{\lambda_\ast}{2}[d_0]_+}$, the two estimates combine into
\[
\s I\;\leqslant\;K_\ast'\,\lB^2\sqrt{1+r}\;\expo{-\frac{\lambda_\ast}{2}[d_0]_+}\,,\qquad d_0=\frac{
|\bl-\bz|}{\lB}-2\sqrt{2r}\,.
\]
Inserting into \eqref{eq:overlap-int}, the factor $\sqrt{1+r}$ cancels with the one in the denominators and, with
the factor $\delta_{r,r'}$, yields the claim  with 
\[
D_\ast\;:=\;\frac{C_\ast K_\ast'\lB}{\sqrt{2\pi}}\;=\;2\pi\,C_\ast\lB\Bigl(2+\lambda_\ast\sqrt2+2
\sqrt{\tfrac{2}{\pi}}\Bigr)\expo{\frac{\lambda_\ast^2}{4}}
\]
 level--independent.
\end{proof}

\subsection{Proofs and auxiliary material for Section \ref{sec:car}}
\label{ap:proof_S2}

\begin{proof}[Proof of Proposition \ref{lem:strong-cont}]
Fix $\imath_0$ and let $\imath\to \imath_0$. We first check continuity on the generators. Since
$a$ is antilinear and isometric, $\norm{a(\phi)}=\norm{\phi}$, one has
\[
\begin{aligned}
\norm{\alpha_{u_\imath}(a(\phi))-\alpha_{u_{\imath_0}}(a(\phi))}\;&
=\;\norm{a(u_\imath \phi)-a(u_{\imath_0}\phi)}
\;=\;\norm{a\bigl((u_\imath-u_{\imath_0})\phi\bigr)}\\
&=\;\norm{(u_\imath-u_{\imath_0})\phi}\;\xrightarrow\;0
\end{aligned}
\]
by strong continuity of $\imath\mapsto u_\imath$. The same holds for $a^{\dag}(\phi)$. Hence
$t\mapsto\alpha_{u_t}(B)$ is norm--continuous for every $B$ in the dense $*$--algebra
$\rr{A}_0$ of polynomials in the generators, being a finite algebraic expression of
such terms.
Let now $A\in\rr{A}$, $\epsilon>0$ and  pick
$B\in\rr{A}_0$ with $\norm{A-B}<\epsilon$. As each $\alpha_{u_\imath}$ is a
$*$--automorphism, hence  isometric, one gets
\[
\norm{\alpha_{u_\imath}(A)-\alpha_{u_\imath}(B)}\;=\;\norm{\alpha_{u_\imath}(A-B)}\;=\;\norm{A-B}\;<\; {\epsilon}
\]
uniformly in $\imath$  (and likewise at $\imath_0$). 
By continuity on $\rr{A}_0$ there is
 an open neighborhood of $\imath_0\in\bb{O}_\epsilon\subset \bb{I}$
such that
$\norm{\alpha_{u_\imath}(B)-\alpha_{u_{\imath_0}}(B)}<\epsilon$ for every
$\imath\in \bb{O}_\epsilon$. Combining,
\[
\begin{aligned}
\norm{\alpha_{u_\imath}(A)-\alpha_{u_{\imath_0}}(A)}
\;\leqslant\;&\norm{\alpha_{u_\imath}(A)-\alpha_{u_\imath}(B)}
+\norm{\alpha_{u_\imath}(B)-\alpha_{u_{\imath_0}}(B)}\\
&+\;\norm{\alpha_{u_{\imath_0}}(B)-\alpha_{u_{\imath_0}}(A)}\;<\; 3\epsilon
\end{aligned}
\]
for every
$\imath\in {\bb{O}_\epsilon}$.     Since $\epsilon$ is arbitrary small, it follows that 
$\imath\mapsto\alpha_{u_\imath}(A)$ is norm--continuous.
\end{proof}

\smallskip

\begin{proof}[Proof of Proposition \ref{prop:quasilocal-frame}]
The validity of properties (1) and (2) in Definition~\ref{def:AQL}
 has already been established above, prior to the statement of the claim.
We only need to prove the property (3), namely the density of $\rr{A}_{\mathrm{loc}}$.
The inclusion $\overline{\rr{A}_{\mathrm{loc}}}\subseteq\rr{A}$ is clear. For
the reverse one, since $\{\chi_\gamma\}_{\gamma\in\Gamma}$ is a frame of $\HH$, every
$\phi\in\HH$ is the norm limit of the finite reconstructions
$\phi_\Lambda=\sum_{\gamma\in\Lambda}\ip{S^{-1}\chi_\gamma}{\phi}\,\chi_\gamma$ along
$\Lambda\nearrow\Gamma$. As $a$ is antilinear and isometric, one gets $\norm{a(\phi)-a(\phi_\Lambda)}
=\norm{\phi-\phi_\Lambda}\to0$ with $a(\phi_\Lambda)\in\rr{A}_\Lambda$. Therefore
$a(\phi)\in\overline{\rr{A}_{\mathrm{loc}}}$,
and in turn $a^{\dag}(\phi)=a(\phi)^*\in\overline{\rr{A}_{\mathrm{loc}}}$
 for every $\phi$. Since these
generate $\rr{A}$ one ends with $\rr{A}\subseteq \overline{\rr{A}_{\mathrm{loc}}}$.
\end{proof}

\smallskip

\begin{proof}[Proof of Proposition~ \ref{prop:core-general-0}]
Since $\n{L}_{\rm g}(a_\gamma)=-a_\gamma$ and $\n{L}_{\rm g}(a_\gamma^*)=a_\gamma^*$  for every $\gamma\in\Gamma$,
one concludes that $\rr{A}_{\mathrm{loc}}\subset\rr{D}(\n{L}_{\rm g})$
by the Leibniz rule and linearity.
Let $A=a_{\gamma_1}^{\#}\cdots a_{\gamma_k}^{\#}$ be a monomial supported in $\Lambda\in
\s{P}_f(\Gamma)$ with $p$ creation and $q=k-p$ annihilation
operators. The  
on--site action of the gauge group gives $\vartheta_\theta(A)=\expo{\ii\theta(2p-k)}A$
for every $\theta\in\R$, the latter having the same support $A$. Therefore, the gauge group preserves the support of each monomial, and hence, by linearity, we obtain that 
$\vartheta_\theta(\rr{A}_\Lambda)=\rr{A}_\Lambda$ 
for every $\Lambda\in
\s{P}_f(\Gamma)$. This shows that $\rr{A}_{\mathrm{loc}}$ is invariant under the gauge group.
Therefore  \cite[Corollary~3.1.7]{Bratteli-Robinson-1}     guarantees that $\rr{A}_{\mathrm{loc}}$
is a core for $\n{L}_{\rm g}$.
\end{proof}

\smallskip

\begin{proof}[Proof of Proposition~\ref{lem:almost-comm-gr-delta}]
Let $\pi_{j}$ be the orthogonal projection of $\HH$ onto $\HH_{\Lambda_j}$, $j=1,2$, so that $\delta=\|\pi_2\pi_1\|$, and set $N:=|\Lambda_1|=\dim\HH_{\Lambda_1}$ (Lemma~\ref{lem:findim}).
The proof is organized in four steps.
{Step 1} is the \emph{cosine--sine decomposition.}
The restriction $T:=\pi_2\!\restriction_{\HH_{\Lambda_1}}:\HH_{\Lambda_1}\to\HH_{\Lambda_2}$ acts between finite--dimensional spaces and $\|T\|=\delta$. Let $\{\phi_j\}_{j=1}^{N}$ be an orthonormal basis of $\HH_{\Lambda_1}$ of eigenvectors of $T^{*}T\geqslant0$, with $T^{*}T\phi_j=\delta_j^{2}\phi_j$. Then $\delta_j=\|\pi_2\phi_j\|\in[0,\delta]$ and
\begin{equation}\label{eq:cs-orth}
\ip{\pi_2\phi_i}{\pi_2\phi_j}\;=\;\ip{\phi_i}{T^{*}T\phi_j}\;=\;\delta_j^{2}\,\delta_{i,j}\;.
\end{equation}
Let $J:=\{j\;|\;\delta_j>0\}$. By \eqref{eq:cs-orth} the vectors $\pi_2\phi_j$ with $j\in J$  are non--zero and mutually orthogonal in $\HH_{\Lambda_2}$. Hence
\begin{equation}\label{eq:n-J}
n\;:=\;|J|\;\leqslant\;\min\{|\Lambda_1|,|\Lambda_2|\}\;.
\end{equation}
For $j\in J$ set $\varrho_j:=\delta_j^{-1}\pi_2\phi_j\in\HH_{\Lambda_2}$, which is an orthonormal family by \eqref{eq:cs-orth}, and $\theta_j:=\arcsin(\delta_j)\in(0,\sfrac{\pi}{2}]$.
Again by \eqref{eq:cs-orth}, $\ip{(\unit-\pi_2)\phi_i}{(\unit-\pi_2)\phi_j}=\delta_{i,j}(1-\delta_j^{2})$. Thus, for every $j$ with $\delta_j<1$ (this includes all $j\notin J$) the vectors
\[
\psi_j\;:=\;\frac{1}{\sqrt{1-\delta_j^{2}}}\,(\unit-\pi_2)\phi_j\;\in\;\HH_{\Lambda_2}^{\perp}
\]
are orthonormal, and $\psi_j=\phi_j$ if $j\notin J$ (indeed $\pi_2\phi_j=0$). For the $j$ with $\delta_j=1$ (if any) choose $\psi_j\in\HH_{\Lambda_2}^{\perp}$ by the Gram--Schmidt algorithm, orthonormal to all the other $\psi_i$. This is possible since $\HH_{\Lambda_2}^{\perp}$ is infinite--dimensional. In this way $\{\psi_j\}_{j=1}^{N}\subset\HH_{\Lambda_2}^{\perp}$ is an orthonormal family, $\{\varrho_j\}_{j\in J}\subset\HH_{\Lambda_2}$ is orthonormal, the two families are mutually orthogonal, and
\begin{equation}\label{eq:cs}
\phi_j\;=\;\cos(\theta_j)\,\psi_j+\sin(\theta_j)\,\varrho_j\quad(j\in J)\;,\qquad \phi_j\;=\;\psi_j\quad(j\notin J)\;.
\end{equation}
This is the cosine--sine decomposition of the basis $\{\phi_j\}$ relative to $\HH_{\Lambda_2}$. Note that no vector of $\HH_{\Lambda_2}$ is needed for $j\notin J$.
Step 2 consists in implementing in  $\rr A$ an automorphism induced by a \emph{rotation of 
$\{\phi_j\}_{j=1}^N$  into $\HH_{\Lambda_2}^\perp$}.
For $j\in J$ let $\sigma_j:=\ketbra{\varrho_j}{\psi_j}-\ketbra{\psi_j}{\varrho_j}$. It is anti--self--adjoint, $\sigma_j\psi_j=\varrho_j$, $\sigma_j\varrho_j=-\psi_j$, and $\sigma_j$ vanishes on $\{\psi_j,\varrho_j\}^{\perp}$. Since $\expo{\theta\sigma_j}\psi_j=\cos(\theta)\,\psi_j+\sin(\theta)\,\varrho_j$, the operator $\sigma:=\sum_{j\in J}\theta_j\sigma_j$ satisfies, by the mutual orthogonality of the planes $\mathrm{span}\{\psi_j,\varrho_j\}$ and by \eqref{eq:cs},
\begin{equation}\label{eq:rot-one-body}
\expo{\sigma}\psi_j\;=\;\phi_j\;,\qquad\text{equivalently}\qquad \expo{-\sigma}\phi_j\;=\;\psi_j\;,\qquad j=1,\ldots,N\;.
\end{equation}
Consider the elements
$\Sigma_j:=a^{\dag}(\varrho_j)a(\psi_j)-a^{\dag}(\psi_j)a(\varrho_j)\in\rr A$, which satisfy $\Sigma_j^{*}=-\Sigma_j$, and $\Sigma:=\sum_{j\in J}\theta_j\Sigma_j$. Using $[XY,Z]=X\{Y,Z\}-\{X,Z\}Y$ and the CAR (recall $\{a(\psi),a^{\dag}(\phi)\}=\ip{\psi}{\phi}$) one finds
\[
[\Sigma_j,a^{\dag}(\phi)]\;=\;\ip{\psi_j}{\phi}\,a^{\dag}(\varrho_j)-\ip{\varrho_j}{\phi}\,a^{\dag}(\psi_j)\;=\;a^{\dag}(\sigma_j\phi)\;,\qquad\forall\;\phi\in\HH\;,
\]
and hence $[\Sigma,a^{\dag}(\phi)]=a^{\dag}(\sigma\phi)$. Let $F(s):=\expo{-s\Sigma}a^{\dag}(\expo{s\sigma}\phi)\expo{s\Sigma}$ with $s\in\R$, where $\expo{s\Sigma}$ is a unitary of $\rr A$. Then $F$ is norm--differentiable and
\[
F'(s)\;=\;\expo{-s\Sigma}\bigl(-[\Sigma,a^{\dag}(\expo{s\sigma}\phi)]+a^{\dag}(\sigma\expo{s\sigma}\phi)\bigr)\expo{s\Sigma}\;=\;0\;.
\]
This implies that $F(s)=F(0)=a^{\dag}(\phi)$, that is $\expo{s\Sigma}a^{\dag}(\phi)\expo{-s\Sigma}=a^{\dag}(\expo{s\sigma}\phi)$. Taking the adjoint (note that $\sigma^{*}=-\sigma$) gives the same relation for $a(\phi)$. With $U:=\expo{-\Sigma}\in\rr A$ and $s=-1$ this reads
\begin{equation}\label{eq:Ad-U}
U\,a^{\sharp}(\phi)\,U^{*}\;=\;a^{\sharp}(\expo{-\sigma}\phi)\;,\qquad\forall\;\phi\in\HH\;.
\end{equation}
Step 3 consists in the computation of the  \emph{graded commutation of the rotated element.}
Let $A\in\rr A_{\Lambda_1}=\rr A(\HH_{\Lambda_1})$ be homogeneous of parity $d_A$. Since $\{\phi_j\}_{j=1}^N$ is a basis of $\HH_{\Lambda_1}$, then $A=P\bigl(a^{\sharp}(\phi_1),\ldots,a^{\sharp}(\phi_N)\bigr)$ for a (non--commutative) polynomial $P$. Set $A_0:=UAU^{*}$. By \eqref{eq:Ad-U} and \eqref{eq:rot-one-body},
\[
A_0\;=\;P\bigl(a^{\sharp}(\psi_1),\ldots,a^{\sharp}(\psi_N)\bigr)\;\in\;\rr A_{\Lambda_2}^{\perp}\;:=\;\rr A(\HH_{\Lambda_2}^{\perp})\;.
\]
The element $\Sigma$ is even (each of its terms contains one creation and one annihilation operator), hence $U$ commutes with the parity automorphism and $A_0$ has the same parity $d_A$ as $A$. For $\psi\in\HH_{\Lambda_2}^{\perp}$ and $\phi\in\HH_{\Lambda_2}$ one has $\{a(\psi),a^{\dag}(\phi)\}=\ip{\psi}{\phi}=0$, so that all the generators $a^{\sharp}(\psi)$ anticommute with all the generators $a^{\sharp}(\phi)$. By the graded Leibniz rule, homogeneous elements of $\rr A_{\Lambda_2}^{\perp}$ and of $\rr A_{\Lambda_2}$ graded--commute (equivalently, $\rr A\simeq\rr A_{\Lambda_2}\hat\otimes\rr A_{\Lambda_2}^{\perp}$). Hence, for $B\in\rr A_{\Lambda_2}$ homogeneous, $[A_0,B]_{\rm gr}=0$, and therefore
\[
\bigl\|[A,B]_{\rm gr}\bigr\|\;=\;\bigl\|[A-A_0,B]_{\rm gr}\bigr\|\;\leqslant\;2\,\|A-A_0\|\,\|B\|\;,
\]
where the fact that $A$ and $A_0$ have the same parity implies that the graded commutator is linear in the first argument.
{Step 4} is the  \emph{estimate of $\|A-A_0\|$.}
By the fundamental theorem of calculus
\[
A_0-A\;=\;\int_0^1\dd s\,\frac{\dd}{\dd s}\Bigl(\expo{-s\Sigma}A\expo{s\Sigma}\Bigr)\;=\;-\int_0^1\dd s\;\expo{-s\Sigma}[\Sigma,A]\expo{s\Sigma}\;,
\]
and therefore $\|A-A_0\|\leqslant\|[\Sigma,A]\|\leqslant2\|\Sigma\|\|A\|$. We claim that $\|\Sigma_j\|\leqslant1$. Let $f_{\pm}:=\sfrac{(\psi_j\mp\ii\varrho_j)}{\sqrt2}$, an orthonormal pair with $\sigma_jf_{\pm}=\pm\ii f_{\pm}$. Expanding $\psi_j$ and $\varrho_j$ in the pair $f_\pm$ (recall that $a(\cdot)$ is antilinear) one finds $\Sigma_j=\ii\,(n_+-n_-)$ with $n_\pm:=a^{\dag}(f_\pm)a(f_\pm)$. The elements $n_\pm$ are commuting projections, since $\{a(f_+),a^{\dag}(f_-)\}=0$ and $\|f_\pm\|=1$. Thus $\|n_+-n_-\|\leqslant1$ and $\|\Sigma_j\|\leqslant1$. By the triangle inequality, the inequality $\arcsin (x)\leqslant(\sfrac{\pi}{2})\,x$ on $[0,1]$, $\delta_j\leqslant\delta$ and \eqref{eq:n-J}, one gets
\[
\|A-A_0\|\;\leqslant\;2\|A\|\sum_{j\in J}\theta_j\;\leqslant\;\pi\,\|A\|\sum_{j\in J}\delta_j\;\leqslant\;\pi\,n\,\delta\,\|A\|\;.
\]
Combining with Step 3 and \eqref{eq:n-J},
\[
\bigl\|[A,B]_{\rm gr}\bigr\|\;\leqslant\;2\pi\,\min\{|\Lambda_1|,|\Lambda_2|\}\,\|A\|\,\|B\|\,\delta\;\leqslant\;4\pi\sqrt{|\Lambda_1|\,|\Lambda_2|}\;\|A\|\,\|B\|\,\delta\;,
\]
where the last step uses $\min\{x,y\}\leqslant\sqrt{xy}$. 
\end{proof}

\smallskip

\begin{proof}[Proof of Lemma \ref{lem:delta-gram}]
(i) Let $\varphi\in\HH_{\Lambda_1}$ and $\psi\in\HH_{\Lambda_2}$ be unit vectors and consider the expansions  
$\varphi=\sum_{\gamma\in\Lambda_1}c_\gamma\chi_\gamma$ and $\psi=\sum_{\gamma'\in\Lambda_2}d_{\gamma'}\chi_{\gamma'}$. Then 
\[
1\;=\;\|\varphi\|^2_{\HH_{\Lambda_1}}\;=\;\langle c,G_{\Lambda_1}c\rangle_{\ell^2(\Lambda_1)}\;\geqslant\; m(\Lambda_1)\|c\|_{\ell^2(\Lambda_1)}^2\;.
\]
Therefore
\[
\|c\|_1\;\equiv\;\|c\|_{\ell^1(\Lambda_1)}\;\leqslant\;\sqrt{|\Lambda_1|}\,{\|c\|_{\ell^2(\Lambda_1)}}\;\leqslant\;\sqrt{\frac{|\Lambda_1|}{m(\Lambda_1)}}\;.
\]
Likewise $\|d\|_1\leqslant\sqrt{|\Lambda_2|/m(\Lambda_2)}$. By Proposition~\ref{prop:magloc} and $d(\gamma,\gamma')\geqslant d(\Lambda_1,\Lambda_2)$ for $\gamma\in\Lambda_1$ and $\gamma'\in\Lambda_2$, one gets
\[
|\langle\psi,\varphi\rangle|\;\leqslant\;\sum_{\gamma\in\Lambda_1}\sum_{\gamma'\in\Lambda_2}|d_{\gamma'}||c_\gamma|\,|\langle\chi_{\gamma'},\chi_\gamma\rangle|\;\leqslant\;G_\lambda\,\expo{-\lambda d(\Lambda_1,\Lambda_2)}\,\|c\|_1\|d\|_1\,.
\]
Since $\delta(\Lambda_1,\Lambda_2)=\|\pi_2\pi_1\|=\sup_{\varphi,\psi}|\langle\psi,\varphi\rangle|$ over such unit vectors, (i) follows.
(ii) By \eqref{eq:t-chi}, $t_\bl\chi_\gamma=\expo{\ii\phi_\gamma}\chi_{\gamma+\gamma_\bl}$ with $\phi_\gamma\in\R$. As $t_\bl$ is unitary one gets
$\langle\chi_{\gamma+\gamma_\bl},\chi_{\gamma'+\gamma_\bl}\rangle=\expo{\ii(\phi_\gamma-\phi_{\gamma'})}\langle\chi_\gamma,\chi_{\gamma'}\rangle$. This yields  $G_{\Lambda+\gamma_\bl}=UG_\Lambda U^*$ with $U$ the diagonal unitary matrix of the phases $\expo{\ii\phi_\gamma}$. Hence $G_{\Lambda+\gamma_\bl}$ and $G_\Lambda$ have the same spectrum.
\end{proof}

\smallskip

\begin{proof}[Proof of Proposition \ref{lem:almost-comm-gr-II}]
By Proposition~\ref{lem:almost-comm-gr-delta} and Lemma~\ref{lem:delta-gram}\,(i) one obtains
\eqref{eq:almost-comm-gr-II}.
The invariance of $C_\lambda$ under $\Lambda_j\mapsto\Lambda_j+\gamma_\bl$ follows from $|\Lambda+\gamma_\bl|=|\Lambda|$ and Lemma~\ref{lem:delta-gram}\,(ii).
\end{proof}

\subsection{Proofs and auxiliary material for Section \ref{Sect:frech_struct}}\label{ap:proof_S3}

\begin{proof}[Proof of Proposition \ref{prop:cond-exp}]
Let $\sigma^{\perp}\in
\Aut(\rr{A}_\Lambda^{\perp})$ be the parity automorphism of the second factor,
$\sigma^{\perp}(a(\psi))=-a(\psi)$ for $\psi\in\HH_\Lambda^{\perp}$, extended by the
identity on $\rr{A}_\Lambda$. Following \cite[eq.~(4.19)]{ArakiMoriya2003}, set
\[
\s{E}^{(1)}_\Lambda\;:=\;\frac12\bigl(\mathrm{id}+\sigma^{\perp}\bigr)\;.
\]
It is a norm--one projection onto the fixed--point algebra
$\rr{A}_\Lambda\otimes({\rr{A}_\Lambda^{\perp}})^{+}$, an \emph{ordinary} tensor product,
since projecting out the odd part of the second factor removes exactly the
anti--commuting elements. It preserves $\omega_{\mathrm{tr}}$, since
$\omega_{\mathrm{tr}}\circ\sigma^{\perp}=
\omega_{\mathrm{tr}}$. In fact 
$\omega_{\mathrm{tr}}$ is even, or equivalently it vanishes on odd monomials  (being quasi--free and 
gauge--invariant, see Section \ref{sec:Q-F}).
 On
$\rr{A}_\Lambda\otimes({\rr{A}_\Lambda^{\perp}})^{+}$ the trace $\omega_{\mathrm{tr}}$
factorizes as $\tau_\Lambda\otimes\omega_{\mathrm{tr}}^{\perp}$, with
$\tau_\Lambda$ the unique normalized trace of the matrix factor
$\rr{A}_\Lambda\cong M_{2^{|\Lambda|}}(\C)$ and
$\omega_{\mathrm{tr}}^{\perp}:=\omega_{\mathrm{tr}}\!\restriction_{({\rr{A}_\Lambda^{\perp}})^{+}}$ a faithful trace.
Define the analog of the slice map in
\cite[eq.~(4.20)]{ArakiMoriya2003} as $\s{E}^{(2)}_\Lambda:\rr{A}_\Lambda\otimes({\rr{A}_\Lambda^{\perp}})^{+}\to\rr{A}_\Lambda$  given by
\[
\s{E}^{(2)}_\Lambda\;:=\;\mathrm{id}_{\rr{A}_\Lambda}\otimes\,\omega_{\mathrm{tr}}^{\perp}\;.
\]
This is a unital, completely positive, norm--one projection onto $\rr{A}_\Lambda$, restricting to
the identity there. 
The conditional expectation is the composition
\[\n{E}_\Lambda\;:=\;\s{E}^{(2)}_\Lambda\circ\s{E}^{(1)}_\Lambda\;.
\]
For $B\in\rr{A}_\Lambda\otimes({\rr{A}_\Lambda^{\perp}})^{+}$ one has $\sigma^{\perp}(B)=B$
and, by invariance
\[
\omega_{\mathrm{tr}}(AB)\;=\;\omega_{\mathrm{tr}}(\sigma^\bot(AB))\;=\;\omega_{\mathrm{tr}}(\sigma^{\perp}(A)B)\;.
\]
Hence
$\omega_{\mathrm{tr}}(AB)=\omega_{\mathrm{tr}}(\s{E}^{(1)}_\Lambda(A)B)$ by linearity. Now let $B\in\rr{A}_\Lambda$ 
 in the first
factor. Then $\s{E}^{(1)}_\Lambda(A)B\in \rr{A}_\Lambda\otimes({\rr{A}_\Lambda^{\perp}})^{+}$ and 
$\omega_{\mathrm{tr}}$ splits as a
product state there. Therefore
\[
\begin{aligned}
\omega_{\mathrm{tr}}(\s{E}^{(1)}_\Lambda(A)B)\;&=\;\omega_{\mathrm{tr}}\big(\s{E}^{(2)}_\Lambda(\s{E}^{(1)}_\Lambda(A)B)\big)\\
&=\;{\omega_{\mathrm{tr}}\big(\s{E}^{(2)}_\Lambda(\s{E}^{(1)}_\Lambda(A))\,B\big)}\;=\;\omega_{\mathrm{tr}}\big(\n{E}_\Lambda(A)B\big)\;,
\end{aligned}
\]
{where the second equality is $\omega_{\mathrm{tr}}=\omega_{\mathrm{tr}}\circ\s E^{(2)}_\Lambda$ on $\rr A_\Lambda\otimes(\rr A_\Lambda^{\perp})^{+}$, and the third one holds because $\s E^{(2)}_\Lambda=\mathrm{id}_{\rr A_\Lambda}\otimes\omega^{\perp}_{\mathrm{tr}}$ is an $\rr A_\Lambda$--bimodule map and $B\in\rr A_\Lambda$.}
which along with the previous equation yields 
the \eqref{eq:cond-exp-defining}.
If $\n{E}'$ also satisfies \eqref{eq:cond-exp-defining}, then
$\omega_{\mathrm{tr}}((\n{E}_\Lambda(A)-\n{E}'(A))B)=0$ for all $B\in\rr{A}_\Lambda$.
Since $\n{E}_\Lambda(A)-\n{E}'(A)\in\rr{A}_\Lambda$ and
$\omega_{\mathrm{tr}}\!\restriction_{\rr{A}_\Lambda}=\tau_\Lambda$ is faithful, taking
$B=(\n{E}_\Lambda(A)-\n{E}'(A))^{*}$ forces equality. This is the uniqueness of
\cite[Lemma~2.2]{ArakiMoriya2003}, here reduced to faithfulness of the trace on the
finite matrix algebra.
Taking
$A\in\rr{A}_\Lambda$ shows $\n{E}_\Lambda|_{\rr{A}_\Lambda}=\mathrm{id}$. 
As a composition of a twirl and a slice map,
$\n{E}_\Lambda$ is a bimodule map over $\rr{A}_\Lambda$, giving \emph{(i)}. For \emph{(iii)},
the gauge automorphisms $\vartheta_\theta$ preserve $\omega_{\mathrm{tr}}$ (see Remark \ref{rem:inv_tr_0}),
the subalgebra $\rr{A}_\Lambda$ and the grading. {Therefore} $\vartheta_\theta\circ\n{E}_\Lambda\circ
\vartheta_{-\theta}$ satisfies \eqref{eq:cond-exp-defining} and, by uniqueness, it  equals
$\n{E}_\Lambda$. Hence $\n{E}_\Lambda$ commutes with $\vartheta_\theta$ and preserves the
gauge--invariant subalgebra.
Finally, let $\Lambda_1\subseteq\Lambda_2$. Then $\n{E}_{\Lambda_1}(\n{E}_{\Lambda_2}(A))\in\rr{A}_{\Lambda_1}$ and, for every
$B\in\rr{A}_{\Lambda_1}\subseteq\rr{A}_{\Lambda_2}$, applying \eqref{eq:cond-exp-defining} twice yields
$\omega_{\mathrm{tr}}(AB)=\omega_{\mathrm{tr}}(\n{E}_{\Lambda_2}(A)B)
=\omega_{\mathrm{tr}}(\n{E}_{\Lambda_1}(\n{E}_{\Lambda_2}(A))B)$. By uniqueness,
$\n{E}_{\Lambda_1}\circ\n{E}_{\Lambda_2}=\n{E}_{\Lambda_1}$. Conversely, $\n{E}_{\Lambda_1}(A)\in\rr{A}_{\Lambda_2}$ and
$\n{E}_{\Lambda_2}$ is the identity on $\rr{A}_{\Lambda_2}$, so $\n{E}_{\Lambda_2}\circ\n{E}_{\Lambda_1}=\n{E}_{\Lambda_1}$.
\end{proof}

\smallskip

\begin{proof}[Proof of Lemma \ref{lem:E-gen}]
Write $a(\phi)=a(\pi_\Lambda\phi)+a\bigl((\unit-\pi_\Lambda)\phi\bigr)$. The first term belongs to $\rr A_\Lambda$, hence is
fixed by $\n E_\Lambda$.
The second term is a generator of $\rr A_\Lambda^\perp$, hence odd under $\sigma^\perp$, namely  $\sigma^\perp(a((\unit-\pi_\Lambda)\phi))
=-a((\unit-\pi_\Lambda)\phi)$. Therefore,  it is annihilated by $\s E^{(1)}_\Lambda=(\sfrac12)(\mathrm{id}+\sigma^\perp)$. This gives the first identity.
The second follows by taking adjoints observing that $\n E_\Lambda$ is $\ast$--preserving, being completely positive. Finally
$\|(\mathrm{id}-\n E_\Lambda)a^\sharp(\phi)\|=\|a^\sharp((\unit-\pi_\Lambda)\phi)\|=\norm{(\unit-\pi_\Lambda)\phi}$ by the CAR isometry.
\end{proof}

\smallskip

\proof[Proof of Lemma \ref{lemm:eq_npr_p}]
To show that $\rr{A}_p$ does not depend on the choice of base point, it is enough to
prove that for any $A\in\rr{A}$ and any pair $\gamma_0,\gamma_1\in\Gamma$
\begin{equation}\label{eq_bound}
\norm{A}_{p,\gamma_1}\;\leqslant\;3\,(1+R)^{p}\,\norm{A}_{p,\gamma_0}\;,\qquad
R\;:=\;{\lceil d(\gamma_0,\gamma_1)\rceil\in\N_0}
\end{equation}
where $\lceil x\rceil :={\min}\{n\in\Z\;|\;n\geqslant x\}$ is the \emph{ceiling} of $x$ (the smallest integer not smaller than $x$).
By the triangle inequality for $d$  (and $R\geqslant d(\gamma_0,\gamma_1)$), every ball centered at $\gamma_1$ contains a smaller
ball centered at $\gamma_0$,
\[
\Lambda_{k-R}(\gamma_0)\;\subseteq\;\Lambda_k(\gamma_1)\;,\qquad k\geqslant R\;.
\]
The tower property for nested regions in Proposition~\ref{prop:cond-exp} (ii) gives
 $\n{E}_{\Lambda_k(\gamma_1)}\circ\n{E}_{\Lambda_{k-R}(\gamma_0)}=
\n{E}_{\Lambda_{k-R}(\gamma_0)}$, hence
\[
A-\n{E}_{\Lambda_k(\gamma_1)}(A)\;=\;\bigl(A-\n{E}_{\Lambda_{k-R}(\gamma_0)}(A)\bigr)
-\n{E}_{\Lambda_k(\gamma_1)}\bigl(A-\n{E}_{\Lambda_{k-R}(\gamma_0)}(A)\bigr)\;.
\]
Since $\n{E}_{\Lambda_k(\gamma_1)}$ is a norm--one projection,
$\norm{\n{E}_{\Lambda_k(\gamma_1)}}\leqslant1$, so
\[
\bigl\|A-\n{E}_{\Lambda_k(\gamma_1)}(A)\bigr\|
\;\leqslant\;2\,\bigl\|A-\n{E}_{\Lambda_{k-R}(\gamma_0)}(A)\bigr\|\;.
\]
Using $(1+k)\leqslant(1+R)\bigl(1+(k-R)\bigr)$ for $k\geqslant R$ and weighting by
$(1+k)^{p}$,
\[
\begin{aligned}
\bigl\|A-\n{E}_{\Lambda_k(\gamma_1)}(A)\bigr\|\,(1+k)^p
&\;\leqslant\;2(1+R)^p\,\bigl\|A-\n{E}_{\Lambda_{k-R}(\gamma_0)}(A)\bigr\|\,
\bigl(1+(k-R)\bigr)^p\\
&\;\leqslant\;2(1+R)^p\,\norm{A}_{p,\gamma_0}\;.
\end{aligned}
\]
For the finitely many $k<R$, the crude bound
$\bigl\|A-\n{E}_{\Lambda_k(\gamma_1)}(A)\bigr\|\,(1+k)^{p}\leqslant2\norm{A}(1+R)^{p}
\leqslant2(1+R)^{p}\norm{A}_{p,\gamma_0}$ holds, using $\norm{A}\leqslant
\norm{A}_{p,\gamma_0}$. Taking the supremum over ${k\in\N_0}$ therefore gives
$\sup_k(\cdots)\leqslant2(1+R)^p\norm{A}_{p,\gamma_0}$, and adding
$\norm{A}\leqslant\norm{A}_{p,\gamma_0}\leqslant(1+R)^p\norm{A}_{p,\gamma_0}$ yields
\eqref{eq_bound}. Exchanging $\gamma_0$ and $\gamma_1$ gives the reverse bound, so the two norms are equivalent and
$\rr A_p$ is independent of the base point.
 Fix $\gamma_0=0$ and abbreviate $\norm{\,\cdot\,}_p:=\norm{\,\cdot\,}_{p,0}$. The functional $\norm{\,\cdot\,}_p$ is a genuine norm, since $\norm{A}_p\geqslant\norm{A}$
forces $\norm{A}_p=0$ and in turn $A=0$. Let $\{A_n\}_n$ be $\norm{\,\cdot\,}_p$--Cauchy. As $\norm{\,
\cdot\,}\leqslant\norm{\,\cdot\,}_p$, it is also $\norm{\,\cdot\,}$--Cauchy, so $A_n\to A$ in the $C^{*}
$--algebra $\rr A$ for some $A\in\rr A$. 
Throughout write $\n{E}_k:=\n{E}_{\Lambda_k(0)}$.
Each conditional expectation $\n E_{k}$
is $\norm{\,\cdot\,}$--continuous, so for every $k$, $(\mathrm{id}-\n E_k)(A_n)
\to(\mathrm{id}-\n E_k)(A)$ in norm. Given $\varepsilon>0$, choose $N$ with
$\norm{A_n-A_m}_p<\varepsilon$ for $n,m\geqslant N$. Then, for every $k$,
\[
(1+k)^p\bigl\|(\mathrm{id}-\n E_{k})(A_n-A_m)\bigr\|\;\leqslant\;\varepsilon\;.
\]
Letting $m\to\infty$ and then taking the supremum on $k$  gives 
\[
\sup_k(1+k)^p\|(\mathrm{id}-\n E_{k})(A_n-A)\|\;\leqslant\;\varepsilon
\] 
for $n\geqslant N$. Adding $\norm{A_n-A}\leqslant
\varepsilon$ shows $\norm{A_n-A}_p\to0$. In particular, from $\norm{A}_p\leqslant\norm{A_N}_p+
\varepsilon<\infty$ one has  $A\in\rr A_p$ and $A_n\to A$  in  $\rr A_p$. Hence $(\rr A_p,\norm{\,\cdot\,}
_p)$ is complete.
 The involution is isometric. In fact
the conditional expectation is
$\ast$--preserving, $\n{E}_k(A^{*})=\n{E}_k(A)^{*}$, so
\[
\norm{A^{*}-\n{E}_k(A^{*})}\;=\;\norm{(A-\n{E}_k(A))^{*}}\;=\;
\norm{A-\n{E}_k(A)}\;,
\] whence
$\norm{A^{*}}_p=\norm{A}_p$. For the product, put $A_k:=\n{E}_k(A)$, $B_k:=\n{E}_k(B)$ both in $
\rr{A}_{\Lambda_k}$. Then $A_kB_k\in\rr{A}_{\Lambda_k}$, so $\n{E}_k(A_kB_k)=A_kB_k$ and
\[
AB-\n{E}_k(AB)\;=\;\bigl(AB-A_kB_k\bigr)-\n{E}_k\bigl(AB-A_kB_k\bigr)\;.
\]
Using $\norm{\n{E}_k}\leqslant1$, $\norm{A_k}\leqslant\norm{A}$ and
$AB-A_kB_k=(A-A_k)B+A_k(B-B_k)$ one gets
\[
\norm{AB-\n{E}_k(AB)}\leqslant2\norm{AB-A_kB_k}
\;\leqslant\;2\bigl(\norm{A-A_k}\,\norm{B}+\norm{A}\,\norm{B-B_k}\bigr)\;.
\]
Multiplying by $(1+k)^{p}$ and using
$\norm{A-A_k}(1+k)^{p}\leqslant\norm{A}_p$, $\norm{A}\leqslant\norm{A}_p$ and the same for $B_k$, one gets
\[
\norm{AB-\n{E}_k(AB)}\,(1+k)^{p}\;\leqslant\;4\,\norm{A}_p\norm{B}_p\;.
\]
Taking the supremum over $k$ and adding $\norm{AB}\leqslant\norm{A}\,\norm{B}\leqslant
\norm{A}_p\norm{B}_p$ gives
\[
\norm{AB}_p\;\leqslant\;5\,\norm{A}_p\norm{B}_p\;.
\]
Thus $\norm{\,\cdot\,}_p$ is submultiplicative 
with the pre-factor $C=5$ uniform in $p$. 
Equivalently, the rescaled
norm $|||A|||:=5\norm{A}_p$ satisfies $|||AB|||\leqslant|||{A}|||\,|||{B}|||$, so $(\rr A
_p,|||{\,\cdot\,}|||)$ is a genuine Banach $\ast$--algebra {and $\rr A_p$ is a Banach $\ast$--algebra under the equivalent norm $|||\cdot|||=5\norm{\cdot}_p$, as stated.} 
\qed

\smallskip

\begin{proof}[Proof of Theorem \ref{lem:loc-in-Ainfty}]
Throughout write $\n{E}_k:=\n{E}_{\Lambda_k}$. One has that $\Lambda_k\subseteq\Lambda_{k+1}$
and $\bigcup_{k}\Lambda_k=\Gamma$, and also
$\n{E}_k\circ\n{E}_n=\n{E}_{\min(k,n)}$
by the tower property in Proposition~\ref{prop:cond-exp}. Each $\norm{\cdot}_p$ is a genuine norm, since
$\norm{A}_p\geqslant\norm{A}$. Moreover, $\norm{A}_p\leqslant\norm{A}_q$ for $p\leqslant q$.
The proof of the isometry of the involution $\norm{A^{*}}_p=\norm{A}_p$, and of the
 submultiplicativity $\norm{AB}_{p}\leqslant5\,\norm{A}_{p}\,\norm{B}_{p}$
is contained in the proof of  Lemma \ref{lemm:eq_npr_p}.
In particular one infers that $\rr{A}_\infty$ is a $\ast$--subalgebra of
$\rr{A}$, the adjoint is continuous and the multiplication is jointly continuous with respect to the family of norms.
  Let
$A\in\rr{A}_{\Lambda_0}$ with $\Lambda_0\in\s{P}_f(\Gamma)$, and set
$k_0:=\max_{\gamma\in\Lambda_0}\norm{\gamma}_1<\infty$. For $k\geqslant k_0$ one has
$\Lambda_0\subseteq\Lambda_k$, hence $\rr{A}_{\Lambda_0}\subseteq\rr{A}_{\Lambda_k}$ and
$\n{E}_k(A)=A$, so $\norm{A-\n{E}_k(A)}=0$. The supremum in the definition of
$\norm{A}_p$ runs then over the finite set $\{k<k_0\}$, on which
$\norm{A-\n{E}_k(A)}\leqslant2\norm{A}$ and 
$(1+k)^{p}\leqslant(1+k_0)^{p}$. Thus
$\norm{A}_p\leqslant\bigl(1+2(1+k_0)^{p}\bigr)\norm{A}<\infty$ for every $p$. Then
$A\in\rr{A}_\infty$ which shows that $\rr{A}_{\mathrm{loc}}\subset\rr{A}_\infty$.
Let us prove the completeness in the Fr\'echet topology.
Let $\{A_n\}_{n\in\N}$ be Cauchy for every
$\norm{\cdot}_p$. Since $\norm{\cdot}\leqslant\norm{\cdot}_0$, it is Cauchy in the
$C^{*}$--norm, so $A_n\to A$ in $\rr{A}$. Fix $p$ and $\varepsilon>0$, and $N$ with
$\norm{A_n-A_m}_p<\varepsilon$ for $n,m\geqslant N$. Then for each fixed $k$ one has
\[\norm{(A_n-A_m)-\n{E}_k(A_n-A_m)}(1+k)^{p}\;<\;\varepsilon\;.
\]
Since $B\mapsto B-\n{E}_k(B)$
is $C^{*}$--norm continuous, letting $m\to\infty$ gives
\[
\norm{(A_n-A)-\n{E}_k(A_n-A)}(1+k)^{p}\;\leqslant\;\varepsilon
\] 
for every $k$. So the
supremum over $k$ is bounded by $\varepsilon$. Together with
$\norm{A_n-A}=\lim_m\norm{A_n-A_m}\leqslant\varepsilon$ this yields
$\norm{A_n-A}_p\leqslant2\varepsilon$ for $n\geqslant N$. Hence $A_n\to A$ in every
$\norm{\cdot}_p$, and $\norm{A}_p\leqslant\norm{A_N}_p+\norm{A-A_N}_p<\infty$ for all
$p$, so $A\in\rr{A}_\infty$. Thus $\rr{A}_\infty$, with the countable increasing family
of norms $\norm{\cdot}_p$, is a complete metrizable locally convex space,
i.e.\ a Fr\'echet space. In view of  the continuity of the algebraic operation,  it is a Fr\'echet $\ast$--algebra.
The last step is the proof of the density of the local structure.
 We show the local
truncations $\n{E}_n(A)\in\rr{A}_{\Lambda_n}\subset\rr{A}_{\mathrm{loc}}$ converge to $A$
in every $\norm{\cdot}_p$. Put $B_n:=A-\n{E}_n(A)$. By the tower property,
$\n{E}_k(B_n)=\n{E}_k(A)-\n{E}_{\min(k,n)}(A)$. 
\begin{itemize}
\item If $k\geqslant n$, then $B_n-\n{E}_k(B_n)=A-\n{E}_k(A)$, and
\[
\norm{B_n-\n{E}_k(B_n)}(1+k)^{p}\;=\;\norm{A-\n{E}_k(A)}(1+k)^{p}
\;\leqslant\;\frac{\norm{A}_{p+1}}{1+k}\;\leqslant\;\frac{\norm{A}_{p+1}}{1+n}\;.
\]
\item If $k<n$, then $\n{E}_k(B_n)=0$, and
\[
\begin{aligned}
\norm{B_n-\n{E}_k(B_n)}(1+k)^{p}\;&=\;\norm{A-\n{E}_n(A)}(1+k)^{p}
\\
&\leqslant\;\norm{A-\n{E}_n(A)}(1+n)^{p}\;\leqslant\;\frac{\norm{A}_{p+1}}{1+n}\;.
\end{aligned}
\]
\end{itemize}
Since 
$\norm{B_n}=\norm{A-\n{E}_n(A)}\leqslant\norm{A}_{p+1}(1+n)^{-(p+1)}\leqslant
\norm{A}_{p+1}(1+n)^{-1}$, one obtains the rate estimate
\[
\norm{A-\n{E}_n(A)}_p\;\leqslant\;\frac{2}{1+n}\,\norm{A}_{p+1}\;
\xrightarrow[n\to\infty]{}\;0\;.
\]
Hence every $A\in\rr{A}_\infty$ is a Fr\'echet limit of local observables, so
$\rr{A}_\infty\subseteq\overline{\rr{A}_{\mathrm{loc}}}^{\;\rm Fr}$ where the right-hand side denotes the Fr\'echet  closure. The reverse inclusion holds
because $\rr{A}_{\mathrm{loc}}\subseteq\rr{A}_\infty$  and $\rr{A}_\infty$ is
complete, so it contains the closure of any of its subsets. Therefore
$\rr{A}_\infty=\overline{\rr{A}_{\mathrm{loc}}}^{\;\rm Fr}$. \end{proof}

\smallskip

\begin{proof}[Proof of Lemma \ref{lemm:aut-fr-lattice}]
Being a $\ast$--automorphism, $\alpha_\bl$ is
isometric, $\norm{\alpha_\bl(X)}=\norm{X}$ for all $X\in\rr A$.
The automorphism $\alpha_\bl$ shifts the spatial part of
the frame index by $\bl$, hence maps generators supported in $\Lambda$ to generators supported in
$\Lambda+\gamma_\bl$, so that
\[
\alpha_\bl\bigl(\rr A_\Lambda\bigr)\;=\;\rr A_{\Lambda+\gamma_\bl}\,,\qquad\forall\,\Lambda\in\s P_f(
\Gamma)\,.
\]
Since $\alpha_\bl$ preserves the tracial state $\omega_{\rm tr}$ (see Remark \ref{rem:inv_tr_0}) and $\n E_\Lambda$ is the unique
$\omega_{\rm tr}$--preserving conditional expectation onto $\rr A_\Lambda$, the family $\{\n E_\Lambda
\}_{\Lambda\in\s
P_f(\Gamma)}$ is covariant. {In fact, Lemma \ref{lem:dyn-commutes} applied with $\Lambda'=\Lambda+\gamma_\bl$ gives} 
\begin{equation}\label{eq:cov-cond-exp}
\alpha_\bl\circ\n E_\Lambda\;=\;\n E_{\Lambda+\gamma_\bl}\circ\alpha_\bl\,,\qquad\forall\,\Lambda\in\s
P_f(\Gamma)\,
\end{equation}
for every $\Lambda\in\s
P_f(\Gamma)$ and ${\bl}\in\s L$.
The metric \eqref{eq:Xi-22} is invariant under the lattice
translation $\gamma\mapsto\gamma+\gamma_\bl$, namely $d(\gamma+\gamma_\bl,\gamma'+\gamma_\bl)=d(\gamma,
\gamma')$. Therefore
\[
\Lambda_k(\gamma_0)+\gamma_\bl\;=\;\Lambda_k(\gamma_0+\gamma_\bl)\,,\qquad\forall\,k\,\in\N_0\,,\;\;\; \forall\,\gamma_0
\in\Gamma\,.
\]
Apply \eqref{eq:cov-cond-exp} with $\Lambda=\Lambda_k(\gamma_0-\gamma_\bl)$ and using
$\Lambda_k(\gamma_0-\gamma_\bl)+\gamma_\bl=\Lambda_k(\gamma_0)$,
\[
\n E_{\Lambda_k(\gamma_0)}\bigl(\alpha_\bl(A)\bigr)\;=\;\alpha_\bl\Bigl(\n E_{\Lambda_k(\gamma_0-\gamma
_\bl)}(A)\Bigr)\,.
\]
By isometry of $\alpha_\bl$,
\[
\begin{aligned}
\bigl\|\alpha_\bl(A)-\n E_{\Lambda_k(\gamma_0)}\bigl(\alpha_\bl(A)\bigr)\bigr\|\;&=\;\bigl\|\alpha_\bl
\bigl(A-\n E_{\Lambda_k(\gamma_0-\gamma_\bl)}(A)\bigr)\bigr\|\\&=\;\bigl\|A-\n E_{\Lambda_k(\gamma_0-
\gamma_\bl)}(A)\bigr\|\,.
\end{aligned}
\]
Multiplying by $(1+k)^{p}$, taking the supremum over $k\in\N_0$, and adding the term $\norm{\alpha_\bl
(A)}=\norm{A}$, we obtain \eqref{eq:lattice-cov}.
If $A\in\rr A_\infty$ then $\norm{A}_{p,\gamma}<\infty$ for all $p\in\N_0$ and
$\gamma\in\Gamma$. By \eqref{eq:lattice-cov} the same holds for $\alpha_\bl(A)$, so $\alpha_\bl(\rr A_
\infty)\subseteq\rr A_\infty$. Applying the same to $\alpha_\bl^{-1}=\alpha_{-\bl}$ gives the reverse
inclusion, whence $\alpha_\bl(\rr A_\infty)=\rr A_\infty$.
In view of  Theorem \ref{lem:loc-in-Ainfty}
the Fr\'echet topology of $\rr A_\infty$ is  generated by the norms $\{\norm{\,\cdot\,}_{p}\}_{p\in\N_0}$ (anchored at $\gamma=0$).
Identity \eqref{eq:lattice-cov}, and the equivalence of the norms  
anchored at distinct lattice points (Lemma \ref{lemm:eq_npr_p}, in particular  equation \eqref{eq_bound}) imply
\[
\norm{\alpha_\bl(A)}_{p}\;=\;\norm{A}_{p,\gamma_{-\bl}}
\leqslant\;3\bigl(1+{\lceil\lB^{-1}|\bl|_1\rceil}\bigr)^p\norm{A}_{p}\;.
\]
This shows that $\alpha_\bl$ is continuous inside each $\rr{A}_p$, hence continuous on  $\rr{A}_\infty$. The same holds for the inverse  $\alpha_\bl^{-1}=\alpha_{-\bl}$. This shows that 
$\alpha_\bl\in
\mathrm{Aut}(\rr A_\infty)$ in the category of Fr\'echet $\ast$--algebras, for every $\bl\in \s{L}$.
 \end{proof}

\smallskip

\begin{proof}[Proof of lemma \ref{cor:asympt-ab-lattice}]
Let $A\in\rr A_{\Lambda_1}$ and $B\in\rr A_{\Lambda_2}$ with $\Lambda_1,\Lambda_2\in\s{P}_f(\Gamma)$.
Recall from \eqref{eq:lattice-cov-main} that
$\alpha_\bl(\rr A_{\Lambda_1})=\rr A_{\Lambda_1+\gamma_\bl}$ with $\gamma_\bl:=(\bl,0)$, and that
$|\Lambda_1+\gamma_\bl|=|\Lambda_1|$ for every $\bl\in\s L$. Being a $\ast$--automorphism, $\alpha_\bl$ is isometric, $\norm{\alpha_\bl(A)}=\norm A$. Being a
Bogoliubov automorphism, it preserves the parity decomposition of $\rr A$, so $\alpha_\bl(A)\in\rr
A_{\Lambda_1+\gamma_\bl}$ is homogeneous of the same parity as $A$ (see Remark \ref{rk:B-A-parity}). 
Proposition~\ref{lem:almost-comm-gr-II}
applies directly to the pair $\alpha_\bl(A)\in\rr A_{\Lambda_1+\gamma_\bl}$, $B\in\rr A_{\Lambda_2}$, giving
\[
\bigl\|[\alpha_\bl(A),B]_{\rm gr}\bigr\|\;\leqslant\;C_\lambda(\Lambda_1+\gamma_\bl,\Lambda_2)\,
\norm{\alpha_\bl(A)}\,\norm B\,\expo{-\lambda\,d(\Lambda_1+\gamma_\bl,\Lambda_2)}\,.
\]
{By the translation invariance of $C_\lambda$ in Proposition~\ref{lem:almost-comm-gr-II} one has that}
 $C_\lambda(\Lambda_1+\gamma_\bl,\Lambda_2)=C_\lambda(\Lambda_1,\Lambda_2)$ for
every $\bl\in\s L$. Combined with $\norm{\alpha_\bl(A)}=\norm A$ this gives 
\begin{equation}\label{eq:dysp_ine}
\bigl\|[\alpha_\bl(A),B]_{\rm gr}\bigr\|\;\leqslant\;C_\lambda(\Lambda_1,\Lambda_2)\,
\norm{A}\,\norm B\,\expo{-\lambda\,d(\Lambda_1+\gamma_\bl,\Lambda_2)}\,.
\end{equation}
For the limit, by \eqref{eq:Xi-22} and the triangle inequality,
\[
d(\Lambda_1+\gamma_\bl,\Lambda_2)\;\geqslant\;\tfrac1{\lB}|\bl|_1-{\rm diam}(\Lambda_1)-{\max_{\gamma'\in\Lambda_2}d(\gamma_0,\gamma')}
\;\xrightarrow[\ |\bl|\to\infty\ ]{}\;\infty\,,
\]
for any fixed reference point $\gamma_0\in\Lambda_1$. Therefore the right-hand side of \eqref{eq:dysp_ine} tends exponentially to $0$ for every \revtxtmk{}$\lambda>0$.
\end{proof}

\smallskip

\begin{proof}[Proof of Proposition \ref{cor:asympt-ab-full}]
Let $\sigma=\vartheta_\pi$ be the parity $\ast$--automorphism and
$
P_\pm:=\sfrac12(\mathrm{id}\pm\sigma)
$
the associated contractive projections onto $\rr A^\pm$. Indeed $P_\pm^2=P_\pm$, $P_++P_-=\mathrm{id}$, $P_\pm(X)=X$ if and only $X$ is homogeneous of the corresponding parity, and $\norm{P_\pm}\leqslant1$. 
As recalled in the proof of Proposition~\ref{prop:core-general-0}, the on-site action of the gauge group gives $\vartheta_\theta(\rr A_\Lambda)=\rr A_\Lambda$ for every $\theta\in\R$ and $\Lambda\in\s P_f(\Gamma)$.
 Specializing to $\theta=\pi$ gives $\sigma(\rr A_\Lambda)=\rr A_\Lambda$, hence $P_\pm(\rr A_\Lambda)\subseteq\rr A_\Lambda$, and therefore $P_\pm(\rr A_{\rm loc})\subseteq\rr A_{\rm loc}$. Fix homogeneous $A,B\in\rr A$, of parities $d_A,d_B\in\{0,1\}$, and $\epsilon>0$. By Proposition~\ref{prop:quasilocal-frame}, $\rr A_{\rm loc}$ is dense in $\rr A$, so choose $A'',B''\in\rr A_{\rm loc}$ with $\norm{A-A''}<\epsilon$ and $\norm{B-B''}<\epsilon$. These need not be homogeneous, so project them onto the correct parity,
 \[
 A'\;:=\;P_{{s_A}}(A'')\in\rr A_{\rm loc}\,,\qquad B'\;:=\;P_{{s_B}}(B'')\in\rr A_{\rm loc}\,,
 \]
 with ${s_A}:=(-1)^{d_A}$, and likewise for $B$. 
 Using $P_{{s_A}}(A)=A$ (as $A$ is homogeneous of parity $d_A$) and the contractivity of $P_{{s_A}}$,
 \[
 \norm{A-A'}\;=\;\bigl\|P_{{s_A}}(A-A'')\bigr\|\;\leqslant\;\norm{A-A''}\;<\;\epsilon\,,
 \]
 and likewise $\norm{B-B'}<\epsilon$.
In particular $\norm{A'}\leqslant\norm A+\epsilon$ and $\norm{B'}\leqslant\norm B+\epsilon$. By construction $A'\in\rr A_{\Lambda_1}$ and $B'\in\rr A_{\Lambda_2}$ for some $\Lambda_1,\Lambda_2\in\s P_f(\Gamma)$, homogeneous of parities $d_A,d_B$ respectively. Since $A',A-A'$ share the parity $d_A$, and $B',B-B'$ the parity $d_B$, every graded bracket occurring below is between homogeneous elements of fixed parity. Therefore bilinearity of the graded commutators applies with a consistent sign. Writing $A=A'+(A-A')$, $B=B'+(B-B')$, expanding, and using that $\alpha_\bl$ is isometric (being a $\ast$--automorphism of $\rr A$),
\[\begin{aligned}
\bigl\|[\alpha_\bl(A),B]_{\rm gr}\bigr\|\;&\leqslant\;\bigl\|[\alpha_\bl(A'),B']_{\rm gr}\bigr\|+\bigl\|[\alpha_\bl(A-A'),B']_{\rm gr}\bigr\|\\&\phantom{\leqslant\;\;}+\bigl\|[\alpha_\bl(A'),B-B']_{\rm gr}\bigr\|+\bigl\|[\alpha_\bl(A-A'),B-B']_{\rm gr}\bigr\|\\&\leqslant\;\bigl\|[\alpha_\bl(A'),B']_{\rm gr}\bigr\|+2\norm{A-A'}\,\norm{B'}+2\norm{A'}\,\norm{B-B'}\\&\phantom{\leqslant\;\;}+2\norm{A-A'}\,\norm{B-B'}\\&\leqslant\;\bigl\|[\alpha_\bl(A'),B']_{\rm gr}\bigr\|+2\epsilon\bigl(\norm{A'}+\norm{B'}+\epsilon\bigr)\,.
\end{aligned}
\]
Since $A'\in\rr A_{\Lambda_1}$, $B'\in\rr A_{\Lambda_2}$ are homogeneous \emph{local} elements, Lemma~\ref{cor:asympt-ab-lattice} applies to the first term, so $\|[\alpha_\bl(A'),B']_{\rm gr}\|\to0$ as $\bl\in\s L$, $|\bl|\to\infty$. Hence, using $\norm{A'}\leqslant\norm A+\epsilon$, $\norm{B'}\leqslant\norm B+\epsilon$,\[\limsup_{\substack{\bl\in\s L\\|\bl|\to\infty}}\bigl\|[\alpha_\bl(A),B]_{\rm gr}\bigr\|\;\leqslant\;2\epsilon\bigl(\norm A+\norm B+3\epsilon\bigr)\,.\]As $\epsilon>0$ is arbitrary, the left-hand side vanishes, proving the statement for homogeneous $A,B$. 
\end{proof}

\smallskip

\begin{proof}[Proof of Lemma \ref{lemma_trasl_loc}]
In view of the decomposition \eqref{eq:dec_trasl} and Lemma \ref{lemm:aut-fr-lattice} we can assume that  $|\bz|<\ell_0<+\infty$. Since $\alpha_\bz$ is the Bogoliubov automorphism of the one--particle translation
$t_{\bz}$ one has $\alpha_\bz(a_\gamma)=a(t_{\bz}\chi_\gamma)$. The action of $t_{\bz}$ on the frame is
encoded by the frame matrix
\begin{equation}\label{eq:Tz-matrix}
t_{\bz}\chi_\gamma=\sum_{\gamma'\in\Gamma}[t_{\bz}]_{\gamma',\gamma}\,\chi_{\gamma'}\,,\qquad[t_{\bz}]_{
\gamma',\gamma}:=\ip{S^{-1}\chi_{\gamma'}}{t_{\bz}\chi_\gamma}\,.
\end{equation}
 Since $t_{\bz}$ and $S^{-1}$ commute with $h_B$
\textup{(}Proposition~\ref{pr:comm}\textup{)}, $t_{\bz}\chi_{(\bl,r)}$ and $S^{-1}\chi_{(\bl',r')}$ lie
in the level--$r$ and level--$r'$ eigenspaces, so \eqref{eq:Tz-matrix} is block--diagonal in the
energy levels, \ie $[t_{\bz}]_{\gamma',\gamma}=0$ for $r'\neq r$. Within a level, the magnetic
translation of a coherent state is again a coherent state at the shifted guiding center, $t_{\bz}\chi_{
(\bl,r)}=\expo{\ii\varphi}\chi_{(\bl+\bz,r)}$, the phase dropping out of the modulus. Hence $[t_{\bz}]_{
\gamma',\gamma}$ is estimated exactly by dual--coherent overlap of Lemma~\ref{lem:dual-coh-overlap} evaluated
at guiding center $\bl+\bz$, so that
\[
\bigl|[t_{\bz}]_{\gamma',\gamma}\bigr|\;=\;\bigl|\ip{S^{-1}\chi_{\gamma'}}{\chi_{(\bl+\bz,r)}}\bigr|\;
\leqslant\;\delta_{r,r'}\,D_\ast\,\expo{-\frac{\lambda_\ast}{2}\bigl[\frac{|\bl'-(\bl+\bz)|}{\lB}-2
\sqrt{2r}\bigr]_+}\,.
\]
We transfer this to the frame metric. From $|\bz|<\ell_0$ one has $|\bl'-(\bl+\bz)|\geqslant|\bl'-\bl
|-\ell_0$. At fixed level $r'=r$, $d(\gamma,\gamma')=\lB^{-1}|\bl-\bl'|_1\leqslant\sqrt2\,\lB^{-1}
|\bl-\bl'|$, whence
\[
\begin{aligned}
\frac{|\bl'-(\bl+\bz)|}{\lB}-2\sqrt{2r}\;&\geqslant\;\frac{d(\gamma,\gamma')}{\sqrt2}-\frac{\ell_0}{\lB
}-2\sqrt{2r}\\
&=\;\frac{1}{\sqrt2}\Bigl(d(\gamma,\gamma')-4\sqrt r-\sqrt2\,\tfrac{\ell_0}{\lB}\Bigr)\,.
\end{aligned}
\]
Taking positive parts yields the \emph{level--uniform} exponential banding
\begin{equation}\label{eq:Tz-banded}
\bigl|[t_{\bz}]_{\gamma',\gamma}\bigr|\;\leqslant\;\delta_{r,r'}\,D_\ast\,\expo{-\omega_\ast(d(\gamma,\gamma
')-\kappa_r)_+}\,,
\end{equation}
where 
\[
\omega_\ast\;:=\;\frac{\lambda_\ast}{2\sqrt2}\,,\qquad\kappa_r\::=\;4\sqrt r+\frac{\sqrt2
\,\ell_0}{\lB}\,,
\]
with $D_\ast,\omega_\ast$ \emph{independent} of the level $r$ and of $\bz$, and $\kappa_r=O(\sqrt r)$.
By \eqref{eq:Tz-matrix} and the anti--linearity of
$a(\cdot)$, $\alpha_\bz(a_\gamma)=\sum_{\gamma'\in\Gamma}\overline{[t_{\bz}]_{\gamma',\gamma}}\,a_{
\gamma'}$. Split this off inside and outside $\Lambda_k(\gamma)$. Namely let $\alpha_\bz(a_\gamma)=X_{\rm in}+X_{\rm out}$ with
\[
 X_{\rm in}\;:=\;\!\!\sum_{\gamma'\in\Lambda_k(\gamma)}
\!\!\overline{[t_{\bz}]_{\gamma',\gamma}}\,a_{\gamma'}\,,\qquad X_{\rm out}\;:=\;\!\!\sum_{\gamma'\notin
\Lambda_k(\gamma)}\!\!\overline{[t_{\bz}]_{\gamma',\gamma}}\,a_{\gamma'}\,.
\]
The inside part is a single mode, $X_{\rm in}=a(\psi_{\rm in})$ with $\psi_{\rm in}:=\sum_{\gamma'\in
\Lambda_k(\gamma)}[t_{\bz}]_{\gamma',\gamma}\,\chi_{\gamma'}\in\HH_{\Lambda_k(\gamma)}$. Therefore  $X_{\rm in}\in\rr
A_{\Lambda_k(\gamma)}$, and consequently $\n E_{\Lambda_k(\gamma)}(X_{\rm in})=X_{\rm in}$. This implies
$\alpha_\bz(a_\gamma)-\n E_{\Lambda_k(\gamma)}(\alpha_\bz(a_\gamma))=(\mathrm{id}-\n E_{\Lambda_k(
\gamma)})(X_{\rm out})$, and
\[
\bigl\|\alpha_\bz(a_\gamma)-\n E_{\Lambda_k(\gamma)}(\alpha_\bz(a_\gamma))\bigr\|\;\leqslant\;2\,\norm{
X_{\rm out}}\,.
\]
since $\n E_{\Lambda_k(\gamma)}$ is a norm--one projection.
By the {\rm CAR} isometry $\norm{a(f)}=\norm{f}$ and the boundedness of the synthesis operator $D:
\ell^2(\Gamma)\to\HH$ by the upper frame bound $B$ \cite[Thm.~3.2.3]{Christensen-frames},
\[
\norm{X_{\rm out}}^2=\left\|\sum_{\gamma'\notin\Lambda_k(\gamma)}[t_{\bz}]_{\gamma',\gamma}\,\chi_{
\gamma'}\right\|^2\;\leqslant\;B\sum_{\gamma'\notin\Lambda_k(\gamma)}\bigl|[t_{\bz}]_{\gamma',\gamma}
\bigr|^2\,.
\]
For $k>\kappa_r$ every $\gamma'\notin\Lambda_k(\gamma)$ satisfies $d(\gamma,\gamma')>k>\kappa_r$, so
the positive part in \eqref{eq:Tz-banded} is active and
\[
\sum_{\gamma'\notin\Lambda_k(\gamma)}\bigl|[t_{\bz}]_{\gamma',\gamma}\bigr|^2\;\leqslant\;D_\ast^2\,
\expo{2\omega_\ast\kappa_r}\sum_{\gamma'\notin\Lambda_k(\gamma)}\expo{-2\omega_\ast d(\gamma,\gamma')}\,.
\]
Grouping over the shells $\Sigma_j:=\{\gamma'\,|\,d(\gamma,\gamma')\in(j,j+1]\}$, whose volume is bounded by $|\Sigma
_j|\leqslant\kappa_B(2+j)^3$ by {\eqref{eq:geo_gamma2} with $R=j+2$}, and using $d(\gamma,\gamma')>j$ on $\Sigma_j$,
\[
\sum_{\gamma'\notin\Lambda_k(\gamma)}\expo{-2\omega_\ast d(\gamma,\gamma')}\;\leqslant\;\kappa_B\sum_{j
\geqslant k}(2+j)^3\expo{-2\omega_\ast j}\,.
\]
Writing $j=k+m$ ($m\geqslant0$) and $2+j=2+k+m\leqslant(2+k)(1+m)$,
\[
\sum_{j\geqslant k}(2+j)^3\expo{-2\omega_\ast j}\;\leqslant\;(2+k)^3\expo{-2\omega_\ast k}\underbrace{\sum_{m
\geqslant0}(1+m)^3\expo{-2\omega_\ast m}}_{=:R_\ast}\,,
\]
the series $R_\ast<\infty$ converging because the exponential dominates the polynomial shell count. Since
$(2+k)^3\leqslant8(1+k)^3$, one gets
\[
\sum_{\gamma'\notin\Lambda_k(\gamma)}\bigl|[t_{\bz}]_{\gamma',\gamma}\bigr|^2\;\leqslant\;I_\ast\,(1+k)^3
\,\expo{-2\omega_\ast(k-\kappa_r)}\,,\qquad I_\ast:=8\,D_\ast^2\,\kappa_B\,R_\ast\,,
\]
with $I_\ast>0$ level--independent (the level enters only through the offset $\kappa_r$). Therefore, for $k>\kappa_r$,
\[\bigl\|\alpha_\bz(a_\gamma)-\n E_{\Lambda_k(\gamma)}(\alpha_\bz(a_\gamma))\bigr\|\;\leqslant\;2\sqrt{I
_\ast B}\,(1+k)^{\frac{3}{2}}\,\expo{-\omega_\ast(k-\kappa_r)}\,.
\]
For the complementary range $k\leqslant\kappa_r$ the same bound holds \emph{a fortiori}. Indeed by the
norm--one projection property, $\bigl\|\alpha_\bz(a_\gamma)-\n E_{\Lambda_k(\gamma)}(\alpha_\bz(a_
\gamma))\bigr\|\leqslant2\norm{\alpha_\bz(a_\gamma)}=2\norm{a(t_{\bz}\chi_\gamma)}=2$, while
$\expo{-\omega_\ast(k-\kappa_r)}\geqslant1$ and $(1+k)^{\sfrac{3}{2}}\geqslant1$ there. Therefore, by setting $I_\ast':=\max\{2,
{2\sqrt{I_\ast B}}\}$  (still level--independent) we conclude that for all $k\in\N_0$
\begin{equation}\label{eq:smooth-all-k}
\bigl\|\alpha_\bz(a_\gamma)-\n E_{\Lambda_k(\gamma)}(\alpha_\bz(a_\gamma))\bigr\|\;\leqslant\;I_\ast'
\,(1+k)^{\frac{3}{2}}\,\expo{-\omega_\ast(k-\kappa_r)}\,.
\end{equation}
Since $\kappa_r=O(\sqrt r)$ is a fixed finite offset (fixed a priori by $\gamma$), the right--hand side decays faster than every
inverse power of $k$. Multiplying \eqref{eq:smooth-all-k} by $(1+k)^p$ and taking $\sup_{k}$ gives
$\norm{\alpha_\bz(a_\gamma)}_{p,\gamma}<\infty$ for every $p\in\N_0$. by the equivalence of the
 norms anchored at different points \textup{(}Lemma~\ref{lemm:eq_npr_p}\textup{)} this holds at
every anchor, so $\alpha_\bz(a_\gamma)\in\rr A_p$ for all $p\in\N_0$, \ie $\alpha_\bz(a_\gamma)\in\rr A_
\infty$. As $\alpha_\bz$ is a $\ast$--automorphism and each $\rr A_p$ is closed under the involution,
$\alpha_\bz(a_\gamma^{*})=\alpha_\bz(a_\gamma)^{*}\in\rr A_\infty$ as well.
Recall that the  $p$-norms are submultiplicative up to a
constant (Theorem \ref{lem:loc-in-Ainfty}). 
Now let $A\in\rr A_{\mathrm{loc}}$. Then $A\in\rr A_\Lambda$ for some finite $\Lambda$, hence $A$ is a
noncommutative polynomial in the generators $\{a_\gamma^\sharp\,|\,\gamma\in\Lambda\}$. As $\alpha_\bz$ is a $\ast$--homomorphism it commutes with sums, products and
the involution, so $\alpha_\bz(A)$ is the \emph{same} polynomial in the transported generators
$\alpha_\bz(a_\gamma^{\#})$, each of which lies in $\rr A_\infty$ by the argument above. Iterating
the submultiplicative inequality, every monomial of degree $n$ satisfies
\[
\bigl\|\alpha_\bz(a_{\gamma_1}^{\#})\cdots\alpha_\bz(a_{\gamma_n}^{\#})\bigr\|_{p}\;\leqslant
\;5^{\,n-1}\prod_{i=1}^{n}\norm{\alpha_\bz(a_{\gamma_i}^{\#})}_{p}\;<\;\infty\,,
\]
each factor having finite $p$--norm, for all $p\in\N_0$. Summing the finitely many monomials preserves finiteness of
every $p$--norm, whence $\alpha_\bz(A)\in\bigcap_p\rr A_p=\rr A_\infty$. Therefore $\alpha_\bz(\rr
A_{\mathrm{loc}})\subseteq\rr A_\infty$, which is the assertion.
\end{proof} 
 
\smallskip

We need the following technical  result.

\smallskip

\begin{lemma}\label{lem:mono-trasl}
Let $\Lambda\in\s P_f(\Gamma)$, $\gamma_0\in\Gamma$ and $p\in\N_0$. Let $L=a^\sharp(\varphi_1)\cdots a^\sharp(\varphi_n)$ with $\varphi_1,\dots,\varphi_n\in\HH_\Lambda$ unit vectors (not necessarily distinct) and $n\leqslant2|\Lambda|$. Then, for every $\bz\in\R^2$ with $|\bz|<\ell_0$ one has 
\[
\norm{\alpha_\bz(L)}_{p,\gamma_0}\;\leqslant\;K_p(\Lambda,\gamma_0)\,,\qquad
\]
where
\[
K_p\;:=\;1+\max\left\{2(1+q)^p\,,\;4|\Lambda|\sup_{s\in\N_0}(1+q+s)^p\,\Phi_\Lambda(s)\right\}\;<\;\infty\,,
\]
with $q:=\lceil\max_{\gamma\in\Lambda}d(\gamma,\gamma_0)\rceil$,
\[
\Phi_\Lambda(s)\;:=\;\sqrt{\frac{|\Lambda|}{m(\Lambda)}}\;I_\ast'(1+s)^{\frac32}\expo{-\omega_\ast(s-\kappa_\Lambda)}\;,
\] 
where the constants $I_\ast'$ and $\omega_\ast$ are as in \eqref{eq:smooth-all-k},
$\kappa_\Lambda:=4\sqrt{r_\Lambda}+\sqrt2(\sfrac{\ell_0}{\lB})$ and $r_\Lambda$ is the highest Landau level among the sites of $\Lambda$. Finally $m(\Lambda)$ is the smallest eigenvalue of the Gram matrix of $\{\chi_\gamma\}_{\gamma\in\Lambda}$ as in Lemma~\ref{lem:delta-gram}.
The bound does not depend on the degree $n$ nor on $\bz$.
\end{lemma}
\begin{proof}
Write $\n E_k:=\n E_{\Lambda_k(\gamma_0)}$, $\pi_k$ for the orthogonal projection onto $\HH_{\Lambda_k(\gamma_0)}$, and $b_i:=\alpha_\bz(a^\sharp(\varphi_i))=a^\sharp(t_\bz\varphi_i)$ so that $\alpha_\bz(L)=b_1\cdots b_n$. Since
$\norm{b_i}=1$ it follows that $\norm{\alpha_\bz(L)}\leqslant1$. Fix $k\in\N_0$.
Let us start with the case $k<q$.
 Since $\n E_k$ is a norm--one projection, $\norm{(\mathrm{id}-\n E_k)\alpha_\bz(L)}(1+k)^p\leqslant2(1+q)^p$. Now the case $k\geqslant q$.
Every $\gamma\in\Lambda\subseteq\Lambda_q(\gamma_0)$, hence, by the triangle inequality for $d$,
$\Lambda_{k-q}(\gamma)\subseteq\Lambda_k(\gamma_0)$ and so $\HH_{\Lambda_{k-q}(\gamma)}\subseteq\HH_{\Lambda_k(\gamma_0)}$. Let $g_i:=\n E_k(b_i)$. By Lemma~\ref{lem:E-gen}, $\norm{b_i-g_i}=\|(\unit-\pi_k)t_\bz\varphi_i\|$. Expand $\varphi_i=\sum_{\gamma\in\Lambda}c_\gamma\chi_\gamma$. Since $\|\varphi_i\|=1$ one has $\langle c,G_\Lambda c\rangle=1$.
Therefore,  $\norm c_1\leqslant\sqrt{|\Lambda|}\norm c_2\leqslant\sqrt{|\Lambda|/m(\Lambda)}$
as in the proof of Lemma~\ref{lem:delta-gram}.
 Using $\|(\unit-\pi_k)v\|\leqslant\|(\unit-\pi_{\Lambda_{k-q}(\gamma)})v\|$ and again Lemma~\ref{lem:E-gen},
\[
\begin{aligned}
\norm{b_i-g_i}\;&\leqslant\;\sum_{\gamma\in\Lambda}|c_\gamma|\,\bigl\|\bigl(\mathrm{id}-\n E_{\Lambda_{k-q}(\gamma)}\bigr)\alpha_\bz(a_\gamma)\bigr\|\\
&\leqslant\;\norm c_1\,I_\ast'(1+k-q)^{\frac32}\expo{-\omega_\ast(k-q-\kappa_\Lambda)}\;\leqslant\;\Phi_\Lambda(k-q)\,,
\end{aligned}
\]
where we used \eqref{eq:smooth-all-k} {centered} at $\gamma$ with radius $k-q$ and $\kappa_{r_\gamma}\leqslant\kappa_\Lambda$.
Now $\prod_ig_i\in\rr A_{\Lambda_k(\gamma_0)}$ is fixed by $\n E_k$, and the telescoping identity
$\prod_{i}b_i-\prod_ig_i=\sum_{i=1}^n\bigl(\prod_{l<i}g_l\bigr)(b_i-g_i)\bigl(\prod_{l>i}b_l\bigr)$ with $\norm{g_l}\leqslant1$, $\norm{b_l}=1$ gives
\[
\begin{aligned}
\norm{(\mathrm{id}-\n E_k)\alpha_\bz(L)}\;&=\;\Bigl\|(\mathrm{id}-\n E_k)\Bigl(\prod_ib_i-\prod_ig_i\Bigr)\Bigr\|\\
&\leqslant\;2\sum_{i=1}^n\norm{b_i-g_i}\;\leqslant\;4|\Lambda|\,\Phi_\Lambda(k-q)\,.
\end{aligned}
\]
Writing $k=q+s$, multiplying by $(1+k)^p$, taking the supremum over $k\in\N_0$ and adding $\norm{\alpha_\bz(L)}\leqslant1$ yields the claim.
The supremum over $s$ is finite because $\Phi_\Lambda$ decays exponentially. All constants are independent of $\bz$ for $|\bz|<\ell_0$, since $I_\ast',\omega_\ast$ are.
\end{proof}

\smallskip

\begin{proof}[Proof of Lemma \ref{cor:mono-improve-full}]
Let $n:=|\Lambda|$ and let $e_1,\dots,e_n$ be an orthonormal basis of $\HH_\Lambda$ (recall $\dim\HH_\Lambda=|\Lambda|$ by Lemma~\ref{lem:findim}). The operators $a(e_j)$ satisfy the exact {\rm CAR}, so $\rr A_\Lambda=\rr A(\HH_\Lambda)\simeq M_{2^n}(\C)$ carries the tensor structure of $n$ independent fermionic modes, and the trace factorizes over the modes. For \emph{orthonormal} modes one can therefore use the argument of the single--site dual basis: the $4^n$ ordered monomials
$N_\mu:=\prod_{j=1}^nE^{\mu_j}_j$, $\mu\in\{0,1,2,3\}^n$, with $E^0_j=\unit$, $E^1_j=a(e_j)$, $E^2_j=a^*(e_j)$, $E^3_j=a^*(e_j)a(e_j)$, form a basis of $\rr A_\Lambda$, and the dual basis $N'_\mu:=\prod_jF^{\mu_j}_j$ with $F^0_j=\unit-a^*(e_j)a(e_j)$, $F^1_j=a(e_j)$, $F^2_j=a^*(e_j)$, $F^3_j=-\unit+2a^*(e_j)a(e_j)$ satisfies $\Tr(N_\mu'^{\,*}N_\nu)=\delta_{\mu\nu}$. Since $\|F^i_j\|_{\rm tr}=1,1,1,2$ (trace norm), writing $A=\sum_\mu\alpha_\mu N_\mu$ one gets $\alpha_\mu=\Tr(N_\mu'^{\,*}A)$, $|\alpha_\mu|\leqslant\|N'_\mu\|_{\rm tr}\norm A$ and
\[
\sum_\mu|\alpha_\mu|\;\leqslant\;\norm A\prod_{j=1}^n\Bigl(\sum_{i=0}^3\|F^i_j\|_{\rm tr}\Bigr)\;=\;5^{|\Lambda|}\norm A\,.
\]
 Each $N_\mu$ is a product of at most $2|\Lambda|$ factors $a^\sharp(e_j)$ with $e_j\in\HH_\Lambda$ unit vectors. Therefore Lemma~\ref{lem:mono-trasl} gives $\norm{\alpha_\bz(N_\mu)}_{p,\gamma_0}\leqslant K_p(\Lambda,\gamma_0)$ for every $\mu$ and every $|\bz|<\ell_0$. By linearity of $\alpha_\bz$ and the triangle inequality for $\norm{\cdot}_{p,\gamma_0}$,
\[
\norm{\alpha_\bz(A)}_{p,\gamma_0}\;\leqslant\;\sum_{\mu}|\alpha_\mu|\,\norm{\alpha_\bz(N_\mu)}_{p,\gamma_0}\;\leqslant\;5^{|\Lambda|}K_p(\Lambda,\gamma_0)\,\norm A\,,
\]
which is the claim with $C_p(\Lambda,\gamma_0):=5^{|\Lambda|}K_p(\Lambda,\gamma_0)$.
\end{proof}

\smallskip

\begin{proof}[Proof of proposition \ref{prop:cont-asympt-ab-local}]
Let $A\in\rr A_{\Lambda_1}^{\pm}$ and $B\in\rr A_{\Lambda_2}^{\pm}$ be homogeneous elements of parities $d_A,d_B$, respectively, and with
$\Lambda_1,\Lambda_2\in\s P_f(\Gamma)$. Fix $\gamma_0\in\Lambda_1$\revtxtmk{.}
Write $\bz=\bz_L+\bz_c$, and in turn $\alpha_\bz=\vartheta_{\theta}\circ\alpha_{\bz_L}\circ\alpha_{\bz_c}$ as in \eqref{eq:dec_trasl},
 with the short notation
$\theta\equiv\theta(\bz):=\sigma_B(\bz_L,\bz_c)$. Set $A':=\alpha_{\bz_c}(A)$. By
Lemma~\ref{lemma_trasl_loc}, $A'\in\rr A_\infty$ and by Remark~\ref{rk:B-A-parity}, $A'$ is homogeneous of
parity $d_A$ (Bogoliubov automorphisms preserve the parity).
 Since $\vartheta_\theta$ is a $\ast$--automorphism
preserving parity, for homogeneous $X,Y$ one has
$\vartheta_\theta([X,Y]_{\rm gr})=[\vartheta_\theta(X),\vartheta_\theta(Y)]_{\rm gr}$
By isometry of
$\vartheta_\theta$,
\begin{equation}\label{eq:z->z_l}
\begin{aligned}
\bigl\|[\alpha_\bz(A),B]_{\rm gr}\bigr\|\;=\;\bigl\|\vartheta_\theta\bigl([\alpha_{\bz_L}(A'),
\vartheta_{-\theta}(B)]_{\rm gr}\bigr)\bigr\|\;=\;\bigl\|[\alpha_{\bz_L}(A'),B_\theta]_{\rm gr}\bigr\|\;,
\end{aligned}
\end{equation}
with $B_\theta:=\vartheta_{-\theta}(B)$.
Using the argument in the proof of Proposition~\ref{prop:core-general-0} one has that  $\vartheta_{-\theta}(\rr A_{\Lambda_2})=\rr
A_{\Lambda_2}$ for every $\theta\in\R$. Therefore, $B_\theta\in\rr A_{\Lambda_2}$ is homogeneous of parity $d_B$
with $\norm{B_\theta}=\norm B$, {independently of $\theta$} (hence of $\bz$).
By
{Lemma}~\ref{cor:mono-improve-full}, $\norm{A'}_{p,\gamma_0}\leqslant C_p(\Lambda_1,\gamma_0)\norm A$
for every $\bz_c$ with $|\bz_c|<\ell_0$.
For $k\in\N_0$ (the value of $k$ will be fixed later) set $A'_k:=\n E_{\Lambda_k(\gamma_0)}(A')\in\rr
A_{\Lambda_k(\gamma_0)}$. {By Lemma~\ref{lem:dyn-commutes} applied to the parity automorphism $\vartheta_\pi$ (which preserves $\omega_{\rm tr}$ and each $\rr A_\Lambda$), one infers that}
$A'_k$ is homogeneous of parity
$d_A$. Moreover
\begin{equation}\label{eq:wewejj}
\norm{A'-A'_k}\;\leqslant\;(1+k)^{-p}\,\norm{A'}_{p,\gamma_0}\;\leqslant\;(1+k)^{-p}\,C_p(\Lambda_1,
\gamma_0)\,\norm A\;,
\end{equation}
where the first inequality {is the definition \eqref{eq:norm_p} of the $p$-norm, and the second one follows from Lemma~\ref{cor:mono-improve-full}}.  
This bound is
uniform for $|\bz_c|<\ell_0$.
By bilinearity of the graded commutator 
\begin{equation}\label{eq:wewe}
\begin{aligned}
\bigl\|[\alpha_{\bz_L}(A'),B_\theta]_{\rm gr}\bigr\|\;&\leqslant\;\bigl\|[\alpha_{\bz_L}(A'_k),B_\theta]_{
\rm gr}\bigr\|\;+\;2\,\norm{\alpha_{\bz_L}(A'-A'_k)}\,\norm{B_\theta}\\
&=\;\bigl\|[\alpha_{\bz_L}(A'_k),
B_\theta]_{\rm gr}\bigr\|\;+\;2\,\norm{A'-A'_k}\,\norm{B}\;,
\end{aligned}
\end{equation}
where in the first inequality we used that both $A'-A'_k$ and $A'_k$ have parity $d_A$, so no sign ambiguity arises, and in the last equality
 that $\alpha_{\bz_L}$ is isometric.
Since $A'_k\in\rr
A_{\Lambda_k(\gamma_0)}$ and $B_\theta\in\rr A_{\Lambda_2}$ are homogeneous and local, and
$\alpha_{\bz_L}(\rr A_{\Lambda_k(\gamma_0)})=\rr A_{\Lambda_k(\gamma_0)+\gamma_{\bz_L}}$ with
$|\Lambda_k(\gamma_0)+\gamma_{\bz_L}|=|\Lambda_k(\gamma_0)|$ as in \eqref{eq:lattice-cov-main}, Proposition
\ref{lem:almost-comm-gr-II} applies exactly as in the proof of Lemma~\ref{cor:asympt-ab-lattice} giving
\begin{equation}\label{eq:asi-ab-ines2}
\bigl\|[\alpha_{\bz_L}(A'_k),B_\theta]_{\rm gr}\bigr\|\;\leqslant\;C_\lambda\bigl(\Lambda_k(\gamma_0),
\Lambda_2\bigr)\,\norm{A'_k}\,\norm{B_\theta}\,\expo{-\lambda\,d(\Lambda_k(\gamma_0)+\gamma_{\bz_L},
\Lambda_2)}\;,
\end{equation}
for all $\lambda>0$.
Crucially, the constant $C_\lambda(\Lambda_k(\gamma_0),\Lambda_2)$ and the distance
$d(\Lambda_k(\gamma_0)+\gamma_{\bz_L},\Lambda_2)$ depend only on $\Lambda_k(\gamma_0)$, $\Lambda_2$ and
$\bz_L$, {never on $\theta$}. Moreover $\norm{B_\theta}=\norm B$. Hence the right-hand side of \eqref{eq:asi-ab-ines2}
is manifestly independent of $\theta(\bz)$. Since $d(\Lambda_k(\gamma_0)+\gamma_{\bz_L},\Lambda_2)
\to\infty$ as $|\bz_L|\to\infty$ (same estimate as in the proof of Lemma~\ref{cor:asympt-ab-lattice}), one obtains
\begin{equation}\label{eq:asi-ab-ines3}
\bigl\|[\alpha_{\bz_L}(A'_k),B_\theta]_{\rm gr}\bigr\|\;\xrightarrow[\ |\bz_L|\to\infty\ ]{}\;0\;,
\end{equation}
uniformly over $\theta\in\R$, and in turn over $\bz_c$\revtxtmk{.}
Let $\varepsilon>0$. Choose $p:=1$
and then $k\in\N_0$ such that $2(1+k)^{-p}C_p(\Lambda_1,\gamma_0)\norm A\norm B<\sfrac{\varepsilon}{2}$. This
choice of $k$ fixes $\Lambda_k(\gamma_0)$ once and for all. 
By \eqref{eq:asi-ab-ines3}, there is $R>0$ such that
$\bigl\|[\alpha_{\bz_L}(A'_k),B_\theta]_{\rm gr}\bigr\|<\sfrac{\varepsilon}{2}$ for every $\bz_L\in\s L$ with
$|\bz_L|>R$, uniformly in $\theta$. 
Combining with \eqref{eq:z->z_l}, \eqref{eq:wewejj} and \eqref{eq:wewe} 
\[
\bigl\|[\alpha_\bz(A),B]_{\rm gr}\bigr\|\;\leqslant\;\bigl\|[\alpha_{\bz_L}(A'_k),
B_\theta]_{\rm gr}\bigr\|\;+\;2\,(1+k)^{-p}\,C_p(\Lambda_1,
\gamma_0)\,\norm A\,\norm{B}\;<\;
\varepsilon\,
\]
for every $\bz=\bz_L+\bz_c$ with
$|\bz_L|>R$, or also for every 
  $|\bz| >R+\ell_0>R$.
As $\varepsilon>0$ was arbitrary, $\lim_{|\bz|\to\infty}\|[\alpha_\bz(A),B]_{\rm gr}\|=0$ for $A,B$ local
homogeneous.
\end{proof}

\smallskip

\begin{proof}[Proof of Theorem \ref{thm:asympt-ab}]
The proof of Proposition~\ref{cor:asympt-ab-full} carries over verbatim, with the single
replacement of the citation of Lemma~\ref{cor:asympt-ab-lattice} (decay along $\bl\in\s L$) by
Proposition~\ref{prop:cont-asympt-ab-local} (decay along $\bz\in\R^2$). 
\end{proof}

\renewcommand{\refname}{References}


\begin{thebibliography}{99}

\bibitem{LiebRobinson1972}
Lieb, E.H., Robinson, D.W.:
{\sl The finite group velocity of quantum spin systems}.
Commun. Math. Phys. {\bf 28}, 251--257 (1972)

\bibitem{Robinson1976}
Robinson, D.W.:
{\sl Properties of propagation of quantum spin systems}.
J. Austr. Math. Soc. B {\bf 19}, 387--399 (1976)

\bibitem{NachtergaeleOgataSims2006}
Nachtergaele, B., Ogata, Y., Sims, R.:
{\sl Propagation of correlations in quantum lattice systems}.
J. Stat. Phys. {\bf 124}, 1--13 (2006)

\bibitem{Lundberg1976}
Lundberg, L.-E.:
{\sl Quasi-free ``second quantization''}.
Commun. Math. Phys. {\bf 50}, 103--112 (1976)

\bibitem{Araki1971}
Araki, H.:
{\sl On quasifree states of CAR and Bogoliubov automorphisms}.
Publ. RIMS Kyoto Univ. {\bf 6}, 385--442 ({1970})

\bibitem{Streater1968}
Streater, R.F.:
{\sl On certain non-relativistic quantized fields}.
Commun. Math. Phys. {\bf 7}, 93--98 (1968)

\bibitem{StreaterWilde1970}
Streater, R.F., Wilde, I.F.:
{\sl The time evolution of quantized fields with bounded quasi-local interaction density}.
Commun. Math. Phys. {\bf 17}, 21--32 (1970)

\bibitem{GebertNachtergaeleReschkeSims2020}
Gebert, M., Nachtergaele, B., Reschke, J., Sims, R.:
{\sl Lieb--Robinson bounds and strongly continuous dynamics for a class of many-body fermion systems in $\mathbb{R}^d$}.
Ann. Henri Poincar\'e {\bf 21}, 3609--3637 (2020)

\bibitem{HinrichsLemmSiebert2023}
Hinrichs, B., Lemm, M., Siebert, O.:
{\sl On Lieb--Robinson bounds for a class of continuum fermions}.
Ann. Henri Poincar\'e {\bf 26}, 41--80 (2025)

\bibitem{BachmannDeNittis2024}
Bachmann, S., De Nittis, G.:
{\sl Lieb--Robinson bounds in the continuum via localized frames}.
Ann. Henri Poincar\'e {\bf 26}, 1--40 (2025)



\bibitem{Bargmann1971}
Bargmann, V., Butera, P., Girardello, L., Klauder, J.R.:
{\sl On the completeness of the coherent states}.
Rep. Math. Phys. {\bf 2}, 221--228 (1971)

\bibitem{Perelomov1971}
Perelomov, A.M.:
{\sl On the completeness of a system of coherent states}.
Theor. Math. Phys. {\bf 6}, 156--164 (1971)

\bibitem{BoonZak1978}
Boon, M., Zak, J.:
{\sl Discrete coherent states on the von Neumann lattice}.
Phys. Rev. B {\bf 18}, 6744--6751 (1978)

\bibitem{Daubechies1988}
Daubechies, I., Grossmann, A.:
{\sl Frames in the Bargmann space of entire functions}.
Commun. Pure Appl. Math. {\bf 41}, 151--164 (1988)

\bibitem{Cornean2019}
Cornean, H.D., Garde, H., St\o ttrup, B., S\o rensen, K.S.:
{\sl Magnetic pseudodifferential operators represented as generalized Hofstadter-like matrices}.
J. Pseudo-Differ. Oper. Appl. {\bf 10}, 307--336 (2019)

\bibitem{Cornean2024}
Cornean, H.D., Helffer, B., Purice, R.:
{\sl Matrix representation of magnetic pseudodifferential operators via tight Gabor frames}.
J. Fourier Anal. Appl. {\bf 30}, Art. 21 (2024)

\bibitem{Bratteli-Robinson-2}
Bratteli, O., Robinson, D.W.:
{\em Operator Algebras and Quantum Statistical Mechanics 2: Equilibrium States. Models in Quantum Statistical Mechanics}.
Springer, 2nd edition,  {1997}

\bibitem{Bratteli-Robinson-1}
Bratteli, O., Robinson, D.W.:
{\em Operator Algebras and Quantum Statistical Mechanics 1: $C^*$- and $W^*$-Algebras, Symmetry Groups, Decomposition of States}.
Springer, 2nd edition, 1987

\bibitem{Christensen-frames}
Christensen, O.:
{\em An Introduction to Frames and Riesz Bases}.
Birkh\"{a}user, 2003

\bibitem{Jaffard}
Jaffard, S.:
{\sl Propri\'et\'es des matrices ``bien localis\'ees'' pr\`es de leur diagonale et quelques applications}.
Ann. Inst. Henri Poincar\'e C,  {Anal. Non Lin\'eaire} {\bf 7}{, no.~5}, 461--476 (1990)

\bibitem{Bacry}
Bacry, H., Grossmann, A., Zak, J.:
{\sl Proof of completeness of lattice states in the $kq$ representation}.
Phys. Rev. B {\bf 12}, 1118--1120 (1975)

\bibitem{Daubechies}
Daubechies, I.:
{\sl The wavelet transform, time--frequency localization and signal analysis}.
IEEE Trans. Inf. Theory {\bf 36}, 961--1005 (1990)

\bibitem{Janssen}
Janssen, A.J.E.M.:
{\sl Signal analytic proofs of two basic results on lattice expansions}.
Appl. Comput. Harmon. Anal. {\bf 1}, 350--354 (1994)

\bibitem{ArakiMoriya2003}
Araki, H., Moriya, H.:
{\sl Equilibrium statistical mechanics of fermion lattice systems}.
Rev. Math. Phys. {\bf 15}, 93--198 (2003)

\bibitem{DerezinskiGerard}
Derezi\'{n}ski, J., G\'{e}rard, C.:
{\em Mathematics of Quantization and Quantum Fields}.
Cambridge University Press, 2013

\bibitem{ArakiHaag1967}
Araki, H., Haag, R.:
{\sl Collision cross sections in terms of local observables}.
Commun. Math. Phys. {\bf 4}(2), 77--91 (1967)

\bibitem{BachmannDybalskiNaaijkens2016}
Bachmann, S., Dybalski, W., Naaijkens, P.:
{\sl Lieb--Robinson bounds, Arveson spectrum and Haag--Ruelle scattering theory for gapped quantum spin systems}.
Ann. Henri Poincar\'e {\bf 17}(7), 1737--1791 (2016)

\bibitem{Ogata2021}
Ogata, Y.:
{\sl An {$H^3(G,\mathbb{T})$}-valued index of symmetry-protected topological phases with on-site finite group symmetry for two-dimensional quantum spin systems}.
Forum Math. Pi {\bf 9}, e13 (2021)

\bibitem{KapustinSopenko2022}
Kapustin, A., Sopenko, N.:
{\sl Local Noether theorem for quantum lattice systems and topological invariants of gapped states}.
J. Math. Phys. {\bf 63}, 091903 (2022)




\end{thebibliography}
\end{document}